\documentclass[sigconf]{acmart} 

\AtBeginDocument{%
  \providecommand\BibTeX{{%
    Bib\TeX}}}

\usepackage{enumerate}
\usepackage{algorithmic}
\usepackage{booktabs}
\usepackage[table]{xcolor}
\usepackage{amsthm}
\usepackage{mathdots}
\usepackage[toc,page,header]{appendix}
\usepackage{amsmath}
\usepackage{amsfonts}
\usepackage{tabularx}
\usepackage{mathtools}
\usepackage{balance} 
\usepackage{graphicx}
\usepackage{booktabs,tabularx}
\newcolumntype{Y}{>{\raggedright\arraybackslash}X}
\usepackage{amsmath,amsfonts}
\usepackage{algorithmic}
\usepackage[ruled, vlined, commentsnumbered, linesnumbered]{algorithm2e}
\usepackage{textcomp}
\usepackage{mathabx}
\usepackage{xcolor} 
\usepackage{microtype}      
\usepackage{lipsum}
\usepackage{enumitem}
\usepackage[normalsize]{subfigure}
\usepackage{wrapfig}
\usepackage{color}
\usepackage{amsthm}
\usepackage{amsthm}
\usepackage{mdframed}
\theoremstyle{plain}

\newmdenv[
  topline=true,
  bottomline=true,
  rightline=true,
  leftline=true,
  linecolor=black,
  linewidth=0.8pt,
  backgroundcolor=white,
  skipabove=10pt,
  skipbelow=10pt
]{boxedtheorem}

\definecolor{BrickRed}{rgb}{0.8, 0.25, 0.33}  
\definecolor{RoyalBlue}{rgb}{0.25, 0.41, 0.88}  

\newtheorem{relemma}{Lemma}

\newtheorem{theorem}{\bf{Theorem}}
\newtheorem{retheorem}{\bf{Theorem}}
\newtheorem{definition}{\bf{Definition}}
\newtheorem*{remark}{Remark}
\newtheorem{lemma}{\bf{Lemma}}

\newtheorem{reproposition}{\bf{Proposition}}
\newtheorem{proposition}{\bf{Proposition}}
\usepackage{amsthm}   
\usepackage{tcolorbox} 

\newcommand{\DEL}[1]{\iffalse #1 \fi}

\newcommand{\rd}{\color{BrickRed}}
\newcommand{\bl}{\color{RoyalBlue}}
\newcommand{\rev}{\color{black}}
\newcommand{\dcheck}{\color{black}}

\newcommand{\squishlist}{
\begin{list}{$\bullet$}
  { \setlength{\itemsep}{0pt}
     \setlength{\parsep}{0pt}
     \setlength{\topsep}{0pt}
     \setlength{\partopsep}{0pt}
     \setlength{\leftmargin}{0em}
     \setlength{\labelwidth}{0em}
     \setlength{\labelsep}{0.2em} } }

\setcopyright{acmlicensed}
\copyrightyear{2018}
\acmYear{2018}
\acmDOI{XXXXXXX.XXXXXXX}
\acmConference[Conference acronym 'XX]{Make sure to enter the correct
  conference title from your rights confirmation email}{June 03--05,
  2018}{Woodstock, NY}
\acmISBN{978-1-4503-XXXX-X/2018/06}

\title{A Graph-Based Framework for Extending Metric Differential Privacy Mechanisms}

\author{Ruiyao Liu}
\affiliation{
  \institution{University of North Texas}
  \city{Denton, Texas}
  \country{USA}}
\email{RuiyaoLiu@my.unt.edu}

\author{Chenxi Qiu}\authornote{Chenxi Qiu is the corresponding author.}
\affiliation{
  \institution{University of North Texas}
  \city{Denton, Texas}
  \country{USA}}
\email{chenxi.qiu@unt.edu}

\begin{abstract}
Metric differential privacy (mDP) is well suited to structured secret domains, but directly constructing utility-aware mechanisms over large or fine-grained domains is often computationally prohibitive. We study extension-based mDP design, where a mechanism is first specified on a finite set of seed records and then extended to a larger target domain. To our knowledge, this is the first work to systematically formulate extension as a general design paradigm for mDP rather than a method-specific construction. We present a graph-based extension framework, identify three requirements for correctness, \emph{local mDP constraints}, \emph{overlap consistency}, and \emph{successor-level mDP preservation}, and show that, under these conditions, the induced global mechanism is well defined and satisfies $\epsilon$-mDP on the target domain. We further instantiate the framework with a tree-based extension algorithm for multi-resolution grids, where multi-dimensional extension is realized through one-dimensional interpolation and dimension-wise composition. Experiments on road-map datasets demonstrate that our approach achieves a strong utility-scalability trade-off while preserving exact mDP guarantees.

\end{abstract}



\newcommand{\BibTeX}{\rm B\kern-.05em{\sc i\kern-.025em b}\kern-.08em\TeX}

\begin{document}




\maketitle 

\section{Introduction}
\label{sec:introduction}
\emph{Metric differential privacy (mDP)}~\cite{Andres-CCS2013} {\rev generalizes \emph{differential privacy (DP)}} to domains equipped with a meaningful distance metric. Rather than requiring a uniform privacy level for all pairs of secrets, mDP calibrates indistinguishability according to their distance, requiring stronger protection for nearby records and allowing greater distinguishability as records become farther apart. This makes mDP particularly suitable for structured domains such as geo-locations~\cite{Qiu-TMC2022}, text embeddings~\cite{feyisetan2021private}, and images~\cite{chen2021perceptual}, where a metric can capture application-specific similarity among secrets. {\rev Beyond mDP, federated analytics has also been studied for privacy-preserving image classification~\cite{Hou-TDSC2026}. In domains lacking such a metric, the benefits of mDP over standard DP formulations are less evident. Accordingly, this work focuses on structured domains in which similarity among secrets can be characterized by a meaningful metric. }

{\rev Existing mDP mechanisms can be broadly divided into distance-based~\cite{Andres-CCS2013, chatzikokolakis2015constructing} and optimization-based designs~\cite{Bordenabe-CCS2014,Qiu-TMC2022}. Distance-based mechanisms directly determine perturbation probabilities using a predefined function of the distance between the input and output. Although simple and broadly applicable, such mechanisms may not optimize application-specific utility because the privacy metric does not necessarily capture all relevant utility factors. For example, in fine-grained location privacy, spatial distance can define the mDP privacy metric, while utility may additionally depend on road-network connectivity, travel distance, task-assignment quality, or downstream service requirements~\cite{Qiu-TMC2022}.} 

To improve the privacy--utility trade-off, another line of {\rev work~\cite{Bordenabe-CCS2014,Qiu-TMC2022} has developed} optimization-based mechanisms that explicitly compute perturbation probabilities under mDP constraints. These methods typically formulate mechanism design as a constrained optimization problem, where the decision variables are the conditional reporting probabilities, the objective minimizes expected utility loss, and the constraints enforce mDP and valid probability distributions. This formulation can incorporate prior distributions, heterogeneous utility costs, and downstream task requirements rather than relying on a fixed distance-decay rule. {\rev However, optimization-based mechanisms introduce substantial computational overhead because they must optimize a reporting probability for every input--output pair and enforce privacy constraints between every pair of secrets for each possible output. As the domain grows, the number of variables increases quadratically, while the number of privacy constraints can increase cubically~\cite{Qiu-TMC2022}. Direct optimization over the full target domain therefore becomes computationally prohibitive for large or finely discretized domains~\cite{ImolaUAI2022}.}



{\rev Existing approaches improve the scalability of optimization-based mDP mechanisms through domain decomposition~\cite{qiu-IJCAI2024,Qiu-TMC2022} or localized mechanism construction~\cite{Qiu-PETS2025,Liu-CCS2025}. Although these approaches reduce the size of individual optimization problems, they still rely on centralized or overlapping local optimizations whose computational cost increases with domain resolution. Extension-based design~\cite{Qiu-USec2026} provides an alternative by optimizing perturbation distributions over a small seed domain and extending them to the remaining records through efficient local operations. However, existing methods are limited to specific one-stage interpolation schemes and lack general correctness conditions for multi-stage, branching, or hierarchical extensions. This gap motivates our graph-based framework. Section~\ref{sec:related_work} (related work) provides a detailed comparison of these approaches.}

\DEL{
\smallskip
\noindent \textbf{Works related to scalable mDP mechanism design.}
Several recent studies have sought to improve the scalability of optimization-based mDP mechanisms by exploiting domain structure. For a secret domain $\mathcal{X}$ and output domain $\mathcal{Y}$, a full perturbation mechanism requires a conditional probability matrix of size $|\mathcal{X}|\times|\mathcal{Y}|$; when $\mathcal{Y}=\mathcal{X}$, the number of decision variables grows quadratically with the domain size. Decomposition-based methods~\cite{qiu-IJCAI2024,qiu2025time} reduce this cost by grouping locations, leveraging graph structure, or separating independent feature blocks. However, as illustrated in Fig.~\ref{fig:scalability}(a), they still largely rely on centralized optimization over the secret domain. Consequently, their practical scalability is often limited to relatively small or moderately sized domains, e.g., on the order of $10^3$ records~\cite{qiu-IJCAI2024}, making them difficult to apply to fine-grained domains containing thousands to millions of records.

Another promising direction, illustrated in Fig.~\ref{fig:scalability}(b), is to move from monolithic optimization toward distributed optimization design. Recent works~\cite{Qiu-PETS2025,Liu-CCS2025} follow this idea by allowing each user to construct a perturbation mechanism only over a locally relevant subset of records, such as geo-locations that are plausible for that user, rather than over the full target domain. This reduces each optimization problem from the full domain size $|\mathcal{X}|$ to a user-specific local domain size $|\tilde{\mathcal{X}}|$, where $|\tilde{\mathcal{X}}|\ll|\mathcal{X}|$, enabling these methods to handle larger effective secret domains, e.g., roughly $2{,}500$--$6{,}000$ records. However, such local mechanisms must still be carefully coordinated: locally relevant domains may overlap across users, and the resulting mechanisms must remain compatible with the global metric structure to preserve the desired mDP guarantee.

These limitations motivate an \emph{extension-based paradigm} for scalable mDP mechanism design: as illustrated by Fig.~\ref{fig:scalability}(c), rather than optimizing a complete mechanism over the full target domain, one first optimizes perturbation distributions on a small set of seed records and then extends them to the remaining records through local operations. Initial work has shown that this strategy can substantially reduce optimization cost and scale to substantially larger domains, e.g., on the order of $40{,}000$--$130{,}000$ records~\cite{Qiu-USec2026}. However, existing methods are tailored to specific extension patterns and lack a general correctness principle for multi-stage, branching, or hierarchical extensions, where newly generated records may be further extended and multiple local paths must be assembled into one global mechanism.

The central challenge is to ensure that local extension remains globally valid. A correct procedure must preserve seed-record perturbation distributions, maintain consistency on overlapping local domains, and guarantee that the assembled mechanism still satisfies mDP after many extension steps. These requirements are especially subtle when local mechanisms are generated independently but must remain compatible with the global metric structure. In this paper, \emph{we focus on large structured secret domains, especially fine-grained spatial domains arising from high-resolution geographic discretization.}
}


\subsection*{Our Contributions}

\noindent\textbf{Contribution 1: A general framework for extending mDP
mechanisms.}
We develop a graph-based framework for constructing an mDP mechanism over a large domain from a mechanism optimized over a much smaller set of seed records. The framework is not tied to a specific extension algorithm; rather, it formalizes a broad class of constructions in which the seed mechanism is extended to a larger target domain through local operations. The graph provides a unified way to organize these operations and supports sequential, parallel, and hierarchical constructions.

Within this framework, we identify three conditions that, when satisfied jointly, are {\dcheck necessary and sufficient } to ensure that the local constructions provide the required privacy protection, agree on shared records, and preserve privacy across multiple extension stages. We formally prove in \textbf{Theorem~\ref{thm:global-mdp}} that any construction jointly satisfying these conditions forms a well-defined mDP mechanism over the full target domain. The framework therefore provides a reusable foundation for designing and verifying scalable extension-based mechanisms.



\noindent
\textbf{Contribution 2: A tree-based extension algorithm for fine-grained $\ell_p$-metric domains.} We instantiate the general framework with a tree-based algorithm for hierarchical grid domains. Starting from perturbation distributions optimized over a coarse set of seed records, the algorithm recursively constructs distributions for finer-grid records through efficient local interpolation. It thereby replaces costly full-domain optimization with a smaller seed-level optimization followed by inexpensive extension steps.

The interpolation and hierarchical extension procedures jointly satisfy the conditions established by the general framework. For two-dimensional domains, the algorithm performs interpolation along one coordinate at a time, with the resulting privacy guarantees composing under commonly used $\ell_p$ metrics. We prove in \textbf{Theorem~\ref{thm:subtree-preservation}} that the resulting mechanism satisfies mDP over the complete fine-grained target domain.

\noindent
\textbf{Contribution 3: Empirical validation on large location domains.}
We evaluate the proposed tree-based construction on location domains derived from the Rome, New York City, and London road networks~\cite{openstreetmap} across different domain resolutions, seed densities, and privacy budgets. Comparisons with representative mDP mechanisms show that our methods generally provide favorable utility and scalability while exhibiting no empirical mDP violations in the tested settings. Additional experiments demonstrate how seed density can be adjusted to balance utility and computational cost, supporting both the practical effectiveness and theoretical guarantees of the proposed construction.

The remainder of this paper is organized as follows. Section~\ref{sec:related_work} reviews related work on mDP mechanism design and its scalability. Section~\ref{sec:prelim} introduces the background on mDP and optimization-based mechanism design and formalizes the extension problem. Section~\ref{sec:framework} presents the general extension framework and its correctness conditions. Section~\ref{sec:tree} develops the tree-based instantiation for $\ell_p$-metric domains. Section~\ref{sec:experiments} presents the empirical evaluation, and Section~\ref{sec:conclusion} concludes the paper.

\begin{table*}[t]
\centering
{\rev 
\caption{Comparison of scalable optimization-based mDP mechanism designs.}
\label{tab:scalable_mdp_comparison}
\small 
\begin{tabular}{p{0.14\textwidth}p{0.17\textwidth}p{0.39\textwidth}p{0.21\textwidth}}
\toprule
Approach
& Optimized domain
& Construction strategy
& Number of secret records \\
\midrule
Decomposition-based
~\cite{qiu-IJCAI2024,qiu2025time}
& Original secret domain  $\mathcal{X}$
& Solve centralized optimization subproblems and combine their results
& Approximately $1{,}000$ \\
\hline
Localized
~\cite{Qiu-PETS2025,Liu-CCS2025}
& User-, region-, or anchor-specific domain
$\widetilde{\mathcal{X}}\subset\mathcal{X}$
& Construct multiple smaller local mechanisms and coordinate their outputs
& $2{,}500$--$6{,}000$ \\
\hline
Extension-based
~\cite{Qiu-USec2026}
& Seed domain
$\mathcal{X}_{\mathrm{seed}}\subset\mathcal{X}_{\mathrm{tar}}$
& Optimize a seed mechanism and extend it to the target domain
& $40{,}000$--$130{,}000$ \\
\bottomrule
\end{tabular}
\vspace{2pt}
\parbox{0.96\textwidth}{\footnotesize
The reported domain sizes are taken from the corresponding studies and are
intended to indicate the demonstrated scale of each approach, rather than a
controlled runtime comparison across different experimental environments.}}
\end{table*}

\section{Related Work}
\label{sec:related_work}
\noindent\textbf{Predefined distance-based mechanisms.}
Since the introduction of mDP~\cite{Chatzikokolakis-PETS2013}, extensive work has studied mechanisms that enforce distance-dependent privacy guarantees. A common approach specifies a noise distribution whose probabilities decay with the distance between the input and output. Representative examples include the planar Laplace mechanism~\cite{Andres-CCS2013} for $\ell_2$ geo-indistinguishability and the exponential mechanism~\cite{chatzikokolakis2015constructing}, which favors outputs closer to the true input. Their analytical forms enable efficient construction over continuous or fine-grained domains. However, this efficiency comes at the cost of flexibility, as fixed distributions generally cannot incorporate heterogeneous utility costs, prior distributions, or application-specific downstream objectives.

\textbf{Utility-aware optimization.} 
To overcome the limited flexibility of predefined noise distributions, optimization-based approaches directly optimize perturbation probabilities under mDP constraints~\cite{Bordenabe-CCS2014}. By incorporating data priors, heterogeneous utility costs, domain geometry, and application-specific requirements into the objective, these methods can achieve favorable privacy--utility trade-offs on relatively small domains~\cite{Bordenabe-CCS2014,Wang-ICDM2016,Yu-NDSS2017,Wang-WWW2017}.

However, their scalability is limited by the size of the resulting optimization problem. For a secret domain $\mathcal{X}$ and output domain $\mathcal{Y}$, a direct formulation contains $|\mathcal{X}||\mathcal{Y}|$ decision variables and up to $O(|\mathcal{X}|^2|\mathcal{Y}|)$ privacy constraints. When $\mathcal{Y}=\mathcal{X}$, these quantities become $O(|\mathcal{X}|^2)$ and $O(|\mathcal{X}|^3)$, respectively. Consequently, full-domain optimization becomes computationally prohibitive as the domain grows or is discretized more finely.


\DEL{
\begin{figure}[t]
    \centering
    \includegraphics[width=\columnwidth]{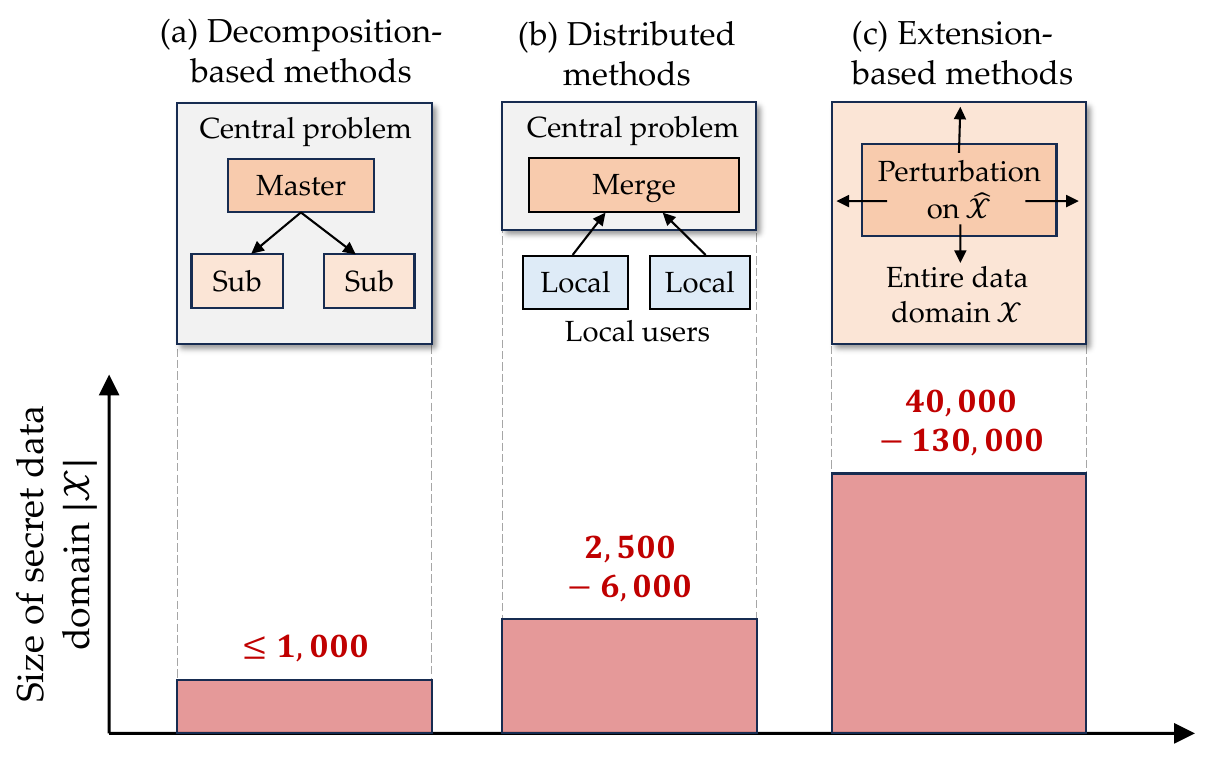}
    \caption{Comparison of approaches for improving the scalability of
    optimization-based mDP design: (a) centralized decomposition,
    (b) localized mechanism construction, and (c) extension-based
    construction.}
    \label{fig:scalability}
\end{figure}}

\noindent\textbf{Scalable optimization-based mechanism design.}
The scalability bottleneck of full-domain optimization has motivated methods that reduce the size, scope, or number of optimization problems. One direction combines predefined mechanisms with restricted optimization or post-processing. For example, Bayesian remapping~\cite{chatzikokolakis2017efficient} improves the utility of an existing mechanism through post-processing, while ConstOPT~\cite{ImolaUAI2022} restricts the optimization structure. Although these hybrid approaches reduce computation, they remain tied to predefined noise structures or may still incur substantial overhead on fine-grained domains. Other approaches exploit the structure of the secret domain to improve scalability. {\rev Table~\ref{tab:scalable_mdp_comparison} compares three representative categories in terms of their optimization scope, construction strategy, and reported evaluation scale.} \looseness = -1

\emph{Decomposition-based methods} partition the original problem by grouping locations, exploiting graph structure, or separating independent feature blocks~\cite{qiu-IJCAI2024,qiu2025time,Qiu-TMC2022,ImolaUAI2022}. By solving smaller subproblems, they retain the flexibility of utility-aware optimization while reducing the cost of each optimization. However, the number or size of these subproblems generally increases with domain resolution, limiting scalability on large, fine-grained domains.

\emph{Localized mechanism-design methods} construct smaller user- or anchor-specific mechanisms. Locally relevant geo-obfuscation~\cite{Qiu-PETS2025} optimizes over subsets of locations relevant to particular users or regions, whereas PAnDA~\cite{Liu-CCS2025} uses anchor-based approximation to reduce the optimization size. If $\widetilde{\mathcal{X}}$ denotes a local domain, these methods reduce each optimization from the full domain $\mathcal{X}$ to $\widetilde{\mathcal{X}}$, where $|\widetilde{\mathcal{X}}|\ll|\mathcal{X}|$. However, overlapping local domains and independently constructed mechanisms require careful coordination to maintain compatibility with the global metric and preserve mDP.

\emph{Extension-based methods} optimize a mechanism over a small seed domain and generate distributions for the remaining records through efficient local operations, thereby avoiding repeated or full-domain optimization. Existing work~\cite{Qiu-USec2026} uses coordinate-wise log-convex interpolation to extend a core mechanism to fine-grained records, supporting substantially larger target domains while retaining utility-aware optimization over the core domain.

However, the existing method is limited to a specific one-stage interpolation scheme and does not establish general correctness conditions for multi-stage, branching, or hierarchical extensions. In these settings, generated records may support subsequent extensions, and multiple local extension paths must be integrated into a globally valid mechanism. Our work addresses this gap through a graph-based framework that establishes {\rev necessary and sufficient  conditions} for preserving mDP under general multi-stage extensions and a tree-based algorithm that instantiates these conditions for hierarchical grid domains.

\DEL{
\section{Related Works}

\begin{figure}[t]
    \centering
    \includegraphics[width=0.49\textwidth]{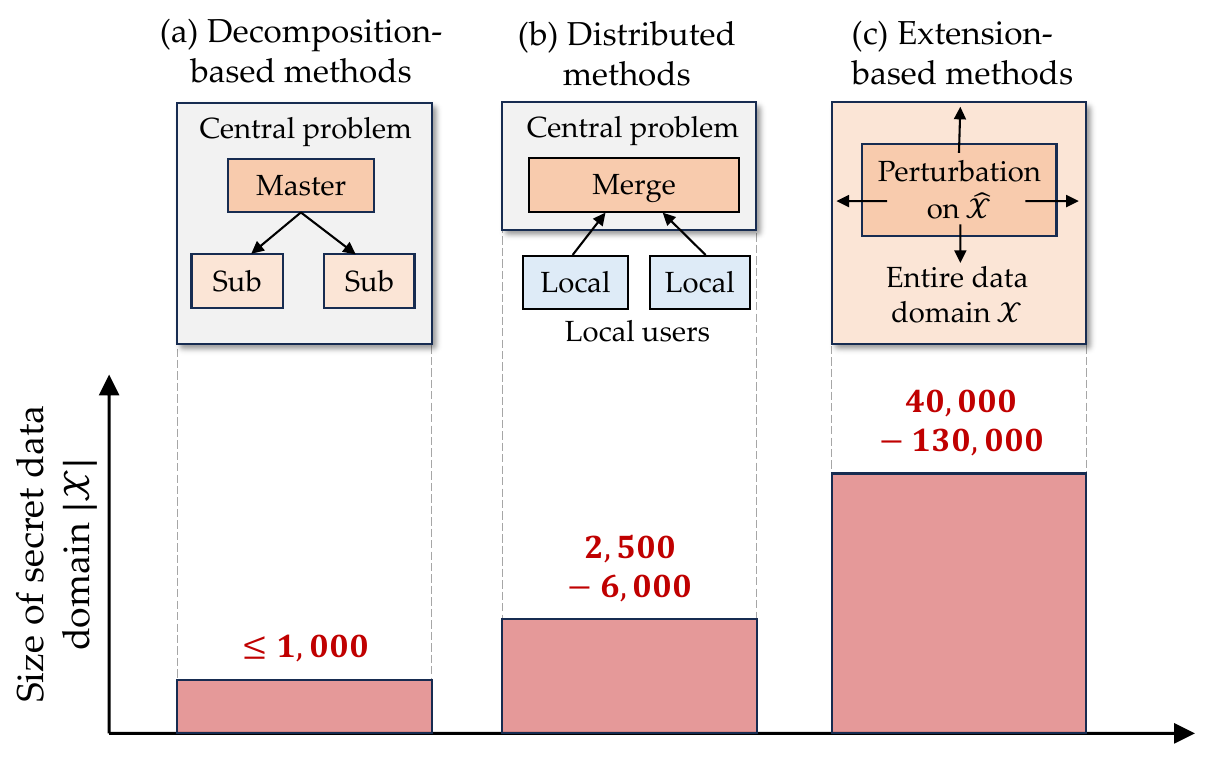} 
    \caption{Comparison of existing scalable mDP optimization approaches.}
    \label{fig:scalability}
\end{figure}

Since the introduction of mDP~\cite{Chatzikokolakis-PETS2013}, a large body of work has studied mechanisms for enforcing distance-based privacy guarantees. One line of work develops \emph{predefined noise mechanisms}, such as the planar Laplace mechanism~\cite{Andres-CCS2013} for $\ell_2$ geo-indistinguishability and the EM~\cite{chatzikokolakis2015constructing}, which favors outputs close to the true input. These mechanisms are efficient and naturally applicable to continuous or fine-grained domains, but their fixed noise distributions often yield suboptimal privacy--utility trade-offs because they do not adapt to context- or direction-dependent utility loss.

To improve utility, optimization-based methods formulate mechanism design as a linear program over a stochastic perturbation matrix~\cite{Bordenabe-CCS2014}. The LP chooses transition probabilities to minimize expected utility loss subject to mDP constraints, and can therefore achieve strong privacy--utility trade-offs on small domains~\cite{Bordenabe-CCS2014,Wang-ICDM2016,Yu-NDSS2017,Wang-WWW2017}. However, when $\mathcal{Y}=\mathcal{X}$, the perturbation matrix contains $O(|\mathcal{X}|^2)$ variables, making full-domain optimization computationally prohibitive for large, continuous, or fine-grained domains. Coarse discretization can reduce the problem size, but may weaken privacy by overestimating distances between neighboring records.

This scalability bottleneck has motivated several approaches that reduce the cost of optimization-based mDP design. Hybrid methods, such as Bayesian remapping~\cite{chatzikokolakis2017efficient} and ConstOPT~\cite{ImolaUAI2022}, combine fixed mechanisms with post-processing or constrained optimization to improve efficiency, but they remain tied to predefined noise structures or still incur substantial overhead on fine-grained domains. 

Decomposition-based methods further exploit domain structure, such as grouping locations, leveraging graph structure, or separating independent feature blocks ~\cite{qiu-IJCAI2024,Qiu-TMC2022,ImolaUAI2022}, to reduce the effective optimization size; however, they still largely rely on centralized optimization, limiting their applicability to very large domains. A related decentralized direction optimizes mechanisms over smaller local or anchor-based subproblems rather than the entire domain. For example, PAnDA~\cite{Liu-CCS2025} uses anchor-based approximation to reduce the optimization size, while locally relevant geo-obfuscation~\cite{Qiu-PETS2025} constructs user- or region-specific mechanisms over locally relevant location sets.

The work most closely related to ours is extension-based mDP design~\cite{Qiu-USec2026}, which optimizes a mechanism on a smaller core set $\hat{\mathcal{X}}$ and extends it to fine-grained records using coordinate-wise log-convex interpolation. While demonstrating the promise of extension, this method is limited to one-stage extension and a specific interpolation formula. In contrast, our work develops a general graph-based framework for multi-stage extension and instantiates it with tree-based interpolation strategies.}

\section{Preliminary}
\label{sec:prelim}
In this section, we review the preliminaries of mDP. 
For convenience, we summarize the main notation used throughout the paper in Table~\ref{tab:notation} in \textbf{Appendix~\ref{sec:notations}}.

\subsubsection*{\textbf{Threat model.}} We let $\mathcal{X}$ denote the input domain (i.e., the set of possible secret records) and $\mathcal{Y}$ denote the output domain (i.e., the set of values real records can be mapped to by perturbation mechanisms). A \emph{perturbation mechanism} is defined as a randomized mapping $\mathcal{M} : \mathcal{X} \rightarrow \mathcal{Y}$. 
To capture record similarity, we equip $\mathcal{X}$ with a distance function $d:\mathcal{X}\times\mathcal{X}\rightarrow \mathbb{R}_{\ge 0}$ that quantifies pairwise dissimilarity, e.g., geographic distance for locations, embedding-based distances for categorical values, or application-defined semantic distances. Let the random variable $X$ denote the true input record drawn from $\mathcal{X}$. 

A perturbation mechanism is defined as a randomized mapping
$\mathcal{M}:\mathcal{X}\rightarrow\mathcal{Y}$.
{\rev For the discrete output domain $\mathcal{Y}$, we write
$\mathcal{M}(\mathbf{y}\mid\mathbf{x})
\coloneqq
\Pr[\mathcal{M}(\mathbf{x})=\mathbf{y}]$
for the conditional probability mass of output $\mathbf{y}$ given input $\mathbf{x}$; accordingly, $\mathcal{M}(\cdot\mid\mathbf{x})$ denotes the output distribution induced by $\mathbf{x}$.}


{\rev We consider a local deployment of an mDP mechanism, analogous to the standard LDP setting~\cite{Wang-WWW2017}. The server is untrusted: it knows the perturbation mechanism $\mathcal{M}$, may have a prior distribution over the true data $X$, and may attempt to infer an individual's sensitive input from the received perturbed report. In particular, given an observed output $y$, the prior, and knowledge of $\mathcal{M}$, the server can compute a posterior belief over $X$ via Bayes' rule~\cite{Yu-NDSS2017}. This captures realistic deployments in which the data collector is incentivized to analyze collected records and therefore cannot be fully trusted from a privacy standpoint. Each user faithfully applies the fixed mechanism $\mathcal{M}$ locally before transmission, so the server observes only the perturbed report and does not control the client-side randomization. The server may perform arbitrary computations on the report or deviate from any prescribed downstream analysis; by the post-processing property, such actions alone cannot degrade the report's formal mDP guarantee. 
}

\subsubsection*{\textbf{Formal definition of mDP}} LDP requires data perturbation to maintain a uniform level of \emph{indistinguishability} between users' data~\cite{Duchi-FOCS2013}. mDP generalizes this concept by tying the privacy guarantee to a distance over the input space: nearby inputs should be harder to distinguish than far-apart ones. Originally proposed for location privacy~\cite{Andres-CCS2013}, mDP has been naturally extended to high-dimensional domains such as text embedding \cite{feyisetan2021private} and image data protection \cite{chen2021perceptual}. \looseness =-1

\begin{definition}[Lipschitz bound and continuity~\cite{Qiu-USec2026}]
Define a function $f:\mathcal X\to\mathbb R$. We say $f$ satisfies an \emph{$(\epsilon,d_p)$-Lipschitz bound} between $\mathbf{x},\mathbf{x}'\in\mathcal X$ if
\begin{equation} 
\label{eq:Lipschitzbound}
|f(\mathbf{x})-f(\mathbf{x}')| \le \epsilon\, d(\mathbf{x},\mathbf{x}').
\end{equation}
We say $f$ is \emph{$(\epsilon,d)$-Lipschitz continuous} if the bound in~\eqref{eq:Lipschitzbound} holds for all $\mathbf{x},\mathbf{x}'\in\mathcal X$.
\end{definition}


{\rev \begin{definition}[$(\epsilon,d)$-mDP~\cite{Andres-CCS2013}]
\label{def:metricDP}
Let $\mathcal{Y}$ be a measurable output space. A randomized mechanism $\mathcal{M}:\mathcal{X}\rightarrow\mathcal{Y}$ satisfies \emph{$(\epsilon,d)$-mDP} if, for all $\mathbf{x},\mathbf{x}'\in\mathcal{X}$ and every measurable set $\mathcal{S}\subseteq\mathcal{Y}$,
\begin{equation}
\label{eq:mDP}
\Pr\!\left[\mathcal{M}(\mathbf{x})\in\mathcal{S}\right]
\leq
\exp\!\left(\epsilon d(\mathbf{x},\mathbf{x}')\right)
\Pr\!\left[\mathcal{M}(\mathbf{x}')\in\mathcal{S}\right].
\end{equation}
\end{definition}
When $\mathcal{Y}$ is countable, Definition~\ref{def:metricDP} is equivalently expressed as
\begin{equation}
\label{eq:mDP_discrete}
\Pr[\mathcal{M}(\mathbf{x})=\mathbf{y}]
\leq
\exp\!\left(\epsilon d(\mathbf{x},\mathbf{x}')\right)
\Pr[\mathcal{M}(\mathbf{x}')=\mathbf{y}]
\end{equation}
for every $\mathbf{y}\in\mathcal{Y}$. When both probability masses are positive, this is equivalent to the log-Lipschitz formulation.}

{\rev In the remainder of this paper, we consider a finite discrete output domain $\mathcal{Y}$. We therefore use $\mathcal{M}(\mathbf{y}\mid\mathbf{x}) \coloneqq
\Pr[\mathcal{M}(\mathbf{x})=\mathbf{y}]$ to denote the conditional probability mass of output $\mathbf{y}$ given input $\mathbf{x}$, and use $\mathcal{M}(\cdot\mid\mathbf{x})$ to denote the corresponding output distribution.} For simplicity, we use $\mathcal{M}(\mathbf{x}) \stackrel{\epsilon}{\approx} \mathcal{M}(\mathbf{x}')$ to represent that Eq.~\eqref{eq:mDP} is satisfied. \looseness=-1

Intuitively, mDP ensures that small changes in the input $\mathbf{x}$ of the perturbation method $\mathcal{M}$ result in bounded changes in the distribution of the output $\mathcal{M}(\mathbf{x})$, thus providing privacy guarantees in the corresponding distance metric space. A lower $\epsilon$ signifies a tighter bound, meaning {\rev less information can be inferred about the secret input $\mathbf{x}$ from the released output $\mathcal{M}(\mathbf{x})$}. {\rev Additionally, to distinguish the privacy model from the form of the privacy guarantee, Table~\ref{tab:dp_comparison} in Appendix \ref{sec:DPvsmDP} compares conventional central DP, LDP, and mDP. }

{\rev Although we use metric terminology throughout, the general extension framework also applies to pseudometrics. If distinct secrets $\mathbf{x},\mathbf{x}'\in\mathcal{X}$ satisfy $d(\mathbf{x},\mathbf{x}')=0$, Definition~\ref{def:metricDP}, applied in both directions, requires $\Pr[\mathcal{M}(\mathbf{x})\in\mathcal{S}] = \Pr[\mathcal{M}(\mathbf{x}')\in\mathcal{S}]$ for every measurable $\mathcal{S}\subseteq\mathcal{Y}$. Thus, the seed mechanism and subsequent extensions must assign identical output distributions to secrets in the same zero-distance equivalence class. 
}

\DEL{
\begin{definition}[\textbf{Exponential Mechanism}~\cite{Chatzikokolakis-PoPETs2015}]
\label{def:ExpMech}
Given a privacy parameter $\epsilon>0$ and an input $\mathbf{x}\in\mathcal{X}$, the (metric) exponential mechanism samples an output $\mathbf{y}\in\mathcal{Y}$ with probability
\begin{equation}
\Pr[\mathcal{M}(\mathbf{x})=y]
=
\frac{\exp\left(-\tfrac{\epsilon}{2} d(x,y)\right)}
{\sum_{\mathbf{y}'\in\mathcal{Y}} \exp\left(-\tfrac{\epsilon}{2} d(x,\mathbf{y}')\right)}, 
\label{eq:metric-em}
\end{equation}
where $d(x,y)$ measures the distance between $x$ and $y$.
\end{definition}
In this paper, we use the exponential mechanism as our primary DPN method considering that both $\mathcal{X}$ and $\mathcal{Y}$ are discrete, and EM provides a canonical way to sample from an arbitrary discrete output space. However, our framework does not rely on EM per se: other DPN mechanisms (like Laplace \cite{Andres-CCS2013}) whose conditional distribution can be evaluated (or sampled) as a function of $d(x,y)$ can be plugged into our pipeline with minimal changes.

\begin{definition}
[\textbf{Laplace Mechanism}~\cite{andres2013geo}] 
\label{def:Laplace}
Given a privacy parameter $\epsilon>0$ and an input $\mathbf{x}\in\mathcal{X}$, the metric-based Laplace (distance-based) mechanism outputs $\mathbf{y}\in\mathcal{Y}$ according to
\begin{equation}
\Pr[\mathcal{M}(\mathbf{x})=y]
=
\frac{e^{-\epsilon d_{\mathcal{X}}(x,y)}}
{\int_{\mathbf{y}'\in\mathcal{Y}} e^{-\epsilon d_{\mathcal{X}}(x,\mathbf{y}')}d\mathbf{y}'}.
\label{eq:metric-laplace}
\end{equation}
(For discrete $\mathcal{Y}$, the integral is replaced by a sum.)
\end{definition}
\begin{definition}
[\textbf{Exponential Mechanism}~\cite{Chatzikokolakis-PoPETs2015}] 
\label{def:ExpMech}
Given a privacy parameter $\epsilon>0$ and an input $\textbf{x} \in \mathcal{X}$, the exponential mechanism selects an output $\textbf{y} \in \mathcal{Y}$ with probability:
\begin{equation}
\label{eq:exp_mech}
\Pr\left[\mathcal{M}(\mathbf{x}) = y\right] = \frac{e^{-\frac{1}{2}\epsilon d(x, y)}}{\sum_{\mathbf{y}' \in \mathcal{Y}} e^{-\frac{1}{2}\epsilon d(x, \mathbf{y}')}}.
\end{equation}
\end{definition}

\begin{definition}[\textbf{LP-based mechanism}~\cite{Bordenabe-CCS2014}]
\label{def:OptMech}
In LP-based (optimization-based) designs, since both the input and output domains $\mathcal{X}$ and $\mathcal{Y}$ are discrete, a perturbation mechanism $\mathcal{M}$ can be represented by a conditional probability matrix $\mathbf{Z}=\bigl\{z_{\mathbf{x},\mathbf{y}}\bigr\}_{\mathbf{x}\in\mathcal{X},\,\mathbf{y}\in\mathcal{Y}}$, where $z_{\mathbf{x},\mathbf{y}}=\Pr[\mathcal{M}(\mathbf{x})=y]$.
Given an input prior $\pi(\mathbf{x})$ and a utility-loss function $c(\mathbf{y}|\mathbf{x})$,
an optimal mechanism is obtained by solving
\begin{eqnarray}
\min_{\mathbf{Z}} && \textstyle \sum_{\mathbf{x}\in\mathcal{X}}\sum_{\mathbf{y}\in\mathcal{Y}} \pi(\mathbf{x}) c(\mathbf{y}|\mathbf{x}) z_{\mathbf{x},\mathbf{y}}  \\
\mathrm{s.t.} && z_{\mathbf{x},\mathbf{y}} \le e^{\epsilon d(\mathbf{x},\mathbf{x}')}  z_{\mathbf{x}',y},
\quad \forall \mathbf{x},\mathbf{x}'\in\mathcal{X},\ \forall \mathbf{y}\in\mathcal{Y}, \\
&& \textstyle \sum_{\mathbf{y}\in\mathcal{Y}} z_{\mathbf{x},\mathbf{y}} = 1,\quad \forall \mathbf{x}\in\mathcal{X}, \\
&& z_{\mathbf{x},\mathbf{y}} \ge 0,\quad \forall \mathbf{x}\in\mathcal{X},\ \forall \mathbf{y}\in\mathcal{Y}.
\end{eqnarray}
Here, the objective minimizes expected utility loss under the prior, and the constraints enforce the
mDP-induced indistinguishability conditions parameterized by $(\epsilon,d)$.
\end{definition}}

\section{Graph-Based Extension Framework}
\label{sec:framework}
In this section, we present a general graph-based framework for extending perturbation mechanisms from a finite seed set $\mathcal{X}_{\mathrm{seed}}$ to a larger target set $\mathcal{X}_{\mathrm{tar}}$. We first {\rev formulate the problem in \textbf{\S\ref{subsec:problem}}}, then define the extension graph in \textbf{\S\ref{subsec:graph}}, and finally present extension algorithms and {\rev necessary and sufficient  conditions} for correctness in \textbf{\S\ref{subsec:algorithms}}. This framework provides the foundation for the tree-based construction developed in the next section.

\noindent\textbf{Motivating example: why one-stage extension is insufficient.}
Consider a multi-resolution grid where perturbation distributions are specified only on coarse seed records, while the target domain is finer. A one-stage extension constructs all target distributions directly from the seed mechanism, either globally or through parallel local steps. However, the generated distributions are treated as final and cannot be further propagated, making such schemes unsuitable for recursive or hierarchical extensions across multiple resolutions.

To address this limitation, we introduce a multi-stage procedure in which mechanisms generated at one stage can serve as inputs to subsequent stages until $\mathcal{X}_{\mathrm{tar}}$ is covered. This supports sequential, parallel, and hierarchical extension within a unified framework. However, locally valid extensions do not automatically guarantee global mDP: different regions may not be jointly constrained, overlapping regions may assign inconsistent distributions, and earlier privacy relations may not be preserved. Additional consistency conditions are therefore required.

\noindent\textbf{Graph-based representation.}
To formalize mechanism propagation across multiple stages, we represent the extension procedure as a \emph{directed acyclic graph (DAG)}. Each node corresponds to a local record region, and each directed edge represents an extension step that constructs a mechanism on a successor region from one defined on its predecessor. Directed paths capture sequential extension, multiple successors capture parallel extension, and multiple graph levels capture recursive or hierarchical extension. This representation provides a systematic way to specify local conditions on extension steps, overlapping regions, and successive stages, and to characterize when they yield a globally valid mDP mechanism.

\DEL{Computing optimal perturbation distributions under mDP remains computationally expensive when the domain is large or continuous. A practical and increasingly adopted strategy is to optimize the perturbation distribution over a subset of representative points, often referred to as \emph{anchors}, and then derive the perturbation behavior of other points from these optimized anchors. While intuitive, this strategy raises two fundamental theoretical challenges:
\begin{enumerate}
    \item \textbf{Extension consistency.} How can an optimized perturbation distribution defined on a subset of points be extended to additional points without re-solving the global optimization problem, while still preserving mDP?
    \item \textbf{Hierarchical dependency.} Since anchor points will later serve as representatives for larger regions, how should their optimization objective account for the utility impact on all points that will be derived from them?
\end{enumerate}

In this paper, we formalize and address these challenges through a hierarchical extension framework for mDP mechanisms.

\smallskip
{\rd \noindent\textbf{Mechanism space.}
For any $\mathcal{X}'\subseteq\mathcal{X}$, let
\begin{equation}
\mathsf{Mech}(\mathcal{X}') \triangleq
\left\{\mathcal{M}:\mathcal{X}'\times\mathcal{Y}\to[0,1]\ \middle|\
\sum_{\mathbf{y}\in\mathcal{Y}}\mathcal{M}(\mathbf{y} \mid \mathbf{x})=1,\ \forall \mathbf{x}\in\mathcal{X}'\right\}
\label{eq:mech-space}
\end{equation}
denote the set of feasible perturbation mechanisms indexed by records $\mathbf{x}\in\mathcal{X}'$.}}

{\rev 
\subsection{Problem Formulation}
\label{subsec:problem}

Let $\mathcal{X}_{\mathrm{seed}}\subseteq\mathcal{X}_{\mathrm{tar}}$ denote,
respectively, the seed and target secret domains, and let $\mathcal{Y}$ denote
the output domain. Suppose that an $(\epsilon,d)$-mDP seed mechanism
$\mathcal{M}_{\mathrm{seed}}:\mathcal{X}_{\mathrm{seed}}\rightarrow
\mathcal{Y}$ has already been constructed. The goal of mechanism extension is
to construct a target mechanism
$\mathcal{M}_{\mathrm{tar}}:\mathcal{X}_{\mathrm{tar}}\rightarrow\mathcal{Y}$
that preserves the existing seed distributions while defining valid
distributions for all additional target records.

\noindent\textbf{Problem (Utility-aware mDP mechanism extension).}
Given
$(\mathcal{X}_{\mathrm{seed}},\mathcal{X}_{\mathrm{tar}},\mathcal{Y},d,
\epsilon,\mathcal{M}_{\mathrm{seed}})$, construct
$\mathcal{M}_{\mathrm{tar}}$ satisfying the following requirements:

\begin{enumerate}
    \item \emph{Seed preservation:} for every
    $\mathbf{x}\in\mathcal{X}_{\mathrm{seed}}$ and
    $\mathbf{y}\in\mathcal{Y}$,
    \[
    \mathcal{M}_{\mathrm{tar}}(\mathbf{y}\mid\mathbf{x})
    =
    \mathcal{M}_{\mathrm{seed}}(\mathbf{y}\mid\mathbf{x}).
    \]

    \item \emph{Valid extension:} for every
    $\mathbf{x}\in\mathcal{X}_{\mathrm{tar}}$,
    $\mathcal{M}_{\mathrm{tar}}(\cdot\mid\mathbf{x})$ is a valid probability
    distribution:
    \[
    \textstyle 
    \mathcal{M}_{\mathrm{tar}}(\mathbf{y}\mid\mathbf{x})\geq 0,
    \qquad
    \sum_{\mathbf{y}\in\mathcal{Y}}
    \mathcal{M}_{\mathrm{tar}}(\mathbf{y}\mid\mathbf{x})=1.
    \]

    \item \emph{Global mDP:} for every
    $\mathbf{x},\mathbf{x}'\in\mathcal{X}_{\mathrm{tar}}$ and
    $\mathbf{y}\in\mathcal{Y}$,
    \[
    \mathcal{M}_{\mathrm{tar}}(\mathbf{y}\mid\mathbf{x})
    \leq
    \exp\!\left(\epsilon d(\mathbf{x},\mathbf{x}')\right)
    \mathcal{M}_{\mathrm{tar}}(\mathbf{y}\mid\mathbf{x}').
    \]
\end{enumerate}

Privacy feasibility alone does not distinguish among potentially useful and
trivial extensions. Let
$c(\mathbf{x},\mathbf{y})$ denote the utility loss incurred when
$\mathbf{y}$ is reported for $\mathbf{x}$, and let
$\pi_{\mathrm{tar}}$ denote a prior distribution over the target domain. The
expected utility loss of an extension is
\[
\textstyle 
\operatorname{UL}(\mathcal{M}_{\mathrm{tar}})
=
\sum_{\mathbf{x}\in\mathcal{X}_{\mathrm{tar}}}
\pi_{\mathrm{tar}}(\mathbf{x})
\sum_{\mathbf{y}\in\mathcal{Y}}
\mathcal{M}_{\mathrm{tar}}(\mathbf{y}\mid\mathbf{x})
c(\mathbf{x},\mathbf{y}).
\]
A utility-aware extension seeks a mechanism satisfying the three mandatory
requirements while minimizing, or otherwise controlling, this loss.

This formulation separates \emph{privacy correctness} from \emph{utility optimization}. The first three requirements define a valid mDP extension, while the utility-loss objective distinguishes among feasible extensions according to their usefulness. The framework introduced next provides {\rev necessary and sufficient } conditions for constructing a correct extension, whereas the choice of seed mechanism and extension procedure determines the resulting utility.}

\subsection{Extension Graph}
\label{subsec:graph}
{\rev The formulation above specifies the requirements of a valid extension from $\mathcal{X}_{\mathrm{seed}}$ to $\mathcal{X}_{\mathrm{tar}}$, but not how the extension is constructed. Rather than determining all target distributions simultaneously, we construct the mechanism progressively through local operations, whose dependencies are represented by a directed graph. Its nodes correspond to local record sets, while its edges describe how existing distributions are extended to new records. This representation supports sequential, parallel, and hierarchical constructions while explicitly capturing dependencies and overlaps among extension operations.

Importantly, the extension graph does not replace or redefine the metric $d$. The metric specifies the privacy relationships among secrets, whereas the graph describes where and in what order their perturbation distributions are constructed. Separating privacy requirements from construction dependencies allows different extension operators to be analyzed within a common correctness framework. We formalize this dependency structure below.}

\begin{definition}[Extension graph]
\label{def:extension-dag}
Given a seed domain
$\mathcal{X}_{\mathrm{seed}}\subseteq\mathcal{X}_{\mathrm{tar}}$, an
\emph{extension graph} is a triple
$\mathcal{G}=(\mathcal{V},\mathcal{E},
\{\mathcal{X}_v\}_{v\in\mathcal{V}})$, where $\mathcal{V}$ is a finite set of
nodes, $\mathcal{E}\subseteq\mathcal{V}\times\mathcal{V}$ is a set of directed
edges such that $(\mathcal{V},\mathcal{E})$ is acyclic, and each node
$v\in\mathcal{V}$ is associated with a local record set
$\mathcal{X}_v\subseteq\mathcal{X}_{\mathrm{tar}}$. A directed edge
$(v,u)\in\mathcal{E}$ indicates that the construction at node $u$ depends on
the perturbation distributions available at node $v$.

For each node $v\in\mathcal{V}$, define
$\mathrm{Succ}(v)\coloneqq
\{u\in\mathcal{V}:(v,u)\in\mathcal{E}\}$ and
$\mathrm{Pred}(v)\coloneqq
\{u\in\mathcal{V}:(u,v)\in\mathcal{E}\}$. The nodes in
$\mathcal{V}_0\coloneqq
\{v\in\mathcal{V}:\mathrm{Pred}(v)=\varnothing\}$ are called
\emph{source nodes}, and their associated sets
$\{\mathcal{X}_v:v\in\mathcal{V}_0\}$ are called \emph{source regions}. The
extension graph satisfies: 
\begin{itemize}
    \item (i) \emph{seed initialization:} $\mathcal{X}_{\mathrm{seed}} =\bigcup_{v\in\mathcal{V}_0}\mathcal{X}_v$; 
    \item (ii) \emph{target coverage:}     $\mathcal{X}_{\mathrm{tar}} =\bigcup_{v\in\mathcal{V}}\mathcal{X}_v$; and 
    \item (iii) \emph{extension coverage:} for every node $v\in\mathcal{V}$ with $\mathrm{Succ}(v)\neq\varnothing$, $\mathcal{X}_v \subseteq\bigcup_{u\in\mathrm{Succ}(v)}\mathcal{X}_u$.
\end{itemize}
\end{definition}
{\rev
\begin{remark}[Single-successor vs.\ multi-successor extension] When constructing the extension graph, the expansion of a node region can be represented either by a \emph{single} successor or by \emph{multiple} successors.
Concretely, for a node $u\in\mathcal{V}$, one may introduce a single successor $v$ and perform a monolithic extension on a larger region $\mathcal{X}_v$, or alternatively introduce several successors $v_1,\dots,v_m$ whose regions jointly cover the extension target, e.g., $\mathcal{X}_u \subseteq \mathcal{X}_v$ or $\mathcal{X}_u \subseteq \bigcup_{i=1}^m \mathcal{X}_{v_i}$. Both representations describe the same conceptual step, extending from the perturbation distributions already specified on $u$ to additional records, and do not change the \emph{definition} of extension. However, they differ operationally in the next stage of the pipeline: multiple successors naturally correspond to extension subproblems that can be executed in parallel (one per successor region), whereas a single successor corresponds to performing the extension as one sequential/monolithic subproblem.
\end{remark}}


\DEL{
\begin{definition}[mDP extension problem]
\label{def:mdp-extension-problem}
Fix a secret domain $\mathcal{X}$, a seed set
$\mathcal{X}_{\mathrm{seed}}\subseteq\mathcal{X}_{\mathrm{tar}}\subseteq\mathcal{X}$,
and an extension graph
$\mathcal{G}=(\mathcal{V},\mathcal{E},\{\mathcal{X}_v\}_{v\in\mathcal{V}})$
with source nodes
$\mathcal{V}_0=\{v\in\mathcal{V}:\mathrm{Pred}(v)=\varnothing\}$.
Suppose a seed mechanism $\mathcal{M}_{\mathrm{seed}}$ is specified on
$\mathcal{X}_{\mathrm{seed}}$. The \emph{mDP extension problem} is to
initialize mechanisms on the source regions
$\{\mathcal{X}_v:v\in\mathcal{V}_0\}$, propagate them through
$\mathcal{G}$ via local extension operations, and construct an induced
mechanism $\mathcal{M}_{\mathrm{tar}}$ on $\mathcal{X}_{\mathrm{tar}}$ such that:
\begin{itemize}
    \item[(i)] $\mathcal{M}_{\mathrm{tar}}$ agrees with
    $\mathcal{M}_{\mathrm{seed}}$ on the seed records;
    \item[(ii)] $\mathcal{M}_{\mathrm{tar}}$ satisfies the required mDP
    constraints on $\mathcal{X}_{\mathrm{tar}}$; and
    \item[(iii)] all local mechanisms are consistent on overlapping node
    regions induced by $\mathcal{G}$.
\end{itemize}
\end{definition}}

\DEL{
\begin{definition}[Extension algorithm]
\label{def:extendalgorithm}
Fix an extension graph $\mathcal{G}=(\mathcal{V},\mathcal{E},\{\mathcal{X}_v\}_{v\in\mathcal{V}})$.
An \emph{extension algorithm} is specified by a family of operators
$\{\mathrm{Ext}_v\}_{v\in\mathcal{V}}$, where each $\mathrm{Ext}_v$ takes as input the perturbation mechanism already specified on the parent region $\mathcal{X}_v$
and outputs perturbation mechanisms for all successor regions of $v$.

\smallskip
Concretely, let $\mathcal{M}_v^{\mathrm{in}}\in \mathsf{Mech}(\mathcal{X}_v)$ denote the mechanism available at node $v$
(e.g., inherited from a seed mechanism or fixed by earlier extensions along the extension process).
The extension rule at node $v$ is a mapping
\begin{equation}
\mathrm{Ext}_v:\ \mathsf{Mech}(\mathcal{X}_v)\ \longrightarrow\ \prod_{u\in\mathrm{Succ}(v)} \mathsf{Mech}(\mathcal{X}_u),
\label{eq:ext-operator}
\end{equation}
and its output is the collection of successor mechanisms
\begin{equation}
\bigl\{\mathcal{M}_{v\to u}\bigr\}_{u\in\mathrm{Succ}(v)}
\ \coloneqq\
\mathrm{Ext}_v\!\bigl(\mathcal{M}_v^{\mathrm{in}}\bigr),
\label{eq:ext-operator-apply}
\end{equation}
where each $\mathcal{M}_{v\to u}\in \mathsf{Mech}(\mathcal{X}_u)$ specifies a valid conditional distribution
$\mathcal{M}_{v\to u}(\cdot\mid \mathbf{x})$ for every $\mathbf{x}\in\mathcal{X}_u$.
\end{definition}}

\DEL{
We then introduce the two structural constraints on how the extension graph should be constructed: 
\paragraph{(G1) Target coverage and seed embedding.}
There exists a nonempty subset of nodes 
$\mathcal{V}_{\mathrm{seed}} \subseteq \mathcal{V}$, called \emph{seed nodes}, such that
\begin{equation}
\label{eq:G1}
  \bigcup_{v \in \mathcal{V}} \mathcal{X}_v \supseteq \mathcal{X}_{\mathrm{tar}},
  \quad
\end{equation}
Thus, each seed region is embedded in at least one seed node of the DAG.
Multiple seed nodes are allowed, and different seed regions may be extended independently.

\paragraph{(G2) Successor extension (branching extension).}
For every node $u \in \mathcal{V}$ with at least one successor,
\begin{equation}
  \mathcal{X}_u \subseteq 
  \bigcup_{v \in \mathrm{Succ}(u)} \mathcal{X}_v.
  \label{eq:G2}
\end{equation}
That is, the region associated with $u$ is covered by the union of its direct successors.
Nodes with $\mathrm{Succ}(u)=\varnothing$ are terminal regions where extension stops.

Condition (G2) enables a branching extension structure:
a region may be decomposed into multiple successor regions that can be processed independently,
supporting parallel computation and grid-based interpolation.}

\begin{figure}[t]
\centering
\hspace{0.00in}
\begin{minipage}{0.48\textwidth}
  \subfigure[Global seed-set extension]{
\includegraphics[width=1.00\textwidth]{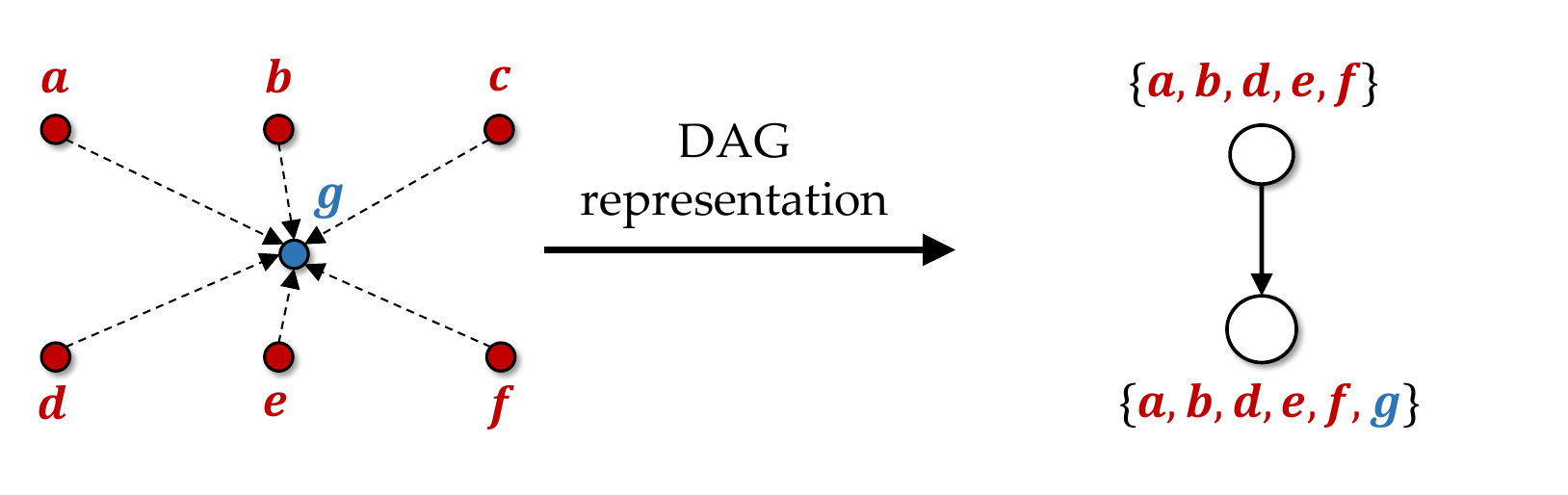}}
  \subfigure[Parallel local extension]{
\includegraphics[width=1.00\textwidth]{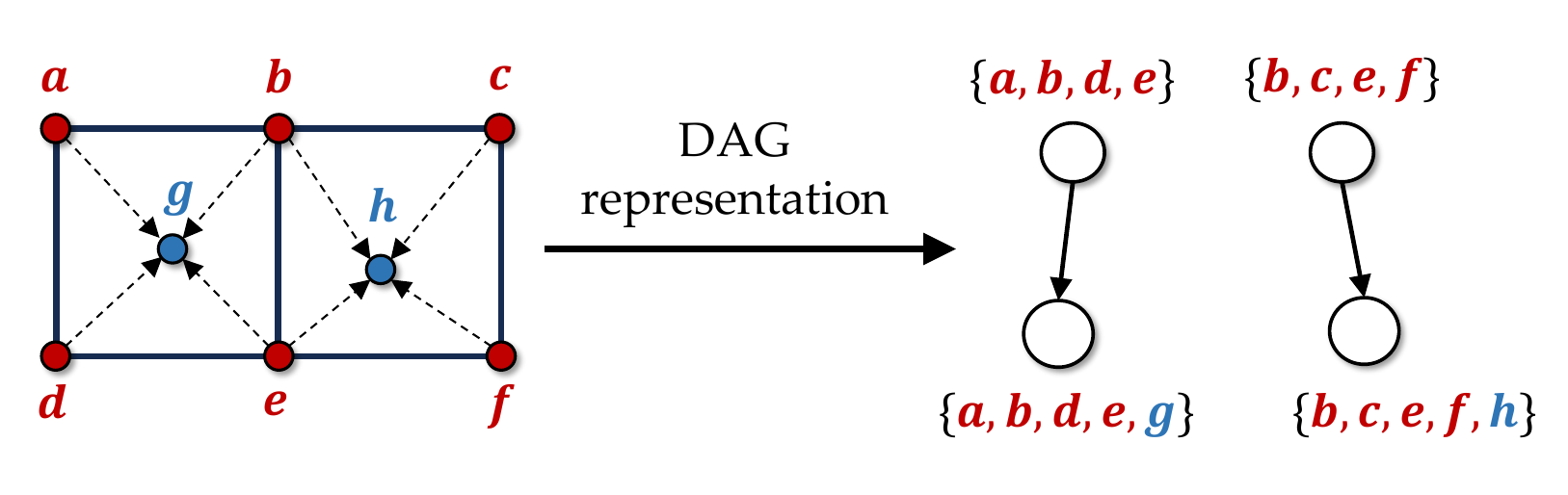}}
\end{minipage}
\caption{Existing extension schemes viewed as instances of the extension-graph abstraction.}
\label{fig:DAGexamples}
\end{figure}

Next, we show that two representative extension schemes from prior work can be viewed as special cases of the proposed extension-graph abstraction. As illustrated in Fig.~\ref{fig:DAGexamples}(a)(b), \emph{global seed-set extension}~\cite{borgs2018extend} performs a one-step global expansion of a seed set, whereas \emph{parallel local extension}~\cite{Qiu-USec2026} performs one-step local extensions over multiple overlapping subregions. Notably, both are shallow instances of the same abstraction; in contrast, \emph{the proposed extension graph supports recursive multi-step extension, in which mechanisms are propagated through multiple layers of successor regions before reaching the target records.}

\noindent \textbf{(a)} \textbf{Global seed-set extension}~\cite{borgs2018extend}.
Fig.~\ref{fig:DAGexamples}(a) shows a McShane-Whitney extension from an already specified record set (e.g., $\{{\rd a},{\rd b},{\rd d}$, ${\rd e},f\}$) to additional target records (e.g., ${\bl g}$), while preserving the perturbation distributions on the original records. Under our abstraction, this is represented by a predecessor-to-successor transition in which the successor region enlarges the predecessor region by adjoining the newly extended records.

\noindent \textbf{(b)} \textbf{Parallel local extension}~\cite{Qiu-USec2026}.
Fig.~\ref{fig:DAGexamples}(b) shows extension over two adjacent subregions. Each subregion selects local anchors (left: $\{{\rd a},{\rd b},{\rd d},{\rd e}\}$; right: $\{{\rd b},{\rd c},{\rd e},{\rd f}\}$) and independently generates new interior records (e.g., ${\bl g}$ and ${\bl h}$) via log-convex interpolation followed by normalization. Under our abstraction, this corresponds to parallel local extension steps over overlapping regions, where shared boundary anchors (e.g., $\{{\rd b},{\rd e}\}$) require consistency across neighboring subproblems.

\subsection{mDP Extension Algorithms}
\label{subsec:algorithms}
We next address the algorithmic side of the framework. Given an extension graph, an mDP extension algorithm specifies how mechanisms are initialized on seed records $\mathcal{X}_{\mathrm{seed}}$ and propagated through successor regions.

For any record set $\mathcal{S}\subseteq\mathcal{X}_{\mathrm{tar}}$, let
$\mathsf{Mech}(\mathcal{S})$ denote the set of perturbation mechanisms
from $\mathcal{S}$ to the common output domain $\mathcal{Y}$. That is,
each $\mathcal{M}\in\mathsf{Mech}(\mathcal{S})$ assigns a distribution
$\mathcal{M}(\cdot\mid x)\in\Delta(\mathcal{Y})$ to every $\mathbf{x}\in\mathcal{S}$.
We use $\mathcal{M}$, $\mathcal{M}_{\mathrm{seed}}$, and $\mathcal{M}_v$
for specific mechanisms, and reserve $\mathsf{Mech}(\cdot)$ for mechanism spaces.

\begin{definition}[Extension algorithm]
\label{def:extension-algorithm}
Fix an extension graph $\mathcal{G}=(\mathcal{V},\mathcal{E},\{\mathcal{X}_v\}_{v\in\mathcal{V}})$.
An \emph{extension algorithm} consists of:
\begin{itemize}
    \item for each node $v\in\mathcal{V}$ with $\mathrm{Succ}(v)\neq\varnothing$, a local extension operator
    \[
    \textstyle \mathrm{Ext}_v:\mathsf{Mech}(\mathcal{X}_v)\to
    \prod_{u\in\mathrm{Succ}(v)}\mathsf{Mech}(\mathcal{X}_u),
    \]
    which maps a mechanism $\mathcal{M}_v$ on $\mathcal{X}_v$ to a collection of successor mechanisms
    $\{\mathcal{M}_{v\to u}\}_{u\in\mathrm{Succ}(v)}$, where
    $\mathcal{M}_{v\to u}\in\mathsf{Mech}(\mathcal{X}_u)$; and
    \item an initialization rule that assigns the prescribed seed mechanism
    $\mathcal{M}_{\mathrm{seed}}$ to the seed records
    $\mathcal{X}_{\mathrm{seed}}$, together with a propagation rule that applies
    the local operators through the graph to construct a mechanism on
    $\mathcal{X}_{\mathrm{tar}}$.
\end{itemize}
\end{definition}

We then introduce the three requirements, \textbf{\emph{(A1)--(A3)}}, that ensure the resulting extended mechanisms are well defined and satisfy the desired $(\epsilon,d)$-mDP constraints when assembled into a global mechanism.
\DEL{
\smallskip
\noindent\textbf{Inheritance on overlaps.}
The output mechanisms must extend the predecessor mechanism on the overlap:
for every $u\in\mathrm{Succ}(v)$,
\begin{equation}
\mathcal{M}_{v\to u}(\mathbf{y} \mid \mathbf{x})\ =\ \mathcal{M}_v^{\mathrm{in}}(\mathbf{y} \mid \mathbf{x}),
\qquad \forall \mathbf{x}\in\mathcal{X}_v\cap\mathcal{X}_u,\ \forall \mathbf{y}\in\mathcal{Y}.
\label{eq:ext-overlap-inherit}
\end{equation}
Equation~\eqref{eq:ext-overlap-inherit} captures the fact that $\mathcal{M}_{v\to u}$ is an \emph{extended} mechanism: it preserves the already-fixed distributions from the predecessor side and assigns distributions to the remaining records in $\mathcal{X}_u$.}

\DEL{
\smallskip
\noindent\textbf{Seed inheritance.}
For each seed region $\mathcal{X}^{(k)}_{\mathrm{seed}}$, let $v_k\in\mathcal{V}_{\mathrm{seed}}$ be a corresponding seed node such that
$\mathcal{X}^{(k)}_{\mathrm{seed}}\subseteq \mathcal{X}_{v_k}$.
The extension algorithm is required to inherit the prescribed seed mechanism on that region:
\begin{equation}
\mathcal{M}(\mathbf{y} \mid x;\mathcal{X}_{v_k}) = Q^{(k)}_{\mathrm{seed}}(\mathbf{y} \mid \mathbf{x}),
\qquad \forall \mathbf{x}\in \mathcal{X}^{(k)}_{\mathrm{seed}},\ \forall \mathbf{y}\in\mathcal{Y}.
\label{eq:seed-inherit}
\end{equation}
We view~\eqref{eq:seed-inherit} as part of the definition of extension: the algorithm extends the given
$\{Q^{(k)}_{\mathrm{seed}}\}_{k\in K}$ from $\{\mathcal{X}^{(k)}_{\mathrm{seed}}\}_{k\in K}$ to the remaining records in the target domain.}

\subsubsection*{\textbf{(A1) Local mDP constraints.}}
We require that for every predecessor $v\in\mathcal{V}$, the mechanisms produced for its successor regions jointly satisfy the $(\epsilon,d)$-mDP inequality over the \emph{entire extension cover} of $v$.
Formally, as Fig.~\ref{fig:algorithmrequirements}(a) shows, for any $v\in\mathcal{V}$ and any two successors $u,w\in\mathrm{Succ}(v)$ (allowing $u=w$), the edge mechanisms
$\mathcal{M}_{v\to u}$ and $\mathcal{M}_{v\to w}$ must satisfy
\begin{equation}
\mathcal{M}_{v\to u}(\mathbf{x})\ \stackrel{\epsilon}{\approx}\ \mathcal{M}_{v\to w}(\mathbf{x}'),
\qquad \forall \mathbf{x}\in\mathcal{X}_u,\ \forall \mathbf{x}'\in\mathcal{X}_w,
\label{eq:A1-local-cross}
\end{equation}
When $u=w$, Eq.~\eqref{eq:A1-local-cross} reduces to the usual \emph{local} $(\epsilon,d)$-mDP constraints within a single successor region $\mathcal{X}_v$; when $u\neq w$, it enforces \emph{cross-successor} privacy across sibling regions refined from the same predecessor.

\textbf{Appendix \ref{sec:app:Interpo}} presents three representative extension algorithms, \emph{McShane--Whitney}, \emph{Log-convex}, and \emph{LP-based} interpolation methods, that can be designed to satisfy the \emph{local mDP constraint} \textbf{\emph{(A1)}}. 



\begin{figure}[t]
\centering
\hspace{0.00in}
\begin{minipage}{0.48\textwidth}
  \subfigure[Local mDP constraints]{
\includegraphics[width=0.48\textwidth]{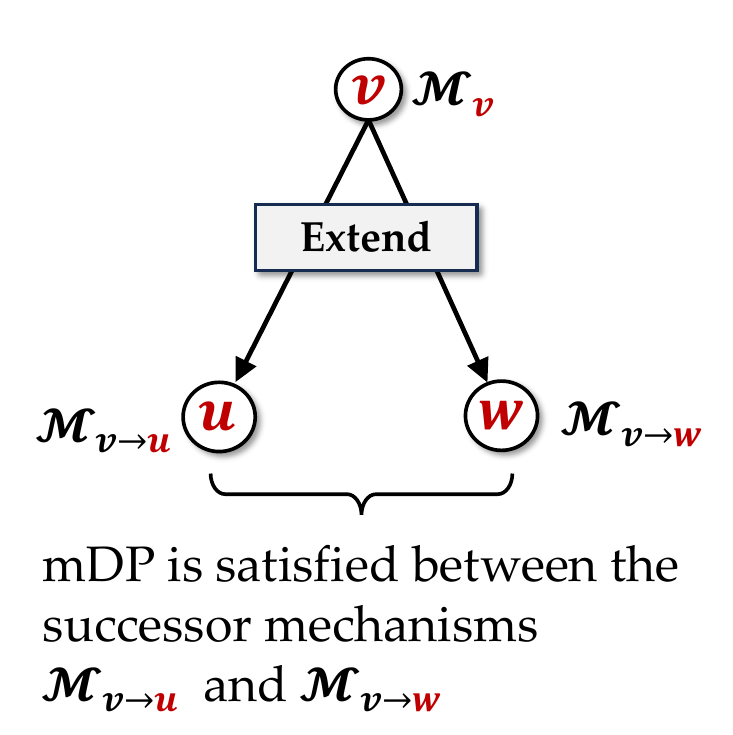}}
\hspace{0.1in}
  \subfigure[Overlap consistency]{
\includegraphics[width=0.48\textwidth]{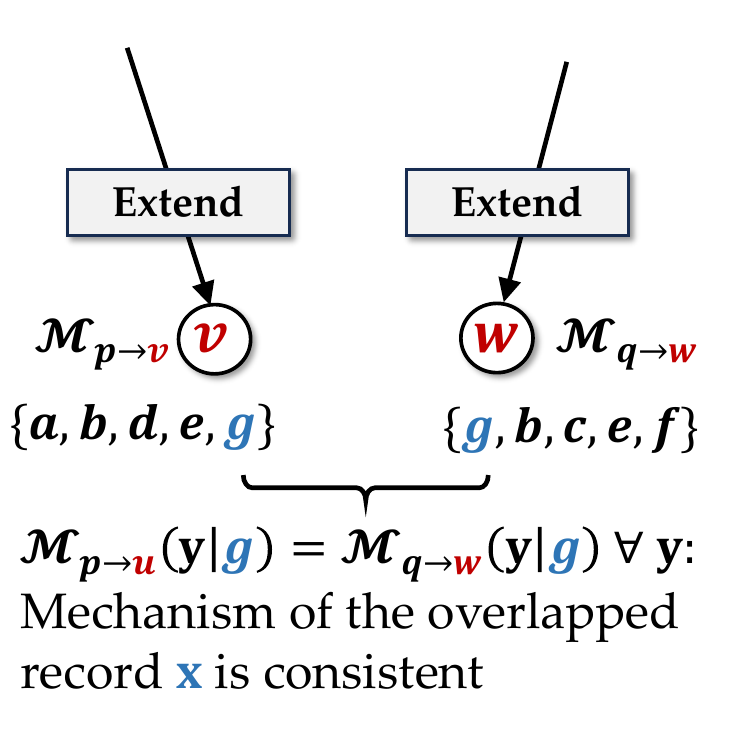}}
  \subfigure[Successor mDP preservation]{
\includegraphics[width=0.48\textwidth]{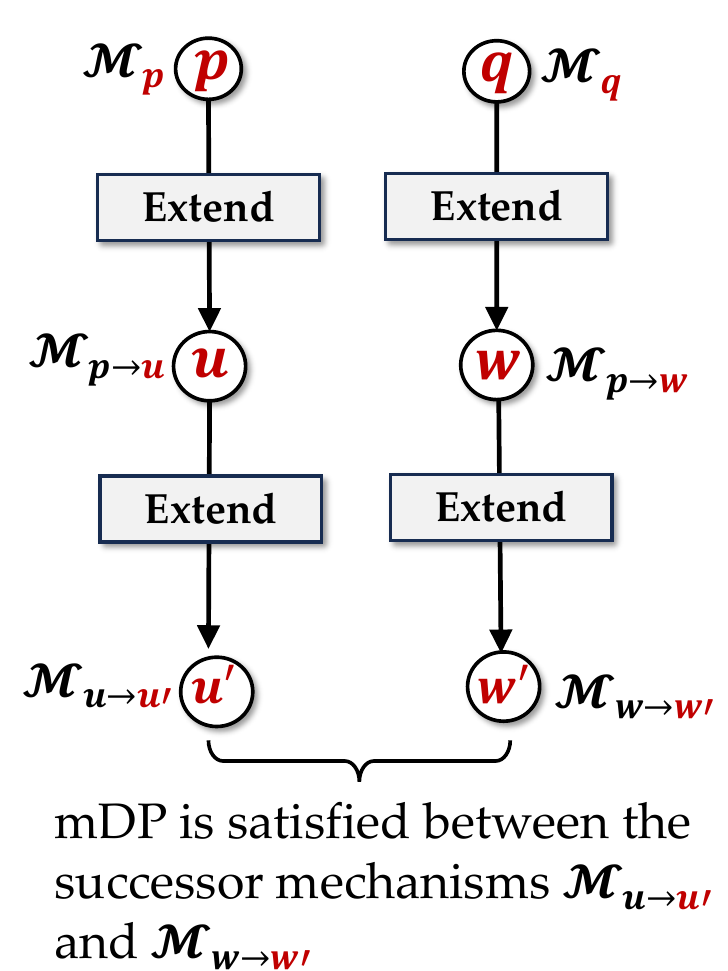}}
\hspace{0.1in}
  \subfigure[Descendant mDP preservation (Lemma \ref{lem:descendant-closure})]{
\includegraphics[width=0.48\textwidth]{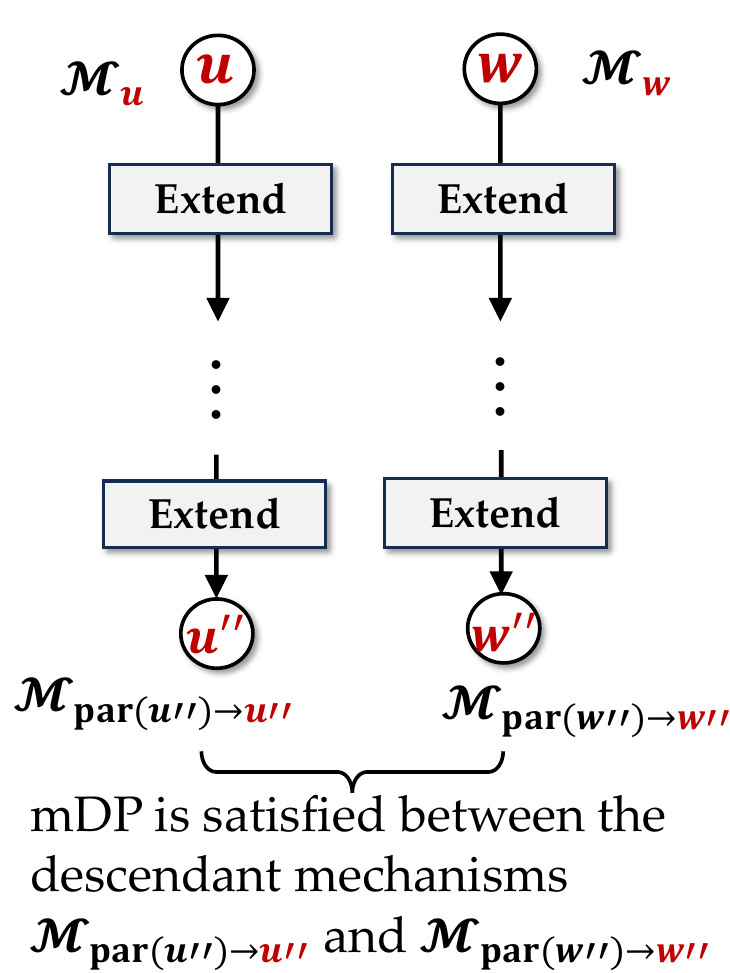}}
\end{minipage}
\caption{Algorithmic requirements for extension on a DAG. 
}
\label{fig:algorithmrequirements}
\end{figure}


Notably, local constraints alone are not sufficient to guarantee global mDP: two points $\mathbf{x},\mathbf{x}'$ that belong to different, disjoint nodes may never appear together in the same local constraint. To bridge this gap, next we introduce additional requirements (A2) and (A3) that relate successors in the extension hierarchy and ensure that privacy guarantees are preserved under extension.


\subsubsection*{\textbf{(A2) Overlap consistency.}}
The extended mechanisms must agree on overlapping records along extension paths. Specifically, for any two edges $(p,v),(q,w)\in\mathcal{E}$ such that 
as Fig.~\ref{fig:algorithmrequirements}(b) shows, for any overlapped record $\mathbf{x}\in \mathcal{X}_v\cap \mathcal{X}_w$ (e.g., $\mathbf{x} = b$), we require
\begin{equation}
\mathcal{M}_{p\to v}(\mathbf{y} \mid \mathbf{x})=\mathcal{M}_{q\to w}(\mathbf{y} \mid \mathbf{x}),
\qquad \forall \mathbf{y}\in\mathcal{Y}.
\label{eq:A2-overlap-general}
\end{equation}
That is, whenever two node regions overlap, the perturbation distribution assigned to a shared record must be identical.

A particularly important special case arises when a node $u$ has multiple predecessors. In this case, for any $v_1,v_2\in\mathrm{Pred}(u)$, we require
\begin{equation}
\mathcal{M}_{v_1\to u}(\mathbf{y} \mid \mathbf{x})=\mathcal{M}_{v_2\to u}(\mathbf{y} \mid \mathbf{x}),
\qquad \forall \mathbf{x}\in\mathcal{X}_u,\ \forall \mathbf{y}\in\mathcal{Y}.
\label{eq:A2-multiparent}
\end{equation}
Under condition~\eqref{eq:A2-multiparent}, the node mechanism $\mathcal{M}_u(\mathbf{y} \mid \mathbf{x})$ is well defined. In practice, successor regions are often chosen to be disjoint, or to overlap only on boundary sets. Moreover, merge nodes are typically constructed so that their incoming mechanisms already coincide. These design choices simplify the enforcement of~\eqref{eq:A2-overlap-general} and~\eqref{eq:A2-multiparent}.

\subsubsection*{\textbf{(A3) Successor-level mDP preservation}}
As illustrated in Fig.~\ref{fig:algorithmrequirements}(c), once a cross-node $\varepsilon$-mDP relation has been established between two extended regions, that relation must be preserved under one additional extension step. Specifically, suppose that two node mechanisms $\mathcal{M}_{p\to u}$ and $\mathcal{M}_{q\to w}$ satisfy
\begin{equation}
\mathcal{M}_{p\to u}(\mathbf{x})\ \stackrel{\epsilon}{\approx}\ \mathcal{M}_{q\to w}(\mathbf{x}'),
\qquad \forall \mathbf{x}\in\mathcal{X}_u,\ \forall \mathbf{x}'\in\mathcal{X}_w,
\label{eq:A3-premise}
\end{equation}
where $(p,u),(q,w)\in\mathcal{E}$. Then the same relation must continue to hold after one additional extension step: for every pair of direct successors
$u'\in \mathrm{Succ}(u)$ and $w'\in \mathrm{Succ}(w)$,
\begin{equation}
\mathcal{M}_{u\to u'}(\mathbf{x})\ \stackrel{\epsilon}{\approx}\ \mathcal{M}_{w\to w'}(\mathbf{x}'),
\qquad \forall \mathbf{x}\in\mathcal{X}_{u'},\ \forall \mathbf{x}'\in\mathcal{X}_{w'}.
\label{eq:A3-conclusion}
\end{equation}
In other words, once the $\varepsilon$-mDP relation is established between two nodes, all of their immediate extensions must satisfy mDP.

\smallskip
For convenience, we write $\mathrm{Desc}(v)$ for the set of all descendants of a node $v\in\mathcal{V}$, including $v$ itself; that is,
$\mathbf{x}\in\mathrm{Desc}(v)$ if and only if there exists a (possibly empty) directed path from $v$ to $z$ in $\mathcal{G}$.

\begin{lemma}[Descendant closure of cross-node mDP]
\label{lem:descendant-closure}
As illustrated in Fig.~\ref{fig:algorithmrequirements}(d), consider two
possibly identical nodes $u,w\in\mathcal{V}$ whose associated mechanisms
satisfy
\begin{equation}
\mathcal{M}_u(\mathbf{x})
\stackrel{\epsilon}{\approx}
\mathcal{M}_w(\mathbf{x}'),
\qquad
\forall\mathbf{x}\in\mathcal{X}_u,\quad
\forall\mathbf{x}'\in\mathcal{X}_w.
\label{eq:lemma-premise}
\end{equation}
Suppose that \emph{(A1)--(A3)} hold along all subsequent extensions from
$u$ and $w$. Then, for every pair of
respective descendants $u''\in\mathrm{Desc}(u)$ and
$w''\in\mathrm{Desc}(w)$, their associated mechanisms satisfy
\begin{equation}
\mathcal{M}_{\mathrm{par}(u'')\to u''}(\mathbf{x})\ \stackrel{\epsilon}{\approx}\ \mathcal{M}_{\mathrm{par}(w'')\to w''}(\mathbf{x}'),
~\forall \mathbf{x}\in\mathcal{X}_{u''},\ \forall \mathbf{x}'\in\mathcal{X}_{w''}.
\label{eq:lem-conclusion}
\end{equation}
\end{lemma}


\subsubsection*{\textbf{Validity of the induced global mechanism}} We are now ready to show that the induced global mechanism (\textbf{Definition \ref{def:globalmech}}) is $(\epsilon,d)$-mDP (in \textbf{Theorem \ref{thm:global-mdp}}).
\begin{definition}[Induced global mechanism]
\label{def:globalmech}
Let $\{\mathcal{M}_v\}_{v\in\mathcal{V}_0}$ denote the source mechanisms
and let $\{\mathcal{M}_{v\to u}\}_{(v,u)\in\mathcal{E}}$ denote the
edge-local mechanisms produced by the extension algorithm. We define the
induced global mechanism $\mathcal{M}$ on $\mathcal{X}_{\mathrm{tar}}$ as
follows. For any $\mathbf{x}\in\mathcal{X}_{\mathrm{tar}}$, choose any
local assignment that contains $\mathbf{x}$:
\[
\mathcal{M}(\mathbf{y}\mid \mathbf{x})
\coloneqq
\begin{cases}
\mathcal{M}_u(\mathbf{y}\mid \mathbf{x}), 
& \text{if } u\in\mathcal{V}_0 \text{ and } \mathbf{x}\in\mathcal{X}_u,\\[2mm]
\mathcal{M}_{v\to u}(\mathbf{y}\mid \mathbf{x}),
& \text{if } (v,u)\in\mathcal{E} \text{ and } \mathbf{x}\in\mathcal{X}_u .
\end{cases}
\]
If multiple such local assignments exist, they give the same distribution
by source-level compatibility and overlap consistency. Hence
$\mathcal{M}$ is well defined.
\end{definition}

\begin{theorem}[Global $(\epsilon,d)$-mDP guarantee via extension algorithms]
\label{thm:global-mdp}

Consider a metric space $(\mathcal{X},d)$, a target set $\mathcal{X}_{\mathrm{tar}}\subseteq\mathcal{X}$, and a finite extension graph
$\mathcal{G}=(\mathcal{V},\mathcal{E}, \{\mathcal{X}_v\}_{v\in\mathcal{V}})$ whose local regions cover $\mathcal{X}_{\mathrm{tar}}$ and whose relevant nodes are reachable from source nodes
$\mathcal{V}_0=\{v\in\mathcal{V}:\mathrm{Pred}(v)=\varnothing\}$. Suppose the source regions are equipped with initial mechanisms $\{\mathcal{M}_v\}_{v\in\mathcal{V}_0}$ that agree with the prescribed seed mechanism $\mathcal{M}_{\mathrm{seed}}$ on $\mathcal{X}_{\mathrm{seed}}$, are consistent on overlapping source regions, and jointly satisfy $(\epsilon,d)$-mDP over $\bigcup_{v\in\mathcal{V}_0}\mathcal{X}_v$.

Assume that the extension algorithm propagates these source mechanisms through $\mathcal{G}$ and produces edge-local mechanisms $\{\mathcal{M}_{v\to u}\}_{(v,u)\in\mathcal{E}}$. The induced global mechanism $\mathcal{M}$ in Definition~\ref{def:globalmech} is well defined and satisfies $(\epsilon,d)$-mDP on $\mathcal{X}_{\mathrm{tar}}$ if and only if the source and edge-local mechanisms satisfy \textbf{\emph{(A1)--(A3)}}.
\end{theorem}
\begin{proof}[Proof sketch]
The proof establishes both directions. For sufficiency, overlap consistency \textbf{\emph{(A2)}} ensures that the induced mechanism is well defined because all local mechanisms assign the same distribution to any shared record. For any two target records, local mDP \textbf{\emph{(A1)}} establishes the required relation when they belong to the same extension region, while \textbf{Lemma~\ref{lem:descendant-closure}} uses successor-level preservation \textbf{\emph{(A3)}} to propagate cross-node mDP relations to their descendant regions. Hence, the induced global mechanism satisfies $(\epsilon,d)$-mDP. For necessity, any local mechanism constituting a global mDP mechanism must satisfy \textbf{\emph{(A1)}} as its restriction; mechanisms sharing a record must agree on its distribution, yielding \textbf{\emph{(A2)}}; and global mDP restricted to any pair of successor regions yields \textbf{\emph{(A3)}}. Therefore, \textbf{\emph{(A1)--(A3)}} are necessary and sufficient within the proposed extension framework. The formal proof is given in \textbf{Appendix~\ref{subsec:proof:thm:global-mdp}}.
\end{proof}
\begin{remark}
Intuitively, Theorem~\ref{thm:global-mdp} formalizes the idea that $(\epsilon,d)$-mDP can be enforced \emph{hierarchically}: once the privacy bound is established at coarse regions, the successor-level preservation rule (A3) ensures that all extensions inherit the same bound. As a consequence, the combination of local constraints within nodes and cross-successor constraints
that are preserved under extension is {\rev necessary and sufficient } to guarantee global $(\epsilon,d)$-mDP over the entire domain.

\end{remark}

\smallskip
\noindent \textbf{Discussion and Algorithmic Implications.}
Among the three algorithmic requirements, \textbf{\emph{(A1)}} and \textbf{\emph{(A2)}} are relatively straightforward to enforce in practice. In particular, \textbf{\emph{(A1)}} can be imposed through a local extension procedure over the successor regions of each node, while \textbf{\emph{(A2)}} can often be ensured either by construction (e.g., by choosing disjoint successors) or by explicitly enforcing consistency on overlapping records.

In contrast, \textbf{\emph{(A3)}} is substantially more subtle. The successor-level preservation condition couples mechanisms across extension levels and ensures that privacy guarantees established at a coarse level are inherited by all subsequent extensions. Whether this condition can be enforced efficiently depends on the construction of the extension graph and on the geometry of the underlying metric space, as both affect how local privacy constraints interact and propagate through the extension process. Enforcing \textbf{\emph{(A3)}} directly at the global level would undermine the scalability benefits of the graph-based decomposition. Therefore, the central algorithmic challenge lies in \emph{how to design local extension rules and optimization procedures that implicitly guarantee \textbf{(A3)} while remaining computationally efficient.}

In the next section, we specialize the framework to a tree-based extension algorithm in an $\ell_p$ metric space, and show how \textbf{\emph{(A1)}}--\textbf{\emph{(A3)}} can be enforced by construction using only local information at each extension step.

\DEL{
\subsection{Dependency Graph Model.}
To enable scalable construction, we decompose the domain into subsets and optimize them sequentially. 
Let $\mathcal{X}$ denote the domain of real records, and let $\{\mathcal{X}_v\}_{v\in\mathcal{V}}$ be a collection of subsets such that
$\bigcup_{v\in\mathcal{V}}\mathcal{X}_v=\mathcal{X}$, where the subsets are not necessarily disjoint (i.e., $\mathcal{X}_u\cap \mathcal{X}_v$ may be nonempty for $u\neq v$).
For each node $v$, we aim to derive a perturbation distribution
\begin{equation}
\mathcal{M}_v(y \mid \mathbf{x}),
\qquad
\mathbf{x} \in \mathcal{X}_v,
y \in \mathcal{Y}.
\end{equation}

We organize these subsets using a dependency graph. 
Let $\mathcal{G}=(\mathcal{V},\mathcal{E})$ be a \emph{directed acyclic graph (DAG)} in which a directed edge $(u,v)\in\mathcal{E}$ indicates that the perturbation distribution at node $v$ depends on (or is obtained by extending) the solution at node $u$. 
The predecessor set of node $v$ is defined as
\begin{equation}
\mathrm{Pa}(v)=\{\,u\in\mathcal{V} : (u,v)\in\mathcal{E}\,\}.
\end{equation}


The acyclicity of $\mathcal{G}$ guarantees the existence of a
topological order such that if $(u,v) \in \mathcal{E}$,
then node $u$ is optimized based on the result of node $v$. For example, if the construction is recursive, we impose
\begin{equation}
\mathcal{M}_v(y \mid \mathbf{x})
=
\mathcal{T}_{u \to v}
\big(
\mathcal{M}_u(\cdot \mid \tilde{x})
\big),
\qquad
\forall \mathbf{x} \in \mathcal{X}_v,
\end{equation}
where $\mathcal{T}_{u \to v}$ denotes an extension operator
that maps the predecessor distribution to the successor domain.

\DEL{
\paragraph{Node-Level Optimization.} For each node $v$, we solve
\begin{equation}
\begin{aligned}
\min_{\mathcal{M}_v} \quad 
& \mathcal{L}_v(\mathcal{M}_v) \\
\text{s.t.} \quad 
& \mathcal{M}_v(\cdot \mid \mathbf{x}) \in \Delta(\mathcal{Y}),
\qquad \forall \mathbf{x} \in \mathcal{X}_v, \\
& \mathcal{M}_v(y \mid \mathbf{x})
\le
e^{\epsilon d(\mathbf{x},\mathbf{x}')}
\mathcal{M}_v(y \mid \mathbf{x}'),
\qquad
\forall \mathbf{x},\mathbf{x}' \in \mathcal{X}_v, y \in \mathcal{Y}, \\
& \text{predecessor consistency constraints.}
\end{aligned}
\end{equation}

\paragraph{Global Optimization.}

The overall multi-stage problem is

\begin{equation}
\begin{aligned}
\min_{\{\mathcal{M}_v\}_{v \in \mathcal{V}}}
\quad &
\sum_{v \in \mathcal{V}} \mathcal{L}_v(\mathcal{M}_v) \\
\text{s.t.} \quad
& \text{mDP constraints at each node}, \\
& \text{tree-structured predecessor consistency constraints}.
\end{aligned}
\end{equation}

}

\paragraph{Extension Operator.}
For any subset $S \subseteq \mathcal{X}$, define the space of perturbation mechanism
\begin{equation}
\mathcal{M}(\mathcal{S}) := \Big\{
Q: \mathcal{S} \times \mathcal{Y} \to [0,1]
\Big|
\sum_{y \in \mathcal{Y}} Q(y \mid \mathbf{x})=1,\ \forall \mathbf{x} \in S
\Big\}.
\end{equation}
For an edge $(u,v) \in \mathcal{E}$, an \emph{extension operator} is a  map
\begin{equation}
\mathcal{T}_{u \to v}: \mathcal{M}(\mathcal{X}_u) \to \mathcal{M}(\mathcal{X}_v),
\qquad
\mathcal{M}_v = \mathcal{T}_{u \to v}(\mathcal{M}_u).
\end{equation}

{\rd Here, we need to discuss the feasibility of the extension, which is non-trivial.}

\paragraph{Closure Under Multi-Step Extension.}

If every edge operator on a root-to-leaf path is non-expansive, then the composed extension operator along the path is also non-expansive, hence preserves $(\epsilon,d)$-mDP throughout the multi-stage construction.

{\rd Ruiyao, it would be interesting if you could formally prove that when using certain extension algorithms, such as predefined function–based interpolation, the final result is independent of the order of extension (which has been discussed). If this can be established, it would justify focusing on the more general cases where order does matter in the remaining part of the paper.}

\paragraph{Tree-Structured Restriction.} In a general DAG, a node may have multiple predecessors,
which may induce incompatible extension constraints.
To ensure well-defined recursive perturbation construction,
we restrict $\mathcal{G}$ to be a rooted tree:

\begin{equation}
|\mathrm{Pa}(v)| \le 1,
\qquad \forall v \in \mathcal{V} \setminus \{r\},
\end{equation}
where $r$ denotes the root node.
This guarantees that each node has at most one predecessor,
thus avoiding multi-predecessor inconsistency.}


\section{Tree-based Extension Algorithms}
\label{sec:tree}
In this section, we instantiate the general graph-based framework with a
\emph{tree-based} extension algorithm for hierarchical grid domains. This
setting is natural for geo-location secrets~\cite{Qiu-EDBT2024} and other
structured domains that admit multi-resolution representations~\cite{Samet-CSurvey1984}.
The key idea is to optimize perturbation distributions on coarse seed records
and then recursively extend them to finer grid records through local interpolation.

We first define the grid-based secret representation and the induced extension
tree in \textbf{\S\ref{subsec:grid-tree}}. We then present the coordinate-wise
interpolation principle and dimension-wise composition argument in
\textbf{\S\ref{subsec:coordinate}}, which provide the key tool for
preserving mDP constraints under extension. Finally, in
\textbf{\S\ref{subsec:algorithm}}, we describe the recursive tree-based extension
algorithm and show how the framework requirements \textbf{\emph{(A1)}}--\textbf{\emph{(A3)}}
are enforced by construction. Together, these results yield a scalable mechanism
for fine-grained domains under an $\ell_p$ metric while preserving the desired
mDP guarantee.

\subsection{Grid-based Secret Representation}
\label{subsec:grid-tree}
{\rev We consider a target geographic area in $\mathbb{R}^2$, such as a city boundary. For any two locations $\mathbf{x}$ and $\mathbf{x}'$ in this area, their distance is measured using the $\ell_p$ metric 
\[
\textstyle 
d(\mathbf{x},\mathbf{x}')
=
\|\mathbf{x}-\mathbf{x}'\|_p
=
\left(
\sum_{t=1}^{2}|x_t-x_t'|^p
\right)^{1/p},
\qquad p\geq 1,
\]
where $\|\cdot\|_p$ denotes the $\ell_p$ norm. We discretize the target area using a multi-resolution grid and treat the grid vertices lying within the area (i.e., cell corners) as the secret records to be protected.} 

{\rev Here, each vertex represents a discretized location coordinate rather than membership in any of its neighboring cells; the cells serve only to organize the extension computation. When the original location is continuous, it can be mapped to a vertex through a predetermined quantization rule, such as nearest-vertex mapping. Vertices are convenient for our construction because adjacent cells share boundary vertices, which serve as common interpolation anchors and facilitate overlap consistency by construction. Although centroid-based grids can also be refined hierarchically, centroids are cell-specific and would require additional rules to coordinate interpolation and ensure consistency across neighboring cells.}
\begin{figure}[t]
\centering
\hspace{0.00in}
\begin{minipage}{0.48\textwidth}
 \subfigure{
\includegraphics[width=1.00\textwidth]{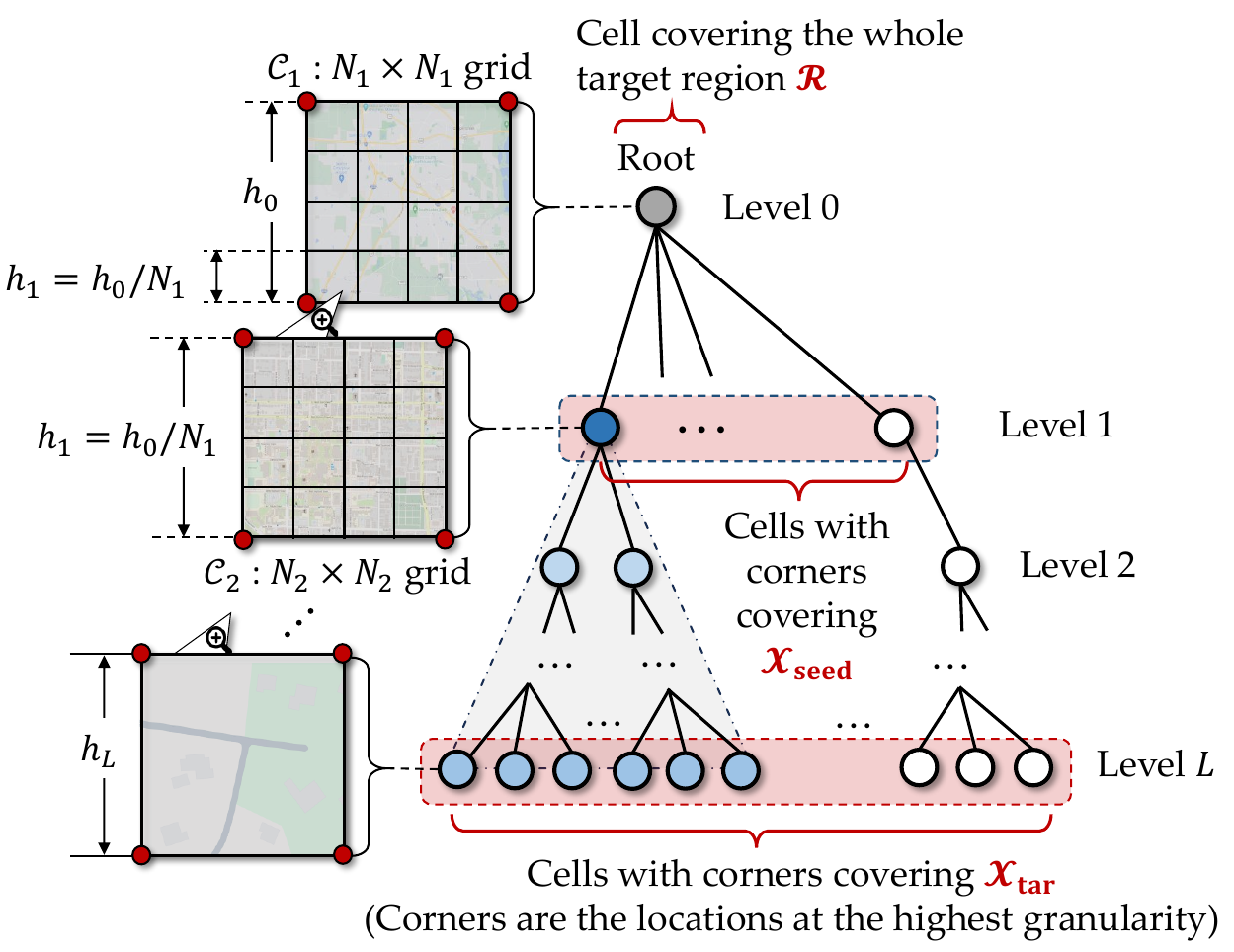}}
\end{minipage}
\caption{Tree representation of the grid-based secret records.}
\label{fig:tree}
\end{figure}

{\rev As illustrated in Fig.~\ref{fig:tree}, we fix a maximum extension level $L\in\mathbb{N}$ and a base cell side length $h_0>0$. For each level $r\in\{0,1,\ldots,L\}$, let $\mathcal{C}_r$ denote the collection of axis-aligned square grid cells with side length $h_r$ that intersect the target area. If a level-$r$ cell is refined into an $N_{r+1}\times N_{r+1}$ subgrid at level $r+1$, then $h_{r+1}=h_r/N_{r+1}$, and only the resulting subcells that intersect the target area are retained in $\mathcal{C}_{r+1}$.} 


We define the target secret domain $\mathcal{X}_{\mathrm{tar}}$ as the set of corner vertices of the retained level-$L$ cells, i.e., the \emph{finest-resolution discrete secret locations} in the target region. Let $r_{\mathrm{src}}<L$ denote the source level. The retained cells at this level provide the initial coarse representation and are chosen so that their associated record sets cover the seed records:
$\mathcal{X}_{\mathrm{seed}}\subseteq \bigcup_{v\in\mathcal{C}_{r_{\mathrm{src}}}}\mathcal{X}_v$.
These level-$r_{\mathrm{src}}$ cells serve as the source nodes of the extension graph and are equipped with initial mechanisms that agree with the prescribed seed mechanism on $\mathcal{X}_{\mathrm{seed}}$.

We then construct the extension tree over the retained cells from level $r_{\mathrm{src}}$ to level $L$, denoted by
$\mathcal{G}=(\mathcal{V},\mathcal{E},\{\mathcal{X}_v\}_{v\in\mathcal{V}})$:
\newline \textbf{(1) Node set $\mathcal{V}$}: Each node $v$ is identified with a retained cell at some level $r\in\{r_{\mathrm{src}},\ldots,L\}$. The source nodes are $\mathcal{V}_0\triangleq \mathcal{C}_{r_{\mathrm{src}}}$, and the full node set is $\mathcal{V}\triangleq \bigcup_{r=r_{\mathrm{src}}}^{L}\mathcal{C}_r$.
\newline \textbf{(2) Secret set $\{\mathcal{X}_v\}_{v\in\mathcal{V}}$}: For each node $v\in\mathcal{V}$, let $\mathrm{corner}(v)$ denote the set of corner vertices of cell $v$. We define $\mathcal{X}_v\triangleq \mathcal{X}_{\mathrm{tar}}\cap \mathrm{corner}(v)$. Thus, $\mathcal{X}_v$ contains the finest-resolution secret records located at the corners of $v$; a standard square cell contains at most four such records, while boundary cells may contain fewer.
\newline \textbf{(3) Edge set $\mathcal{E}$}: For any non-leaf node $v\in\mathcal{C}_r$ with $r<L$, its children are the retained subcells of $v$ in $\mathcal{C}_{r+1}$. We add an edge $(v,v')\in\mathcal{E}$ whenever $v'$ is a retained child of $v$. Since each retained cell at level $r+1$ is refined from exactly one retained cell at level $r$, every non-source node has a unique predecessor. Hence, the resulting extension graph is a rooted tree when there is a single source cell, and a forest when multiple source cells are used.

\begin{figure}
\begin{minipage}{0.48\textwidth}
  \subfigure{
\includegraphics[width=1.00\textwidth]{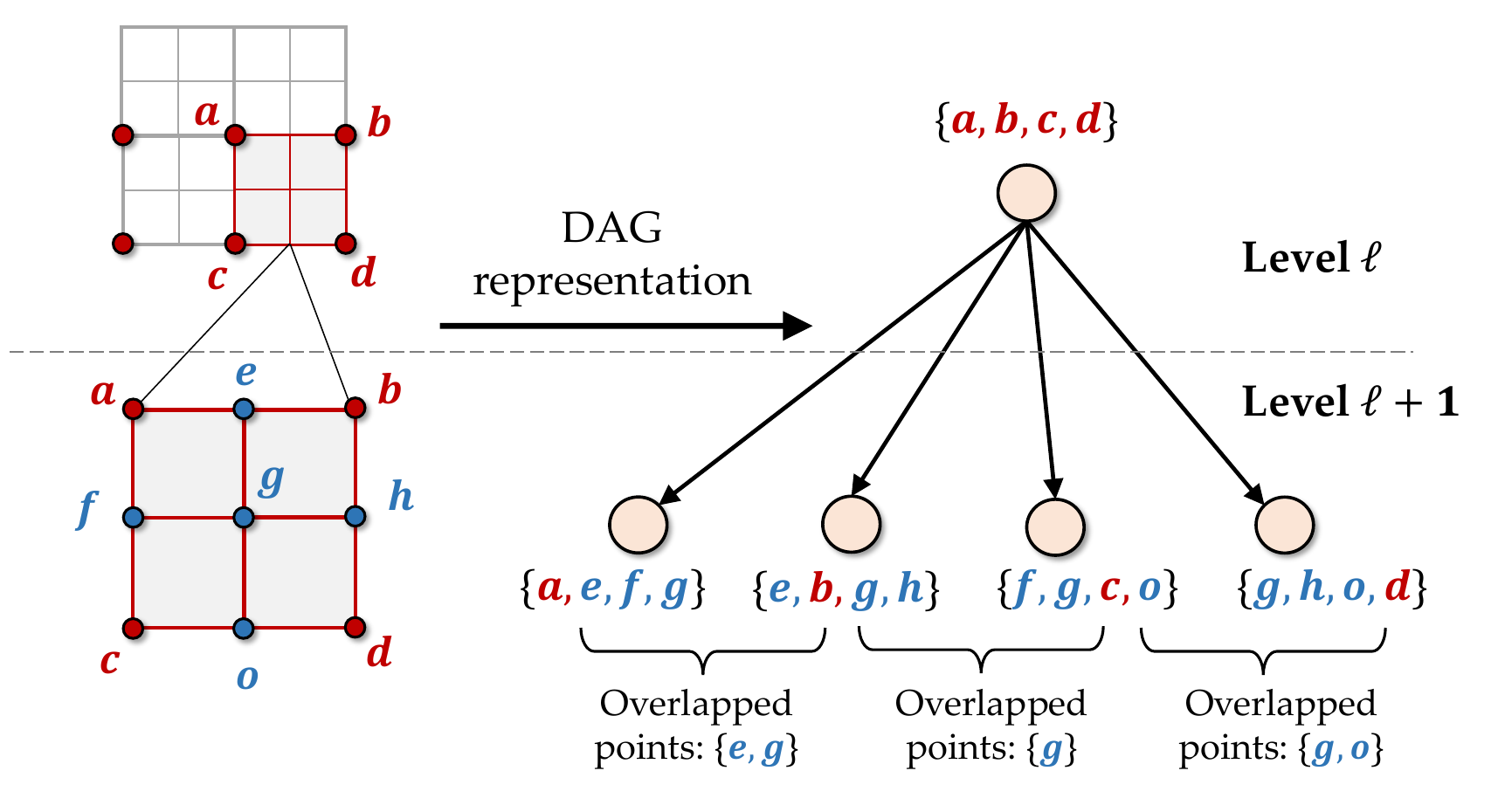}}
\end{minipage}
\caption{Example of a grid-cell hierarchy inducing a tree.}
\label{fig:grid-tree-example}
\end{figure}

Fig.~\ref{fig:grid-tree-example} illustrates the tree construction on a simple two-level grid. At level $t$, a coarse cell has four corner vertices $\{{\rd a},{\rd b},{\rd c},{\rd d}\}\subseteq \mathcal{X}_{\mathrm{tar}}$, which form the record set of the corresponding parent node. Extending this cell to level $t+1$ (here, a $2\times 2$ subgrid) introduces additional vertices (e.g., ${\bl e},{\bl f},{\bl g},{\bl h},{\bl o}$) and yields four child cells. In the induced tree (a special case of the extension graph), each child node corresponds to one subcell and is associated with the corner-vertex set of that subcell, e.g., $\{{\rd a},{\bl e},{\bl f},{\bl g}\}$, $\{{\bl e},{\rd b},{\bl g},{\bl h}\}$, $\{{\bl f},{\bl g},{\rd c},{\bl o}\}$, $\{{\bl g},{\bl h},{\bl o},{\rd d}\}$. Thus, the extension process from level $t$ to level $t+1$ is represented by a single parent node with multiple successor nodes.


\smallskip
\noindent\textbf{Discussion: Selection of seed records
$\mathcal{X}_{\mathrm{seed}}$.}
{\rev We select the corners of coarse-grid cells as seed records because they naturally define local extension blocks, provide shared anchors that maintain consistency across overlapping blocks, and support coordinate-aligned, dimension-wise log-convex interpolation. These properties enable the extension steps to satisfy requirements (A1)--(A3).

The grid resolution controls the seed density and thus the utility--efficiency trade-off. A denser seed set reduces the extent of interpolation and can better preserve the optimized mechanism across the target domain, but increases the cost of seed-level optimization. Conversely, a sparser seed set improves scalability but requires extension over larger regions, which may increase utility loss. In practice, candidate coarse-grid resolutions can be evaluated to select the sparsest seed set that meets the desired utility requirement or, equivalently, the densest set permitted by the computational budget. We empirically examine this trade-off across different seed-grid resolutions in Section~\ref{sec:experiments}.}




\subsection{Interpolation-Based Preservation via Dimension-Wise Composition}
\label{subsec:coordinate}

A central challenge in the tree-based extension algorithm is that direct
multi-dimensional interpolation does not, in general, preserve the
$(\epsilon,d)$-mDP relations required for recursive extension~\cite{Qiu-USec2026}.
For this reason, we construct the extension through multiple one-dimensional
interpolation steps, each of which can be controlled by a coordinate-wise
Lipschitz bound in log-probability space.
\textbf{Proposition~\ref{prop:1d-unified}} establishes the validity of a
single one-dimensional interpolation step, while
\textbf{Theorem~\ref{thm:composition}} shows how these coordinate-wise
guarantees compose into a global $(\epsilon,d_p)$-Lipschitz bound~\cite{Qiu-USec2026}.
Building on these two ingredients,
\textbf{Theorem~\ref{thm:successor-preservation-2d}} further shows that the
two-dimensional interpolation operator preserves cross-cell Lipschitz
relations after refinement, which is the key property needed to satisfy
successor-level mDP preservation \textbf{\emph{(A3)}}. Together, these results
provide the theoretical foundation for the tree-based extension procedure
described next.

To accommodate different interpolation schemes, we first define a generic one-dimensional interpolation rule in log-probability space. The rule is specified along a single coordinate line, matches the anchor log-probabilities, and satisfies a coordinate-wise slope cap. Proposition~\ref{prop:1d-validity} then shows that any such rule automatically preserves the desired $(\epsilon_t,d_1)$-Lipschitz bound for all interpolated records on that line.

\begin{definition}[One-dimensional interpolation]
\label{def:1d-interpolation-rule}
Fix a coordinate $t\in\{1,\dots,N\}$ and an output $\mathbf{y}\in\mathcal{Y}$.
Let $\mathbf{x}^{(1)},\dots,\mathbf{x}^{(m)}\in\hat{X}$ be anchor records that differ only in their $t$-th coordinate, with $
x^{(1)}_t < x^{(2)}_t < \cdots < x^{(m)}_t$.
Define the interval $\mathcal{I}_t := [x^{(1)}_t,\,x^{(m)}_t]$. A one-dimensional interpolation rule along coordinate $t$ is an absolutely continuous function $f_{t,\mathbf{y}}: \mathcal{I}_t \to \mathbb{R}$ 
such that:
\begin{enumerate}[label=(\roman*)]
    \item \textbf{Anchor consistency:} $f_{t,\mathbf{y}}(x^{(r)}_t)=\ln \mathcal{M}(\mathbf{y}
    |\mathbf{x}^{(r)})$,
    $r=1,\dots,m$;
    \item {\rev \textbf{Slope cap:} $\left|\frac{\mathrm{d}f_{t,\mathbf{y}}(s)}{\mathrm{d}s}\right| \le \epsilon_t$,} 
    $\text{for a.e. } s\in (x^{(1)}_t,\,x^{(m)}_t)$.
\end{enumerate}
For any record $\mathbf{x}\in\mathcal{X}$ whose non-$t$ coordinates agree with those of the anchors and whose $t$-th coordinate lies in $\mathcal{I}_t$, we define the interpolated log-probability by $\ln f(\mathbf{y} \mid \mathbf{x}) := f_{t,\mathbf{y}}(x_t)$.
\end{definition}

\begin{proposition}[One-Dimensional Interpolation Validity]
\label{prop:1d-validity}
Let $f_{t,\mathbf{y}}$ be a one-dimensional interpolation rule as in Definition~\ref{def:1d-interpolation-rule}. Then for any two records $\mathbf{x},\mathbf{x}'\in\mathcal{X}$ that differ from the anchors only in the $t$-th coordinate, with
{\rev $x_t,\,x'_t\in \mathcal{I}_t$}, the interpolated log-probabilities satisfy
\begin{equation}
\bigl|\ln f(\mathbf{y} \mid \mathbf{x})-\ln f(\mathbf{y} \mid \mathbf{x}')\bigr|
\le
{\rev \epsilon_t |x_t-x'_t|}.
\end{equation}
Hence, the interpolated values satisfy the $(\epsilon_t,d_1)$-Lipschitz bound along coordinate $t$, regardless of whether $\mathbf{x}$ and $\mathbf{x}'$ lie in the same sub-interval
$
[x^{(r)}_t,\,x^{(r+1)}_t]
$
or in different sub-intervals. The detailed proof is given in {\bf Appendix~\ref{subsec:proof:prop:1d-unified}}.
\end{proposition}

\begin{theorem}
[Dimension-Wise Composition for Lipschitz Bound Condition~\cite{Qiu-USec2026}]
\label{thm:composition}
Let $f:\mathcal X\to\mathbb R$ be a mechanism that interpolates values in an $N$-dimensional space. Suppose that for each $t \in \{1, \dots, N\}$, $f$ satisfies $(\epsilon_t, d_1)$-Lipschitz bound when the input records differ only in the $t$th coordinate. If the parameters $\epsilon_1, \dots, \epsilon_N$ satisfy the following budget composition condition:
\begin{eqnarray}
\label{eq:budgetcompo1}
&& \textstyle \sum_{t=1}^{N} \epsilon_t^{\frac{p}{p-1}} \leq \epsilon^{\frac{p}{p-1}}, \quad \text{for } p > 1, 
\\  \label{eq:budgetcompo2}
\text{and} && 
\max_{t} \epsilon_t \leq \epsilon, \quad \text{for } p = 1,
\end{eqnarray}
then $f$ is $(\epsilon, d_p)$-Lipschitz continuous.
\end{theorem}

\DEL{
\begin{proposition}[One-Dimensional Interpolation Validity]
\label{prop:1d-unified}
Fix a coordinate~$t$ and an output $y \in \mathcal{Y}$.
Let $\mathbf{x}^{(1)}, \ldots, \mathbf{x}^{(m)} \in \hat{X}$ be anchor records
that differ only in their $t$-th coordinate, with $x^{(1)}_t < x^{(2)}_t < \cdots < x^{(m)}_t$. Suppose that consecutive anchor pairs satisfy the
$(\epsilon_t, d_1)$-Lipschitz bound:
\begin{equation}
  \bigl|\ln \mathcal{M}(\mathbf{y} | \mathbf{x}^{(r)})
       - \ln \mathcal{M}(\mathbf{y} | \mathbf{x}^{(r+1)})\bigr|
  \leqslant \epsilon_t \bigl|x^{(r)}_t - x^{(r+1)}_t\bigr|,
  ~r = 1,\ldots,m-1.
\end{equation}
Let $\phi:[x^{(1)}_t,\, x^{(m)}_t]\to\mathbb{R}$ be any absolutely
continuous interpolation function satisfying:
\begin{enumerate}[label=(\roman*)]
  \item $\phi\bigl(x^{(r)}_t\bigr)
  = \ln \mathcal{M}(\mathbf{y} | \mathbf{x}^{(r)})$
  for all $r=1,\ldots,m$;
  \item {\rd $\left|\frac{\mathrm{d}f(t)}{\mathrm{d}t}\right| \leqslant \epsilon_t$
  for a.e.\ $t\in\bigl(x^{(1)}_t,\, x^{(m)}_t\bigr)$.}
\end{enumerate}
Then for \emph{any} two records $\mathbf{x}, \mathbf{x}' \in \mathcal{X}$
that differ from the anchors only in the $t$-th coordinate,
with $x_{a,t},\, x_{t}
\in [x^{(1)}_t,\, x^{(m)}_t]$, regardless of whether $x_{a,t}$ and $x_{t}$ fall in
the same sub-interval
$[x^{(r)}_t,\, x^{(r+1)}_t]$
or in different sub-intervals, 
the interpolated log-perturbation probabilities satisfy
\begin{equation}   
  \bigl|\ln \hat{z}(y \mid \mathbf{x})
       - \ln \hat{z}(y \mid \mathbf{x}')\bigr|
  \leqslant \epsilon_t\,|x_{a,t} - x_{t}|,
\end{equation}
i.e., the $(\epsilon_t, d_1)$-Lipschitz bound holds between
$\mathbf{x}$ and $\mathbf{x}'$. The detailed proof can be found in Appendix \ref{subsec:proof:prop:1d-unified}. 
\end{proposition}
}

\begin{definition}[Two-Dimensional Interpolation]
\label{def:2d-interpolation}
Fix a non-leaf parent cell $v\subseteq \mathcal{X}$, with corner-anchor set $
\mathcal{X}_v=\{\mathbf{x}^{(\gamma)}:\gamma\in\{0,1\}^2\}$, 
$\mathbf{x}^{(\gamma)}=\mathbf{x}^{(0)}+\gamma\odot\Delta$, 
where $\Delta=(\Delta_1,\Delta_2)$. Assume that:
\begin{itemize}
    \item for each coordinate $t\in\{1,2\}$, the restriction of the interpolation to any axis-aligned segment in $v$ parallel to coordinate $t$ is $(\epsilon_t,d_1)$-valid in the sense of \textbf{Proposition~\ref{prop:1d-validity}};
    \item the coordinate-wise privacy budgets $\epsilon_1,\epsilon_2$ satisfy the composition conditions in Eq.~\eqref{eq:budgetcompo1} and Eq.~\eqref{eq:budgetcompo2}.
\end{itemize}

The local extension operator at $v$ is defined as
\[
\textstyle \mathrm{Ext}_v:\mathrm{Mech}(\mathcal{X}_v)\to \prod_{u\in \mathrm{Succ}(v)} \mathrm{Mech}(\mathcal{X}_u),
\]
where, for any parent mechanism $\mathcal{M}_v\in \mathrm{Mech}(\mathcal{X}_v)$, for all $\mathbf{x}\in \mathcal{X}_u$, $u\in \mathrm{Succ}(v)$, and $\mathbf{y}\in\mathcal{Y}$, the child mechanism $\mathcal{M}_{v\to u}$ is defined by the normalized 2-dimensional interpolation
\begin{equation}
\label{eq:tree-ext-child-2d}
\mathcal{M}_{v\to u}(\mathbf{y}\mid \mathbf{x})
=
{\rev
\frac{f_v(\mathbf{y}\mid \mathbf{x})}
{\sum_{\mathbf{y}'\in\mathcal{Y}}f_v(\mathbf{y}'\mid \mathbf{x})}
},
\end{equation}
where $\ln f_v(\mathbf{y}\mid \mathbf{x})
=
\sum_{\gamma\in\{0,1\}^2}
w_\gamma(\mathbf{x})\,
\ln \mathcal{M}_v\!\left(\mathbf{y}\mid \mathbf{x}^{(\gamma)}\right)$, with $w_\gamma(\mathbf{x})
=
\prod_{t=1}^2
\Bigl(
(1-\gamma_t)\lambda_t(\mathbf{x})
+
\gamma_t\bigl(1-\lambda_t(\mathbf{x})\bigr)
\Bigr)$ and $\lambda_t(\mathbf{x})
=
\frac{x^{(0)}_t+\Delta_t-x_t}{\Delta_t}$,
$(t\in\{1,2\})$. Thus, $\mathrm{Ext}_v$ maps the perturbation mechanism on the parent
corner-vertex set $\mathcal{X}_v$ to a collection of child mechanisms on
$\{\mathcal{X}_u\}_{u\in\mathrm{Succ}(v)}$ by treating $\mathcal{X}_v$ as the anchor set and
applying 2-dimensional interpolation inside each child cell.
\end{definition}

\begin{theorem}[Successor-level preservation under two-dimensional interpolation]
\label{thm:successor-preservation-2d}
Let $u$ and $w$ be two cells at the same extension level, and suppose
their mechanisms satisfy the coordinate-wise cross-cell Lipschitz bound:
for any $\mathbf{x}\in\mathcal{X}_u$, $\mathbf{x}'\in\mathcal{X}_w$, and $\mathbf{y}\in\mathcal{Y}$,
\[
\textstyle  \left| \ln \mathcal{M}_u(\mathbf{y}|\mathbf{x})-\ln \mathcal{M}_w(\mathbf{y}|\mathbf{x}')\right| \le \sum_{t=1}^{2}\epsilon_t |x_t-x'_t|.
\]
Let $u'\in\mathrm{Succ}(u)$ and $w'\in\mathrm{Succ}(w)$, and suppose that
the mechanisms on $u'$ and $w'$ are generated from $\mathcal{M}_u$ and $\mathcal{M}_w$,
respectively, using the two-dimensional interpolation operator in
Definition~\ref{def:2d-interpolation}. Then, for any
$\mathbf{x}\in\mathcal{X}_{u'}$, $\mathbf{x}'\in\mathcal{X}_{w'}$, and $\mathbf{y}\in\mathcal{Y}$,
the pre-normalization interpolants satisfy
\[
\textstyle \left|
\ln f_u(\mathbf{y}|\mathbf{x})-\ln f_w(\mathbf{y}|\mathbf{x}')
\right|
\le
\sum_{t=1}^{2}\epsilon_t |x_t-x'_t|.
\]
Consequently, if $\epsilon_1$ and $\epsilon_2$ satisfy the
dimension-wise composition condition in Theorem~\ref{thm:composition},
then
\[
\left|
\ln f_u(\mathbf{y}|\mathbf{x})-\ln f_w(\mathbf{y}|\mathbf{x}')
\right|
\le
\bar{\epsilon} d_p(\mathbf{x},\mathbf{x}'),
\]
where $\bar{\epsilon}$ is the composed pre-normalization budget. In
particular, with the calibrated choice $\bar{\epsilon}=\epsilon/2$, the
normalized successor mechanisms satisfy the target $(\epsilon,d)$-mDP bound. The detailed proof can be found in \textbf{Appendix~\ref{subsec:app:proof:thm:successor-preservation-2d}}. 
\end{theorem}

\DEL{
\begin{theorem}[Correctness of Two-Dimensional Interpolation]
\label{thm:N-interpolation-correct}
Let $u$ and $w$ be two cells. Assume that for any anchor pair
$\mathbf{x}\in\mathcal{X}_u$, $\mathbf{x}'\in\mathcal{X}_w$, and any
output $\mathbf{y}\in\mathcal{Y}$,
\begin{equation}
\label{eq:uw-dw-bound}
\bigl|
\ln \mathcal{M}_u(\mathbf{y}\mid \mathbf{z})
-
\ln \mathcal{M}_w(\mathbf{y}\mid \mathbf{z}')
\bigr|
\le
\sum_{t=1}^N \epsilon_t |x_{t}-x'_{t}|.
\end{equation}
Let $u'\in\mathrm{Succ}(u)$ and $w'\in\mathrm{Succ}(w)$, and let
their mechanisms be generated by the local extension operators
$\mathrm{Ext}_u$ and $\mathrm{Ext}_w$, respectively. Then, for any
$\mathbf{x}\in\mathcal{X}_{u'}$, $\mathbf{x}'\in\mathcal{X}_{w'}$, and
any $\mathbf{y}\in\mathcal{Y}$,
\begin{equation}
\label{eq:uw-dw-child}
\bigl|
\ln f_u(\mathbf{y}\mid \mathbf{x})
-
\ln f_w(\mathbf{y}\mid \mathbf{x}')
\bigr|
\le
\sum_{t=1}^N \epsilon_t |x_t-x_t'|.
\end{equation}
Consequently, if the budgets satisfy the composition condition in
Theorem~2, then
\begin{equation}
\label{eq:uw-global-child}
\bigl|
\ln f_u(\mathbf{y}\mid \mathbf{x})
-
\ln f_w(\mathbf{y}\mid \mathbf{x}')
\bigr|
\le
\epsilon\, d_p(\mathbf{x},\mathbf{x}').
\end{equation}
\end{theorem}}

\begin{figure}[t]
\begin{minipage}{0.40\textwidth}
  \subfigure{
\includegraphics[width=1.00\textwidth]{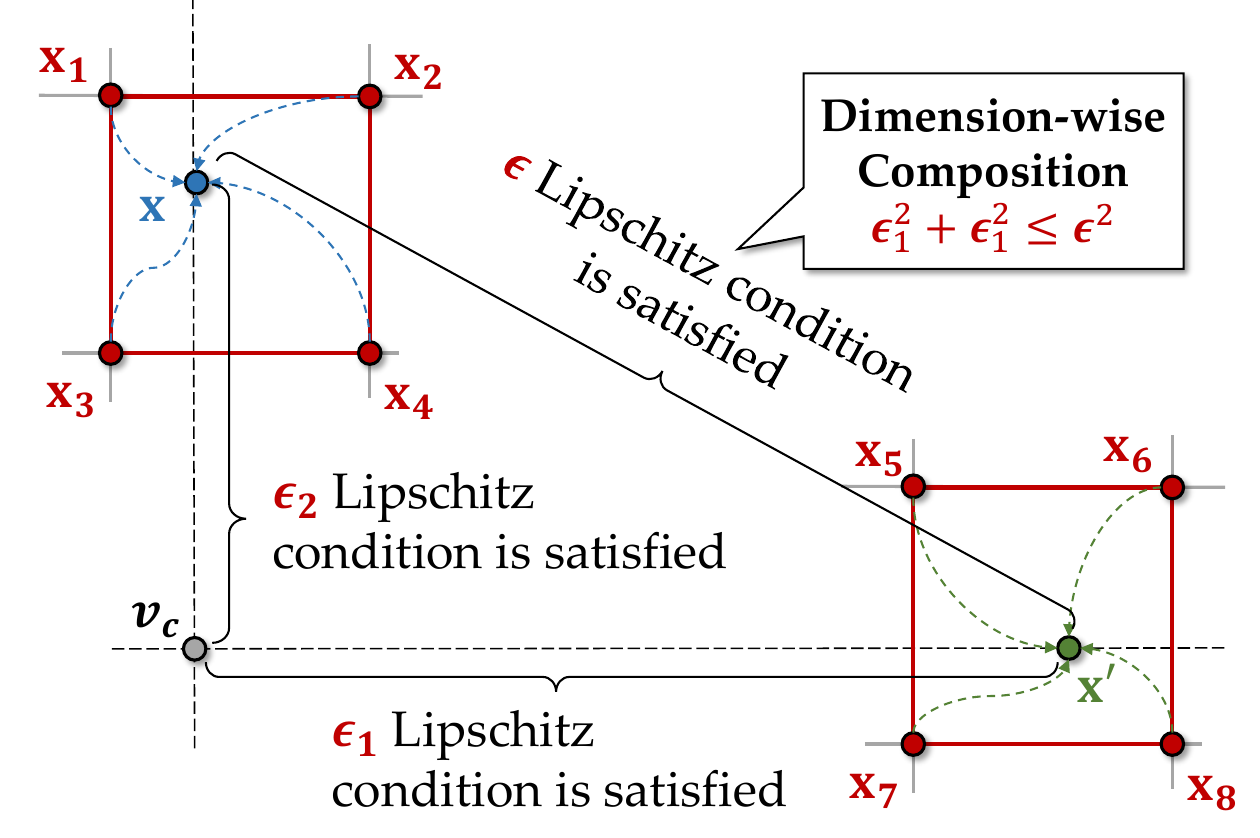}}
\end{minipage}
\caption{Illustration of one-dimensional interpolation validity and dimension-wise composition.}
\label{fig:mDPpreserv_eg}
\end{figure}

Fig.~\ref{fig:mDPpreserv_eg} illustrates the intuition behind
successor-level preservation in the tree-based extension. Consider two
records $\mathbf{x}$ and $\mathbf{x}'$ generated in different successor
cells. Instead of directly bounding the log-probability difference between
them through a single two-dimensional interpolation, we decompose the
comparison into coordinate-wise steps through an intermediate aligned point
$v_c$. Along each axis, the one-dimensional interpolation rule enforces the
corresponding coordinate-wise Lipschitz bound, with budget $\epsilon_1$ for
the first coordinate and $\epsilon_2$ for the second coordinate. By
Proposition~\ref{prop:1d-validity}, these one-dimensional interpolation
steps preserve the required $(\epsilon_t,d_1)$-Lipschitz property. Then,
Theorem~\ref{thm:composition} combines the two coordinate-wise bounds via
dimension-wise composition, e.g., $\epsilon_1^2+\epsilon_2^2\le \epsilon^2$
when $p=2$, to obtain the global $(\epsilon,d_p)$-Lipschitz guarantee.
Thus, the figure shows why the proposed coordinate-wise interpolation
preserves cross-cell mDP relations during recursive refinement.

\subsection{Extension Algorithm Design}
\label{subsec:algorithm}

Having defined the interpolation-based local extension operator in
Definition~\ref{def:2d-interpolation}, we now describe how the operator is
applied recursively over the extension tree. The procedure starts from source
cells whose corner records are equipped with seed mechanisms, and then
propagates these mechanisms from coarse cells to finer cells until all records
in $\mathcal{X}_{\mathrm{tar}}$ are covered.

\noindent \textbf{Seed record optimization.}
Before the recursive extension begins, we first construct a utility-aware
perturbation mechanism on the seed records $\mathcal{X}_{\mathrm{seed}}$,
which correspond to coarse grid corners and serve as the initial anchors for
extension. Instead of optimizing over all fine-grained target records, we
project the full-domain utility loss onto the seed records using interpolation
weights and solve a constrained linear program only on
$\mathcal{X}_{\mathrm{seed}}$. The resulting seed mechanism
$\mathcal{M}_{\mathrm{seed}}$ satisfies the seed-level mDP constraints and is used to
initialize the source mechanisms on the source cells. The detailed seed-level
cost construction and LP formulation are deferred to
Appendix~\ref{sec:seed-optimization}.

\smallskip
\noindent \textbf{Algorithm description.}
The tree-based extension procedure performs a breadth-first traversal over
the extension tree. Starting from the source nodes, whose mechanisms are
initialized from the seed mechanism, it processes cells level by level from
coarse to fine resolution. When a non-leaf cell $v$ is processed, its
mechanism on the corner-vertex set $\mathcal{X}_v$ is used as the input to
the local extension operator $\mathrm{Ext}_v$, which jointly constructs
mechanisms for all child cells in $\mathrm{Succ}(v)$.

At each refinement step, the operator first treats the parent-corner
distributions as anchors. It then assigns mechanisms to child-boundary
vertices using log-convex interpolation along the corresponding parent
edges. Once the boundary values are fixed, the operator extends to
child-interior vertices using coordinate-wise interpolation. The resulting
pre-normalization values are normalized over the output domain
$\mathcal{Y}$ to obtain valid perturbation distributions on the child record
sets. The procedure continues recursively until all reachable nodes have
been processed, after which the global mechanism on
$\mathcal{X}_{\mathrm{tar}}$ is assembled from the resulting local node
mechanisms. Detailed pseudocode is provided in
Appendix~\ref{sec:app:tree-extension-pseudocode}.

We next show how this construction satisfies
\textbf{\emph{(A1)}}--\textbf{\emph{(A3)}} and therefore induces a globally
valid $(\epsilon,d)$-mDP mechanism.

\subsubsection*{\textbf{(A1) Local mDP constraints.}}
For each non-leaf cell $v$, the local operator extends the parent mechanism
jointly to all child cells $\mathrm{Succ}(v)=\{v_1,\ldots,v_m\}$. The resulting
child mechanisms must satisfy $(\epsilon,d)$-mDP over their union, covering both
within-child and cross-sibling pairs: for any $i,j$,
$\mathcal{M}_{v\to v_i}(\cdot\mid \mathbf{x}) \overset{\epsilon}{\sim}
\mathcal{M}_{v\to v_j}(\cdot\mid \mathbf{x}')$, for all
$\mathbf{x}\in\mathcal{X}_{v_i}$ and
$\mathbf{x}'\in\mathcal{X}_{v_j}$. By Proposition~\ref{prop:1d-validity} and
Theorem~\ref{thm:composition}, the two-dimensional interpolation operator
yields a pre-normalization $(\bar{\epsilon},d_p)$-Lipschitz bound over the
child-cell union; after calibrated normalization with
$\bar{\epsilon}=\epsilon/2$, the resulting child mechanisms satisfy the
required $(\epsilon,d)$-mDP constraints.

\subsubsection*{\textbf{(A2) Overlap consistency.}}
In the grid-based tree, neighboring child cells may share boundary vertices,
so a target record may belong to multiple cell record sets. Overlap consistency
requires that every shared record receive the same perturbation distribution
regardless of which cell generates it. That is, for any two cells $u$ and $w$
with $\mathcal{X}_u\cap\mathcal{X}_w\neq\varnothing$, $\mathcal{M}_u(\mathbf{y}\mid \mathbf{x}) = \mathcal{M}_w(\mathbf{y}\mid \mathbf{x})$, $\forall \mathbf{x}\in\mathcal{X}_u\cap\mathcal{X}_w,\ \forall \mathbf{y}\in\mathcal{Y}$. This ensures that the induced global mechanism is well defined.

In our construction, overlaps occur only on shared boundaries. For boundary
vertices introduced along a parent edge, the mechanism is generated by
log-convex interpolation using only the two endpoint anchors of that edge.
Thus, sibling cells sharing the same boundary vertex compute its distribution
from the same two endpoints and assign identical values. If two cells meet
only at a corner, consistency is inherited from the shared corner value.
For neighboring cells with different parents, source-level compatibility and
deterministic edge-based boundary interpolation ensure consistency across
shared boundary records. Thus, \textbf{\emph{(A2)}} is satisfied by construction.

\noindent \textbf{Discussion: Order independence of transform-linear interpolation.}
The consistency argument above relies on deterministic interpolation rules for
shared boundary records. For transform-linear rules, such as log-convex
interpolation, we further have an order-independence property: repeated
interpolation toward a fixed endpoint yields the same final value regardless
of the order in which the interpolation parameters are applied. This property
is not needed for the mDP proof, but it simplifies implementation by ensuring
that coordinate-wise updates are insensitive to the chosen interpolation
order. The formal statement and proof are given in
\textbf{Appendix~\ref{sec:app:discussion:order-independence}}.{\rev \par\medskip These transform-linear assumptions are sufficient, but they are not a necessary characterization of order-independent interpolation. The essential requirement is that repeated update parameters admit a well-defined associative and commutative composition law, or an equivalent path-independent merge rule. Other interpolation operators may satisfy this algebraic property without being representable by the same functions $f$ and $\alpha$. Within the transform-linear family, strict monotonicity is useful because it makes the effective parameter unique; beyond that restricted family, necessity need not hold.} \looseness = -1

\subsubsection*{\textbf{(A3) Successor-level mDP preservation.}} The remaining challenge is to show that cross-region mDP relations are preserved under recursive refinement. In the tree-based setting, this means that if an mDP relation holds between two coarse regions, then the same relation should continue to hold between their descendant regions. This condition is subtle because direct multi-dimensional interpolation does not necessarily preserve cross-region mDP guarantees: simultaneous interpolation along multiple coordinates may produce log-probability changes that are not tightly controlled by the true $\ell_p$-distance between two records, even when the anchor mechanisms themselves satisfy mDP~\cite{Qiu-USec2026}. Our construction avoids this issue by decomposing each local extension into one-dimensional interpolation steps. Each step enforces a coordinate-wise Lipschitz bound in log-probability space, and the dimension-wise composition property in Theorem~\ref{thm:composition} converts these coordinate-wise bounds into an $\ell_p$-metric guarantee. The following theorem states the subtree-level invariant that establishes \textbf{\emph{(A3)}} for the tree-based construction.

\begin{theorem}[Subtree-level preservation of the tree-based extension]
\label{thm:subtree-preservation}
Consider the tree-based extension algorithm. Fix a node $v$ in the extension
tree, and define $\mathcal{X}_{\mathrm{sub}}(v):=\bigcup_{z\in \mathrm{Desc}(v)} \mathcal{X}_z$.
Assume that the corner anchors of $v$ satisfy the coordinate-wise Lipschitz
bounds in log space with pre-normalization budgets $\epsilon_1,\epsilon_2$,
and that every extension inside the subtree rooted at $v$ is performed by
the two-dimensional interpolation operator in
Definition~\ref{def:2d-interpolation}. Assume further that
$\epsilon_1,\epsilon_2$ satisfy the dimension-wise composition condition
with total pre-normalization budget $\bar{\epsilon}$ under the $\ell_p$
metric.

Then, for any two records
$\mathbf{x},\mathbf{x}'\in \mathcal{X}_{\mathrm{sub}}(v)$ and any output
$\mathbf{y}\in\mathcal{Y}$, the pre-normalization interpolants satisfy
\[
\left|
\ln f_v(\mathbf{y}\mid \mathbf{x})
-
\ln f_v(\mathbf{y}\mid \mathbf{x}')
\right|
\le
\bar{\epsilon} d_p(\mathbf{x},\mathbf{x}').
\]
Consequently, after normalization, the induced mechanism on
$\mathcal{X}_{\mathrm{sub}}(v)$ satisfies $2\bar{\epsilon}$-mDP. In
particular, choosing $\bar{\epsilon}=\epsilon/2$ yields the target
$(\epsilon,d)$-mDP guarantee. The detailed proof can be found in \textbf{Appendix~\ref{subsec:app:proof:thm:subtree-preservation}}. 
\end{theorem}

Theorem~\ref{thm:subtree-preservation} shows that the tree-based extension
operator preserves the Lipschitz invariant throughout every subtree. Together
with the local sibling constraints in \textbf{\emph{(A1)}} and the boundary
consistency in \textbf{\emph{(A2)}}, this implies that the recursive
tree-based construction satisfies all three framework requirements and
therefore induces a well-defined global mechanism satisfying $(\epsilon,d)$-mDP
on $\mathcal{X}_{\mathrm{tar}}$.

{\rev
\subsubsection*{\textbf{Complexity analysis.}}
A direct all-pairs optimization over the target domain requires
$O(|\mathcal{X}_{\mathrm{tar}}||\mathcal{Y}|)$ variables and up to
$O(|\mathcal{X}_{\mathrm{tar}}|^2|\mathcal{Y}|)$ mDP constraints. In contrast,
our framework optimizes only over the seed domain, using
$O(|\mathcal{X}_{\mathrm{seed}}||\mathcal{Y}|)$ variables and
$O(\Delta|\mathcal{X}_{\mathrm{seed}}||\mathcal{Y}|)$ neighboring-seed
constraints, where $\Delta$ is the maximum degree of the seed graph. For the
bounded-degree grid graphs considered here, this reduces to
$O(|\mathcal{X}_{\mathrm{seed}}||\mathcal{Y}|)$ constraints.

The remaining records are generated recursively without further global
optimization. Computing one extended distribution costs
$O(|\mathcal{Y}|)$, giving an overall extension cost of
$O(|\mathcal{X}_{\mathrm{tar}}||\mathcal{Y}|L)$ for $L$ extension levels and
$O(|\mathcal{X}_{\mathrm{tar}}||\mathcal{Y}|)$ space to store the final
mechanism. Hence, the computational savings are governed by the seed density
$|\mathcal{X}_{\mathrm{seed}}|/|\mathcal{X}_{\mathrm{tar}}|$: sparser seed sets
reduce the optimization size, while denser seed sets trade additional
computation for potentially improved utility.
}

{\rev In our experimental configuration, described in
Section~\ref{sec:experiments}, seed records constitute approximately $0.26\%$ of the target domain, corresponding to an approximately $378\times$ reduction in the number of optimization variables.}

\subsubsection*{\textbf{Discussion: Applicability beyond $\ell_p$ metrics.}}
{\rev The general graph-based framework is not restricted to $\ell_p$ metrics and can support both central and local instantiations of mDP, provided that the domain admits an appropriate extension graph and local operators satisfying requirements (A1)--(A3). Under suitable metrics, conventional central DP and $\epsilon$-LDP arise as special cases: the graph metric induced by dataset adjacency recovers central DP and its group-privacy extension, while the discrete metric
$d(\mathbf{x},\mathbf{x}')=\mathbf{1}_{[\mathbf{x}\neq\mathbf{x}']}$ recovers $\epsilon$-LDP. However, extension is generally less useful in these standard settings because their distance structures provide limited fine-grained proximity information for seed-based interpolation. Its usefulness for central DP therefore depends on whether the dataset space contains additional exploitable structure beyond basic adjacency.

The tree-based algorithm developed in this work focuses on $\ell_p$ metrics because they are widely used in the mDP literature and naturally support our hierarchical construction. Their coordinate-wise structure allows Theorem~\ref{thm:composition} to compose one-dimensional interpolation guarantees into a global $\ell_p$-metric guarantee. For other metrics, the general framework may still guide extension design, but different extension structures, local operators, algorithms, or proof techniques may be required to satisfy (A1)--(A3).}

{\rev
\subsubsection*{\textbf{Discussion: Limitations of tree-based extension.}} 
The current extension construction assumes that the secret domain admits a nested hierarchy in which each successor record is associated with a well-defined parent block. It may therefore be less suitable for non-hierarchical domains or arbitrary graphs containing cycles and multiple natural predecessors. Imposing a single-parent tree on such domains can discard important cross-branch connections and distort the underlying metric. For example, our tree construction should not be interpreted as a direct representation of the OSM road network: the tree organizes the multi-resolution spatial grid, whereas the road network is used only to evaluate travel-distance utility. Extending the construction to general graphs or multiple-parent structures remains future work.}
 
\section{Empirical Validation}
\label{sec:experiments}
In this section, we empirically evaluate the proposed tree-based extension framework on road-map datasets under different grid granularities and privacy budgets. 
We evaluate each method using three complementary metrics:

\begin{itemize}
    \item {\rev \emph{Utility loss} measures the expected distortion in travel-distance estimation caused by location perturbation. Let $\mathcal{Q}$ denote the set of spatial task locations, let $\pi_{\mathcal{X}}(\mathbf{x})$ and $\pi_{\mathcal{Q}}(\mathbf{q})$ denote the distributions of true and task locations, respectively, and let $d_\mathcal{G}(\cdot,\cdot)$ denote shortest-path distance on the OSM road network after nearest-road-node mapping. For a true location $\mathbf{x}$ and reported location $\mathbf{y}$, we define the expected utility loss as \looseness = -1
    \begin{equation}
    \label{eq:UL1}
    \textstyle c(\mathbf{x},\mathbf{y}) = \sum_{\mathbf{q}\in\mathcal{Q}}\pi_{\mathcal{Q}}(\mathbf{q})\left|d_\mathcal{G}(\mathbf{x},\mathbf{q})-d_\mathcal{G}(\mathbf{y},\mathbf{q})\right|.
    \end{equation}
    The overall utility loss of mechanism $\mathcal{M}$ is
    \begin{equation}
    \label{eq:UL2}
    \textstyle \operatorname{UL}(\mathcal{M})=\sum_{\mathbf{x}\in\mathcal{X}}\pi_{\mathcal{X}}(\mathbf{x})\sum_{\mathbf{y}\in\mathcal{Y}}\mathcal{M}(\mathbf{y}\mid\mathbf{x})c(\mathbf{x},\mathbf{y}).
    \end{equation}
    The loss is measured in kilometers, and lower values indicate more accurate travel-distance estimates. This metric is meaningful for location-based services such as spatial task allocation~\cite{Qiu-TMC2022}, where decisions depend on estimated travel distances from users to task locations.}


    \item \emph{Computation time}, {\rev reported in seconds, measures the end-to-end time required to construct the perturbation mechanism, including seed-level optimization and recursive extension. All experiments were conducted in MATLAB R2025b on a system equipped with an Apple M5 processor and 24\,GB of unified memory.}

    \item \emph{mDP violation rate} empirically checks whether the constructed mechanism satisfies the target mDP constraints. We compute the fraction of triples     $(\mathbf{x},\mathbf{x}',\mathbf{y})$ for which the log-probability difference exceeds the allowed privacy bound, i.e., $|\ln \mathcal{M}(\mathbf{y}\mid\mathbf{x})-\ln \mathcal{M}(\mathbf{y}\mid\mathbf{x}')| > \epsilon d(\mathbf{x},\mathbf{x}')+\tau$, where $\tau$ is a small numerical tolerance. 
\end{itemize}

{\rev In our experiments, the Euclidean privacy metric $d(\mathbf{x},\mathbf{x}')$ is measured in kilometers. Therefore, the privacy budget $\epsilon$ has unit $\mathrm{km}^{-1}$ (i.e., per kilometer), ensuring that $\epsilon d(\mathbf{x},\mathbf{x}')$ in the mDP guarantee is dimensionless. For a fixed distance, a larger $\epsilon$ corresponds to a weaker privacy guarantee.}

\smallskip
\noindent\textbf{Datasets.}
We conduct experiments on three road-network datasets: \emph{Rome}, \emph{London}, and \emph{New York}, with road-network data retrieved from OpenStreetMap (OSM)~\cite{openstreetmap}. {\rev The OSM road network and the grid serve distinct purposes in our experiments. The multi-resolution grid discretizes the geographic domain into secret locations and provides the hierarchical structure used by the extension algorithm. Each grid vertex is mapped to its nearest OSM road-network node, and the shortest-path travel distance between the corresponding road-network nodes is used to compute the utility loss. For each city, we overlay the multi-resolution grid described in Section~\ref{sec:tree} on the corresponding geographic region and retain the grid vertices in the target domain.} These datasets cover urban road networks with different spatial layouts and domain sizes, allowing us to evaluate the proposed extension framework across different geographic settings. Due to space limitations, we present the Rome results in the main text and report the London and New York results in Appendix~\ref{sec:appendix:addexp}, where similar trends are observed.

\smallskip
\noindent\textbf{Compared Methods.} We compare against representative methods from five categories. 
\begin{itemize}
    \item  {\rev \textbf{Predefined noise mechanisms}: We include two predefined distance-based mechanisms: the \emph{Exponential Mechanism (EM)}~\cite{chatzikokolakis2015constructing} and the \emph{planar Laplace mechanism}~\cite{Andres-CCS2013}. Both use the Euclidean ($\ell_2$) distance $d_2(\mathbf{x},\mathbf{y})=\lVert\mathbf{x}-\mathbf{y}\rVert_2$. EM uses the score function $u(\mathbf{x},\mathbf{y})=-d_2(\mathbf{x},\mathbf{y})$ and samples $\mathbf{y}\in\mathcal{Y}$ according to $\Pr[\mathcal{M}_{\mathrm{EM}}(\mathbf{x})=\mathbf{y}] = \frac{\exp\!\left(-\frac{\epsilon}{2}d_2(\mathbf{x},\mathbf{y})\right)}{\sum_{\mathbf{y}'\in\mathcal{Y}} \exp\!\left(-\frac{\epsilon}{2}d_2(\mathbf{x},\mathbf{y}')\right)}$. The planar Laplace mechanism samples a continuous output $\mathbf{y}\in\mathbb{R}^{2}$ with density $p_{\mathrm{Lap}}(\mathbf{y}\mid\mathbf{x}) = \frac{\epsilon^{2}}{2\pi}\exp\!\left(-\epsilon d_2(\mathbf{x},\mathbf{y})\right)$. 
    These mechanisms are computationally efficient because their perturbation distributions are fixed by Euclidean distance, but they are not optimized for the task-specific travel-distance utility considered in our experiments.} 
    
    \item \textbf{Hybrid methods}: we consider \emph{EM+RMP}, which applies \emph{Bayesian remapping}~\cite{chatzikokolakis2017efficient} after \emph{EM} to improve utility through post-processing while preserving the privacy guarantee of the first-stage mechanism. We also include \emph{COPT}~\cite{imola2022balancing}, which scales LP-based mDP design by constraining selected transition probabilities to follow the exponential mechanism and optimizing the remaining probabilities.  \looseness =-1
    
    \item \textbf{Optimization-based methods}: We include three optimization-based baselines. \emph{linear programming (LP)}~\cite{Bordenabe-CCS2014} directly optimizes the perturbation probabilities over the full target domain under the mDP constraints; {\rev \emph{LP-A} uses the same $12\times12$ initial-grid setting as all other evaluated methods in the main comparison and solves the optimization problem on the corresponding reduced domain, which improves scalability but may lose fine-grained utility.} In addition, we also include \emph{PAnDA}~\cite{Liu-CCS2025}, which improves scalability by using anchor-based approximation to reduce the number of optimization variables and constraints. 
        
    \item \textbf{Existing extension-based methods}: We compare with extension-based baselines, including \emph{McShane--Whitney extension (MW)}~\cite{borgs2018extend} and \emph{anchor-based interpolation optimization (AIPO)}~\cite{Qiu-USec2026}, which construct mechanisms on anchor records and extend them to non-anchor records without solving a full-domain optimization problem. \looseness =-1
        
    \item \textbf{Tree-based extension methods}: We evaluate three variants of our tree-based extension framework, \emph{MLaEt-A}, \emph{MLaEt-M}, and \emph{MLaEt-O}, which recursively extend mechanisms from coarse seed records to fine-grained target records using different local extension or optimization strategies.
\end{itemize}
{\rev In the main comparison, all evaluated methods, including LP-A, use the same $12\times12$ initial grid. In the table headings, \emph{depth} denotes the number of recursive grid-refinement levels applied from this initial grid to construct the target domain; thus, depth 2, 3, and 4 correspond to two, three, and four successive refinement levels, respectively. The accompanying value is the total number of participating location points included at the resulting target resolution, rather than the number of possible grid cells.}

\smallskip 
{\rev \noindent\textbf{Experimental protocol.} The main comparison is conducted on the Rome, London, and New York City datasets with privacy budgets $\epsilon\in\{0.5,1.0,1.5\}\,\mathrm{km}^{-1}$ and target depths 2, 3, and 4, using the common $12\times12$ initial grid described above. The seed-grid-density study instead varies the initial grid from $5\times5$ to $17\times17$. For every method--dataset--parameter configuration, we perform 50 independent repetitions.} 
{\rev A ``---'' indicates that the method could not complete mechanism construction within the {\rev 1800-second time limit}; consequently, the corresponding experimental results were unavailable.}

\begin{table*}[t]
\footnotesize 
\begin{tabular}{p{1.85cm} |  p{1.00cm} p{1.21cm} p{1.21cm}  p{1.21cm} | p{1.21cm} p{1.21cm} p{1.21cm} | p{1.21cm} p{1.21cm} p{1.21cm}} 
\toprule
\multicolumn{2}{c }{{\rev \# Participating Points}} & \multicolumn{3}{c }{483 (depth=2)} & \multicolumn{3}{c }{4087 (depth=3)}& \multicolumn{3}{c }{63865 (depth=4)}\\
\midrule
\multicolumn{2}{c }{Privacy budget (km$^{-1}$)} & $\epsilon = 0.5$& $\epsilon = 1.0$& $\epsilon = 1.5$ & $\epsilon = 0.5$& $\epsilon = 1.0$& $\epsilon = 1.5$& $\epsilon = 0.5$& $\epsilon = 1.0$& $\epsilon = 1.5$\\
\midrule
\multicolumn{11}{c }{Rome road map}\\
\hline
Pre-defined& EM & 13.27±0.95 & 10.96±0.92 & 10.31±0.88 & 12.95±0.92 & 10.64±0.89 & 9.99±0.85 & 12.79±0.90 & 10.48±0.88 & 9.83±0.84 \\
\cline{2-2}
Noise Distribution & Laplace  & 12.97±0.95 & 10.62±0.96 & 9.97±0.91 & 12.66±0.93 & 10.32±0.94 & 9.67±0.89 & 12.51±0.92 & 10.17±0.93 & 9.52±0.88 \\
\cline{1-2}
Hybrid  & EM+RMP & 13.14±0.90 & 10.88±0.89 & 10.27±0.86 & 12.83±0.87 & 10.57±0.86 & 9.95±0.84 & 12.67±0.86 & 10.41±0.85 & 9.79±0.82 \\
\cline{2-2}
Methods &  COPT  & ------------- & ------------- & ------------- & ------------- & ------------- & ------------- & ------------- & ------------- & ------------- \\
\cline{1-2}
Optimization &  LP  & ------------- & ------------- & ------------- & ------------- & ------------- & ------------- & ------------- & ------------- & ------------- \\
\cline{2-2}
Based & LP-A & 9.81±0.65 & 9.61±0.66 & 10.13±0.63 & 9.52±0.64 & 9.32±0.64 & 9.83±0.60 & 9.38±0.62 & 9.18±0.63 & 9.65±0.59 \\
\cline{2-2} 
Methods & PAnDA & ------------- & ------------- & ------------- & ------------- & ------------- & ------------- & ------------- & ------------- & ------------- \\
\hline
\cline{1-2}
\rowcolor{lightgray!20} Existing Extension &  AIPO  & 7.05±0.14 & 5.23±0.22 & 4.55±0.27 & 7.09±0.14 & 5.27±0.22 & 4.59±0.26 & 7.10±0.14 & 5.28±0.22 & 4.60±0.26 \\
\cline{2-2}
\rowcolor{lightgray!20} Based Methods & MW & 9.04±0.04 & 7.09±0.12 & 5.95±0.15 & 9.08±0.04 & 7.14±0.11 & 6.01±0.14 & 9.09±0.04 & 7.15±0.11 & 6.03±0.14 \\
\midrule
\rowcolor{lightgray!40} Tree-based &  MLaEt-A  & 7.05±0.14 & 5.23±0.22 & 4.55±0.27 & 7.09±0.14 & 5.27±0.22 & 4.59±0.26 & 7.10±0.14 & 5.28±0.22 & 4.60±0.26 \\
\cline{2-2}
\rowcolor{lightgray!40} extension & MLaEt-M & 6.94±0.14 & 5.12±0.22 & 4.43±0.26 & 6.97±0.14 & 5.16±0.22 & 4.46±0.26 & 6.97±0.14 & 5.16±0.22 & 4.46±0.26 \\
\rowcolor{lightgray!40} algorithms & MLaEt-O & 6.84±0.14 & 5.02±0.23 & 4.39±0.27 & 6.80±0.14 & 4.98±0.23 & 4.37±0.27 & 6.78±0.15 & 4.96±0.23 & 4.35±0.27 \\
\bottomrule
\end{tabular}
\centering
\caption{\centering {\rev Utility loss of different perturbation methods (kilometers) \\ Mean$\pm$1.96$\times$standard error; for the algorithm without getting the results, we label its results by ``--------''.}}
\label{tab:UL}
\end{table*}

\begin{table*}[t]
\footnotesize 
\begin{tabular}{p{1.85cm} |  p{1.00cm} p{1.21cm} p{1.21cm}  p{1.21cm} | p{1.21cm} p{1.21cm} p{1.21cm} | p{1.21cm} p{1.21cm} p{1.21cm}} 
\toprule
\multicolumn{2}{c }{{\rev \# Participating Points}} & \multicolumn{3}{c }{483 (depth=2)} & \multicolumn{3}{c }{4087 (depth=3)}& \multicolumn{3}{c }{63865 (depth=4)}\\
\midrule
\multicolumn{2}{c }{Privacy budget (km$^{-1}$)} & $\epsilon = 0.5$& $\epsilon = 1.0$& $\epsilon = 1.5$ & $\epsilon = 0.5$& $\epsilon = 1.0$& $\epsilon = 1.5$& $\epsilon = 0.5$& $\epsilon = 1.0$& $\epsilon = 1.5$\\
\midrule
\multicolumn{11}{c }{Rome road map}\\
\hline
Pre-defined& EM & 0.0005±0.0017 & 0.0001±0.0000 & 0.0001±0.0000 & 0.0004±0.0002 & 0.0004±0.0000 & 0.0004±0.0000 & 0.0014±0.0000 & 0.0014±0.0000 & 0.0013±0.0000 \\
\cline{2-2}
Noise Distribution & Laplace  & 0.0002±0.0002 & 0.0001±0.0001 & 0.0001±0.0000 & 0.0002±0.0000 & 0.0003±0.0004 & 0.0002±0.0002 & 0.0004±0.0001 & 0.0004±0.0001 & 0.0004±0.0001 \\
\cline{1-2}
Hybrid  & EM+RMP & 0.0010±0.0032 & 0.0002±0.0001 & 0.0002±0.0001 & 0.0007±0.0002 & 0.0006±0.0001 & 0.0006±0.0001 & 0.0023±0.0006 & 0.0021±0.0003 & 0.0021±0.0004 \\
\cline{2-2}
Methods &  COPT  & ------------- & ------------- & ------------- & ------------- & ------------- & ------------- & ------------- & ------------- & ------------- \\
\cline{1-2}
Optimization &  LP  & ------------- & ------------- & ------------- & ------------- & ------------- & ------------- & ------------- & ------------- & ------------- \\
\cline{2-2}
Based & LP-A & 10.35±0.84 & 9.75±0.84 & 10.24±1.51 & 10.35±0.84 & 9.75±0.84 & 10.24±1.51 & 10.35±0.84 & 9.75±0.84 & 10.24±1.51 \\
\cline{2-2} 
Methods & PAnDA & ------------- & ------------- & ------------- & ------------- & ------------- & ------------- & ------------- & ------------- & ------------- \\
\hline
\cline{1-2}
\rowcolor{lightgray!20} Existing Extension &  AIPO  & 8.27±0.29 & 7.28±0.02 & 7.10±0.02 & 9.56±0.29 & 8.57±0.02 & 8.39±0.02 & 280.36±0.99 & 277.77±1.04 & 277.91±0.54 \\
\cline{2-2}
\rowcolor{lightgray!20} Based Methods & MW & 8.10±0.17 & 7.29±0.11 & 7.11±0.07 & 8.23±0.15 & 7.41±0.10 & 7.23±0.06 & 28.70±2.02 & 27.53±1.48 & 27.23±1.58 \\
\midrule
\rowcolor{lightgray!40} Tree-based &  MLaEt-A  & 4.59±0.21 & 3.50±0.06 & 3.26±0.04 & 4.65±0.21 & 3.57±0.06 & 3.32±0.04 & 5.30±0.21 & 4.21±0.08 & 3.96±0.06 \\
\cline{2-2}
\rowcolor{lightgray!40} extension & MLaEt-M & 10.87±0.85 & 6.73±0.18 & 5.72±0.10 & 10.92±0.85 & 6.77±0.18 & 5.76±0.10 & 11.29±0.85 & 7.12±0.18 & 6.11±0.11 \\
\rowcolor{lightgray!40} algorithms & MLaEt-O & 5.40±0.38 & 4.30±0.28 & 4.04±0.28 & 15.13±4.40 & 13.93±4.41 & 10.32±4.07 & 1488.3±241.2 & 821.7±154.1 & 868.9±269.9 \\
\bottomrule
\end{tabular}
\centering
\caption{\centering {\rev Computation time of different perturbation methods (seconds) \\ Mean$\pm$1.96$\times$standard error; for the algorithm without getting the results, we label its results by ``--------''.}}
\label{tab:time}
\end{table*}

\smallskip
Next, we introduce the experimental results: 

\noindent\textbf{Utility Loss.}
Table~\ref{tab:UL} reports the utility loss on the Rome road-map dataset under three domain granularities and three privacy budgets. Across all tested settings, the proposed tree-based variants consistently outperform EM, Laplace, and EM+RMP, reducing utility loss by approximately 40\%--55\%. These results show that optimizing the mechanism over coarse seed records and then extending it can achieve substantially better utility than fixed distance-based perturbation rules.

The proposed variants also remain competitive with existing extension-based methods. In particular, \textsc{MLaEt}-A closely matches AIPO across different tree depths and privacy budgets, indicating that the recursive construction preserves most of the utility benefit of anchor-based interpolation while supporting hierarchical, multi-stage extension. 
{\rev \textsc{MLaEt}-O reduces utility loss by 4.5\% relative to AIPO, although at a higher computational cost.} 

{\rev The limited improvement of the proposed tree-based methods over AIPO is partly attributable to the regular two-dimensional grids considered here, on which a strong one-stage method can already perform the extension effectively. The main advantage of recursive extension is therefore not necessarily lower utility loss, but its ability to support compositional, multi-stage construction while preserving global mDP. This flexibility may be particularly beneficial for large, multidimensional, hierarchical, or irregular domains, which we leave for future evaluation.

LP-A optimizes only over the coarse initial grid, improving feasibility but yielding higher target-domain utility loss than the proposed methods across all tested settings. Its computational cost also grows rapidly with the record density, whereas our methods achieve lower utility loss and scale more effectively through fine-grained recursive extension.}

As expected, utility loss decreases as $\epsilon$ increases because weaker privacy constraints permit more utility-preserving perturbations. The relative ranking of the methods remains stable across domain granularities, indicating that the utility advantage of tree-based extension persists as the domain becomes finer.

\noindent \textbf{Computation time.} Table~\ref{tab:time} compares computation time on the Rome road-map dataset. The proposed tree-based variants, particularly MLaEt-A and MLaEt-M, maintain low and stable runtime across domain granularities and privacy budgets, remaining comparable to existing extension-based methods. Both finish within a few seconds as the domain grows from hundreds to thousands of records, indicating that recursive tree-based extension introduces only modest computational overhead. In contrast, LP-A requires substantially more time when computationally feasible, while several full optimization-based methods exceed the 1000-second time limit. MLaEt-O becomes more expensive at finer granularities, indicating that the choice of local extension strategy directly affects end-to-end efficiency.

Although predefined and hybrid mechanisms are faster, they incur substantially higher utility loss. Overall, the results demonstrate a favorable utility--efficiency trade-off: the proposed tree-based methods retain much of the efficiency of lightweight extension methods while achieving substantially better utility than predefined baselines and scaling more effectively than optimization-intensive approaches.

\begin{table*}[t]
\footnotesize 
\begin{tabular}{p{1.85cm} |  p{1.00cm} p{1.21cm} p{1.21cm}  p{1.21cm} | p{1.21cm} p{1.21cm} p{1.21cm} | p{1.21cm} p{1.21cm} p{1.21cm}} 
\toprule
\multicolumn{2}{c }{{\rev \# Participating Points}} & \multicolumn{3}{c }{483 (depth=2)} & \multicolumn{3}{c }{4087 (depth=3)}& \multicolumn{3}{c }{63865 (depth=4)}\\
\midrule
\multicolumn{2}{c }{Privacy budget (km$^{-1}$)} & $\epsilon = 0.5$& $\epsilon = 1.0$& $\epsilon = 1.5$ & $\epsilon = 0.5$& $\epsilon = 1.0$& $\epsilon = 1.5$& $\epsilon = 0.5$& $\epsilon = 1.0$& $\epsilon = 1.5$\\
\midrule
\multicolumn{11}{c }{Rome road map}\\
\hline
Pre-defined& EM & 0.00±0.00 & 0.00±0.00 & 0.00±0.00 & 0.00±0.00 & 0.00±0.00 & 0.00±0.00 & 0.00±0.00 & 0.00±0.00 & 0.00±0.00 \\
\cline{2-2}
Noise Distribution & Laplace  & 0.00±0.00 & 0.00±0.00 & 0.00±0.00 & 0.00±0.00 & 0.00±0.00 & 0.00±0.00 & 0.00±0.00 & 0.00±0.00 & 0.00±0.00 \\
\cline{1-2}
Hybrid  & EM+RMP & 0.00±0.00 & 0.00±0.00 & 0.00±0.00 & 0.00±0.00 & 0.00±0.00 & 0.00±0.00 & 0.00±0.00 & 0.00±0.00 & 0.00±0.00 \\
\cline{2-2}
Methods &  COPT  & ------------- & ------------- & ------------- & ------------- & ------------- & ------------- & ------------- & ------------- & ------------- \\
\cline{1-2}
Optimization &  LP  & ------------- & ------------- & ------------- & ------------- & ------------- & ------------- & ------------- & ------------- & ------------- \\
\cline{2-2}
Based & LP-A & 0.0465±0.0031 & 0.0245±0.0025 & 0.0119±0.0012 & 0.0465±0.0031 & 0.0245±0.0025 & 0.0119±0.0012 & 0.0465±0.0031 & 0.0245±0.0025 & 0.0119±0.0012 \\
\cline{2-2} 
Methods & PAnDA & ------------- & ------------- & ------------- & ------------- & ------------- & ------------- & ------------- & ------------- & ------------- \\
\hline
\cline{1-2}
\rowcolor{lightgray!20} Existing Extension &  AIPO  & 0.00±0.00 & 0.00±0.00 & 0.00±0.00 & 0.00±0.00 & 0.00±0.00 & 0.00±0.00 & 0.00±0.00 & 0.00±0.00 & 0.00±0.00 \\
\cline{2-2}
\rowcolor{lightgray!20} Based Methods & MW & 0.00±0.00 & 0.00±0.00 & 0.00±0.00 & 0.00±0.00 & 0.00±0.00 & 0.00±0.00 & 0.00±0.00 & 0.00±0.00 & 0.00±0.00 \\
\midrule
\rowcolor{lightgray!40} Tree-based &  MLaEt-A  & 0.00±0.00 & 0.00±0.00 & 0.00±0.00 & 0.00±0.00 & 0.00±0.00 & 0.00±0.00 & 0.00±0.00 & 0.00±0.00 & 0.00±0.00 \\
\cline{2-2}
\rowcolor{lightgray!40} extension & MLaEt-M & 0.00±0.00 & 0.00±0.00 & 0.00±0.00 & 0.00±0.00 & 0.00±0.00 & 0.00±0.00 & 0.00±0.00 & 0.00±0.00 & 0.00±0.00 \\
\rowcolor{lightgray!40} algorithms & MLaEt-O & 0.00±0.00 & 0.00±0.00 & 0.00±0.00 & 0.00±0.00 & 0.00±0.00 & 0.00±0.00 & 0.00±0.00 & 0.00±0.00 & 0.00±0.00 \\
\bottomrule
\end{tabular}
\centering
\caption{\centering {\rev mDP violation ratio of different perturbation methods \\ Mean$\pm$1.96$\times$standard error; for the algorithm without getting the results, we label its results by ``--------''.}}
\label{tab:violation}
\end{table*}


\noindent \textbf{mDP violation rate.} We finally evaluate empirical mDP violation rates on the Rome road-map dataset, with results reported in Table~\ref{tab:violation}. The proposed tree-based variants, MLaEt-A, MLaEt-M, and MLaEt-O, achieve zero empirical violations across all tested granularities and privacy budgets. This is consistent with the theoretical design of the framework: local extension, overlap consistency, and successor-level preservation together prevent privacy errors from accumulating during recursive refinement. In contrast, LP-A has small but nonzero violation ratios, suggesting that its coarse-domain approximation may leave residual privacy violations. EM, Laplace, EM+RMP, MW, AIPO, and the proposed tree-based variants all achieve zero violation in the reported Rome settings. Together with the utility and
runtime results, these findings show that the proposed tree-based framework
offers a favorable balance among utility, efficiency, and empirical privacy
validity.


\begin{table}[h]
\centering
\caption{Effect of the seed-grid density on the three \textsc{MLaEt} variants and the LP-A baseline (Rome).}
\label{tab:seed_ablation}
\resizebox{\linewidth}{!}{
\begin{tabular}{lrrrrrrr}
\toprule
Grid dimensions & $5{\times}5$ & $7{\times}7$ & $9{\times}9$ & $11{\times}11$ & $13{\times}13$ & $15{\times}15$ & $17{\times}17$ \\
\midrule
\multicolumn{8}{l}{\textbf{\textsc{MLaEt}-A}} \\
Loss (km) & 7.2902 & 6.4534 & 5.9906 & 5.7265 & 5.5845 & 5.4542 & 5.6767 \\
Time (seconds) & 0.9248 & 1.3292 & 2.0040 & 3.2189 & 5.0931 & 8.8058 & 14.1384 \\
mDP violation ratio & 0.00 & 0.00 & 0.00 & 0.00 & 0.00 & 0.00 & 0.00 \\
\addlinespace
\hline
\multicolumn{8}{l}{\textbf{\textsc{MLaEt}-M}} \\
Loss (km) & 7.3541 & 6.4185 & 5.9390 & 5.7302 & 5.7552 & 5.6641 & 5.5718 \\
Time (seconds) & 0.9122 & 1.2934 & 1.9375 & 3.0449 & 4.8038 & 8.0187 & 13.0760 \\
mDP violation ratio & 0.00 & 0.00 & 0.00 & 0.00 & 0.00 & 0.00 & 0.00 \\
\addlinespace
\hline
\multicolumn{8}{l}{\textbf{\textsc{MLaEt}-O}} \\
Loss (km) & 6.7389 & 6.1400 & 5.5556 & 5.4801 & 5.4539 & 5.1714 & 5.2085 \\
Time (seconds) & 1.2254 & 2.0426 & 3.3818 & 5.4406 & 8.5851 & 13.6740 & 21.4490 \\
mDP violation ratio & 0.00 & 0.00 & 0.00 & 0.00 & 0.00 & 0.00 & 0.00 \\
\addlinespace
\hline
\multicolumn{8}{l}{LP-A} \\
Loss (km) & 10.2718 & 10.3325 & 10.5798 & 9.6375 & 9.3371 & 9.2056 & 9.1385 \\
Time (seconds) & 0.0622 & 0.3433 & 1.6355 & 6.0257 & 17.1640 & 48.2081 & 108.6403 \\
mDP violation ratio & 0.0334 & 0.0382 & 0.0326 & 0.0289 & 0.0254 & 0.0232 & 0.0203 \\
\bottomrule
\end{tabular}}
\end{table}


{\rev 
\noindent \textbf{Effect of seed-grid density.} 
To examine how seed-set selection affects performance, we conduct an ablation study of the three \textsc{MLaEt} variants and the LP-A baseline on the Rome dataset. We vary the seed-grid resolution from $5\times5$ to $17\times17$ while holding all other experimental settings fixed. A denser grid provides more seed records for representing the target domain but also increases the size of the seed-level optimization problem. Table~\ref{tab:seed_ablation} reports the resulting utility loss, computation time, and empirical mDP violation ratio.}

\begin{wrapfigure}{r}{0.20\textwidth}
\begin{minipage}{0.20\textwidth}
\centering
\subfigure{
\includegraphics[width=1.00\textwidth]{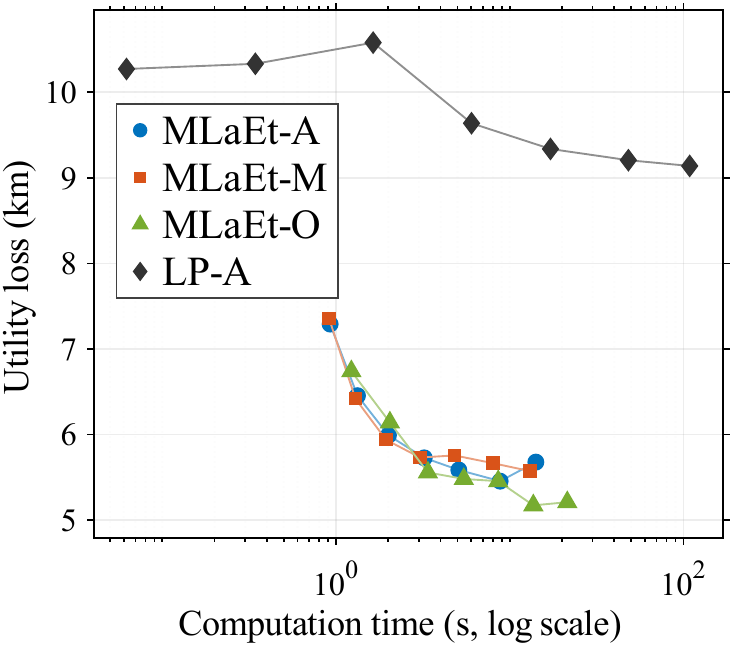}}
\caption{Utility vs. computation time (Rome).}
\label{fig:seed_tradeoff_scatter}
\end{minipage}
\end{wrapfigure}
The {\rev results show a clear utility--efficiency trade-off for the proposed variants. Increasing the seed-grid resolution generally reduces the utility loss of all three \textsc{MLaEt} variants, with most improvements occurring between $5\times5$ and $11\times11$--$13\times13$. Beyond this range, the gains diminish and are not strictly monotonic: \textsc{MLaEt}-A and \textsc{MLaEt}-O achieve their lowest losses at $15\times15$, while \textsc{MLaEt}-M reaches its lowest loss at $17\times17$. Computation time increases with seed density, while all three variants maintain zero empirical mDP violations across the tested settings. LP-A shows a less favorable trade-off: its runtime increases rapidly with seed density, exceeding 100 seconds at $17\times17$, while maintaining higher utility loss and nonzero empirical mDP violation ratios across all tested resolutions. Overall, the proposed \textsc{MLaEt} variants achieve a better utility--efficiency--privacy balance.}

{\rev Fig.~\ref{fig:seed_tradeoff_scatter} further illustrates the utility--efficiency trade-off by plotting utility loss against computation time on a logarithmic scale, where points closer to the lower-left indicate a better balance. Among the proposed variants, \textsc{MLaEt}-O generally achieves lower utility loss but requires more computation, while \textsc{MLaEt}-A and \textsc{MLaEt}-M exhibit similar utility and efficiency. LP-A shows higher utility loss across all tested seed densities, while its runtime grows much faster with seed density.}


\section{Conclusions and Discussions}
\label{sec:conclusion}

This paper presented a graph-based framework for extending mDP mechanisms from seed records to large, fine-grained target domains. We identified three {\rev necessary and sufficient} correctness requirements, local mDP, overlap consistency, and successor-level mDP preservation, and proved that they characterize a well-defined global mDP mechanism within our extension framework. We further instantiated the framework through a tree-based algorithm for multiresolution grid domains. Experiments on road-map datasets demonstrate a favorable utility--scalability--privacy trade-off. The proposed tree-based algorithms reduce utility loss compared with predefined and hybrid mechanisms, scale more effectively than computationally intensive optimization-based approaches, and maintain negligible empirical mDP violations. These results demonstrate the potential of extension-based design for scalable mDP mechanism construction over large fine-grained domains.

Several directions remain for future work. First, the framework can be extended
to higher-dimensional, non-grid, and irregular metric domains, where extension
graphs may be induced by meshes, road-network partitions, embeddings, or other
domain-specific structures. Second, adaptive seed selection could further
improve the utility--efficiency trade-off by placing more anchors in regions
with complex geometry, high prior probability, or large utility sensitivity.
Third, richer task-specific utility objectives could be incorporated, including
downstream prediction accuracy, task allocation quality, or application-aware
travel-cost metrics. Finally, it would be valuable to study tighter
normalization analysis, adaptive privacy-budget allocation across extension
levels, and deployment-oriented evaluations under realistic mobility traces and
stronger inference attacks. \looseness = -1

\section{Acknowledgements}
This paper was edited for grammar using ChatGPT. This research was partially supported by U.S. NSF grants CNS-2136948 and CNS-2313866.

\section*{Ethical Considerations}
We reviewed the venue's ethics guidance and conducted this study accordingly. The primary stakeholders include researchers and practitioners developing metric differential privacy (mDP) mechanisms, users whose locations may be protected, public data and infrastructure providers, and communities vulnerable to location disclosure. Our evaluation uses public OpenStreetMap road-network data for Rome, London, and New York. We conduct only offline perturbation simulations and utility evaluations; we collect no personal data, involve no human subjects, interact with no live platforms, and perform no re-identification. Therefore, the research process poses limited direct risk. The proposed framework may benefit researchers and practitioners by enabling scalable, utility-aware mDP mechanism design for large spatial domains. However, deployment risks remain. Inappropriate privacy budgets, poorly chosen distance metrics, violations of the framework's assumptions, or combination with auxiliary data could provide insufficient protection and enable profiling or tracking. To mitigate these risks, we document the framework's assumptions and calibration procedures, report empirical mDP violations, and recommend conservative privacy budgets and application-specific evaluation. Our artifacts contain only code, scripts, derived results, and public road-network data; no private trajectories or newly collected personal data are released.

We concluded that transparent publication offers a net ethical benefit because the study poses low direct risk while enabling community scrutiny and supporting more robust and auditable location-privacy mechanisms.

\section*{Open Science}
In compliance with the Open Science Policy, we provide all artifacts needed to evaluate the paper's core contributions in an anonymous, login-free repository:
\url{https://anonymous.4open.science/r/metric-DP-Extension-Algorithms-D0DF}.
The repository contains no personal identifiers and includes: 
\begin{itemize}
    \item \textbf{source code}, which implementations of the proposed methods and all evaluated baselines; 
    \item \textbf{datasets} used in the paper, with preprocessing scripts where applicable; and \item \textbf{documentation}, including environment information, reproduction commands, expected runtimes, and a minimal verification example. 
\end{itemize}
All artifacts required to reproduce the main experiments, figures, and tables are shared without licensing, privacy, safety, or responsible-disclosure restrictions.

\section*{Appendix}

\appendix
\setcounter{section}{0}

\section{Mathematical Notations}
\label{sec:notations}
\begin{table}[h]
\centering
\caption{Summary of main notation.}
\label{tab:notation}
\small
\setlength{\tabcolsep}{3pt}
\begin{tabular}{p{0.27\columnwidth} p{0.65\columnwidth}}
\toprule
\textbf{Symbol} & \textbf{Description} \\
\midrule

\multicolumn{2}{l}{\textbf{Mechanisms and mDP}} \\
$\mathcal{X}$ & Secret/input domain. \\
$\mathcal{Y}$ & Output/perturbed domain. \\
$\mathbf{x},\mathbf{x}'$ & Two secret records. \\
$\mathbf{y}$ & A perturbed output. \\
$\mathcal{M}$ & Randomized perturbation mechanism. \\
$\mathcal{M}(\mathbf{y}\mid \mathbf{x})$ & Probability of outputting $y$ given input $x$. \\
$d(\cdot,\cdot)$ & Distance metric on the secret domain. \\
$d_p(\cdot,\cdot)$ & $\ell_p$-metric for grid/location records. \\
$\epsilon$ & Total privacy budget. \\
$f_{\mathbf{y}}(\mathbf{x})$ & Log-probability function, $f_{\mathbf{y}}(\mathbf{x})=\ln \mathcal{M}(\mathbf{y}\mid \mathbf{x})$. \\

\midrule
\multicolumn{2}{l}{\textbf{Extension framework}} \\
$\mathcal{X}_{\mathrm{seed}}$ & Seed records with initially specified or optimized mechanisms. \\
$\mathcal{X}_{\mathrm{tar}}$ & Target records protected by the final mechanism. \\
$\mathcal{G}=(\mathcal{V},\mathcal{E},\{\mathcal{X}_v\})$ & Extension graph. \\
$\mathcal{V},\mathcal{E}$ & Node set and directed edge set. \\
$v,u,w$ & Nodes in the extension graph. \\
$\mathcal{X}_v$ & Record set associated with node $v$. \\
$\mathrm{Succ}(v)$ & Direct successors of $v$. \\
$\mathrm{Pred}(v)$ & Direct predecessors of $v$. \\
$\mathrm{Desc}(v)$ & Descendants of $v$, including $v$. \\
$\mathcal{M}_{v\to u}$ & Edge-local mechanism from $v$ to successor $u$. \\
$\mathrm{Ext}_v$ & Local extension operator at node $v$. \\
\midrule
\multicolumn{2}{l}{\textbf{Tree/grid extension}} \\
$L$ & Maximum extension level. \\
$r$ & Grid level, $r\in\{0,\ldots,L\}$. \\
$\mathcal{C}_r$ & Retained grid cells at level $r$. \\
$h_r$ & Side length of level-$r$ cells. \\
$N_r$ & extension factor between grid levels. \\
$\mathrm{corner}(v)$ & Corner vertices of cell/node $v$. \\
$\mathcal{X}_{\mathrm{sub}}(v)$ & Records in the subtree rooted at $v$. \\

\midrule
\multicolumn{2}{l}{\textbf{Interpolation}} \\
$\epsilon_t$ & Coordinate-wise privacy budget. \\
$f_{t,\mathbf{y}}$ & One-dimensional interpolation function in log-probability space. \\
$I_t$ & Interpolation interval along coordinate $t$. \\
$\lambda_t(\mathbf{x})$ & One-dimensional interpolation weight of point $x$. \\
$w_\gamma(\mathbf{x})$ & Multilinear interpolation weight for corner index $\gamma$. \\
$f_v(\mathbf{y}\mid \mathbf{x})$ & Pre-normalization interpolated value at node $v$. \\
$\gamma\in\{0,1\}^2$ & Corner index of a 2D cell. \\
\bottomrule
\end{tabular}
\end{table}

\section{Omitted Proofs}
\subsection{Proof of Lemma \ref{lem:descendant-closure}}

\begin{relemma}[Descendant closure of cross-node mDP]
Consider two possibly identical nodes $u,w\in\mathcal{V}$ whose associated mechanisms
satisfy
\begin{equation}
\mathcal{M}_u(\mathbf{x})
\stackrel{\epsilon}{\approx}
\mathcal{M}_w(\mathbf{x}'),
\qquad
\forall\mathbf{x}\in\mathcal{X}_u,\quad
\forall\mathbf{x}'\in\mathcal{X}_w.
\label{eq:lemma-premise}
\end{equation}
Suppose that \emph{(A1)--(A3)} hold along all subsequent extensions from
$u$ and $w$. Then, for every pair of
respective descendants $u''\in\mathrm{Desc}(u)$ and
$w''\in\mathrm{Desc}(w)$, their associated mechanisms satisfy
\begin{equation}
\mathcal{M}_{\mathrm{par}(u'')\to u''}(\mathbf{x})\ \stackrel{\epsilon}{\approx}\ \mathcal{M}_{\mathrm{par}(w'')\to w''}(\mathbf{x}'),
~\forall \mathbf{x}\in\mathcal{X}_{u''},\ \forall \mathbf{x}'\in\mathcal{X}_{w''}.
\label{eq:lem-conclusion}
\end{equation}
\end{relemma}

\begin{proof}
Fix descendants $u''\in\mathrm{Desc}(u)$ and
$w''\in\mathrm{Desc}(w)$, and choose directed paths
\begin{align}
u=u_0 \to u_1 \to \cdots \to u_{T_u}=u'',\\
w=w_0 \to w_1 \to \cdots \to w_{T_w}=w''.
\end{align}
Let $T=\max\{T_u,T_w\}$. If the paths have different lengths, pad the
shorter path with the identity extensions permitted by
\textbf{\emph{(A3)}}. Such steps leave the corresponding node and its
mechanism unchanged. We may therefore write
\begin{align}
u=u_0 \to u_1 \to \cdots \to u_T=u'',\\
w=w_0 \to w_1 \to \cdots \to w_T=w''.
\end{align}

We prove by induction on $t=0,\ldots,T$ that
\begin{equation}
\mathcal{M}_{u_t}(\mathbf{x})
\stackrel{\epsilon}{\approx}
\mathcal{M}_{w_t}(\mathbf{x}'),
\qquad
\forall\mathbf{x}\in\mathcal{X}_{u_t},\quad
\forall\mathbf{x}'\in\mathcal{X}_{w_t}.
\label{eq:induction-claim}
\end{equation}

For $t=0$, this is exactly the premise in
Eq.~\eqref{eq:lemma-premise}. Suppose that
Eq.~\eqref{eq:induction-claim} holds for some $t<T$. By construction,
$u_{t+1}$ is either a direct successor of $u_t$ or an identity extension
of $u_t$, and similarly for $w_{t+1}$. Applying
successor-level mDP preservation \textbf{\emph{(A3)}} gives
\begin{equation}
\mathcal{M}_{u_{t+1}}(\mathbf{x})
\stackrel{\epsilon}{\approx}
\mathcal{M}_{w_{t+1}}(\mathbf{x}'),
\qquad
\forall\mathbf{x}\in\mathcal{X}_{u_{t+1}},\quad
\forall\mathbf{x}'\in\mathcal{X}_{w_{t+1}},
\end{equation}
which establishes Eq.~\eqref{eq:induction-claim} for $t+1$.

Applying the induction through $t=T$ yields
\[
\mathcal{M}_{u''}(\mathbf{x})
\stackrel{\epsilon}{\approx}
\mathcal{M}_{w''}(\mathbf{x}'),
\qquad
\forall\mathbf{x}\in\mathcal{X}_{u''},\quad
\forall\mathbf{x}'\in\mathcal{X}_{w''},
\]
which is Eq.~\eqref{eq:lem-conclusion}. If either descendant has multiple
incoming edges, overlap consistency \textbf{\emph{(A2)}} ensures that
its associated node mechanism is independent of the selected incoming
edge.
\end{proof}

\subsection{Proof of Theorem \ref{thm:global-mdp}}
\label{subsec:proof:thm:global-mdp}
\begin{retheorem}[Global $(\epsilon,d)$-mDP guarantee via extension algorithms]
Consider a metric space $(\mathcal{X},d)$, a target set $\mathcal{X}_{\mathrm{tar}}\subseteq\mathcal{X}$, and a finite extension graph
$\mathcal{G}=(\mathcal{V},\mathcal{E}, \{\mathcal{X}_v\}_{v\in\mathcal{V}})$ whose local regions cover $\mathcal{X}_{\mathrm{tar}}$ and whose relevant nodes are reachable from source nodes
$\mathcal{V}_0=\{v\in\mathcal{V}:\mathrm{Pred}(v)=\varnothing\}$. Suppose the source regions are equipped with initial mechanisms $\{\mathcal{M}_v\}_{v\in\mathcal{V}_0}$ that agree with the prescribed seed mechanism $\mathcal{M}_{\mathrm{seed}}$ on $\mathcal{X}_{\mathrm{seed}}$, are consistent on overlapping source regions, and jointly satisfy $(\epsilon,d)$-mDP over $\bigcup_{v\in\mathcal{V}_0}\mathcal{X}_v$.

Assume that the extension algorithm propagates these source mechanisms through $\mathcal{G}$ and produces edge-local mechanisms $\{\mathcal{M}_{v\to u}\}_{(v,u)\in\mathcal{E}}$. The induced global mechanism $\mathcal{M}$ in Definition~\ref{def:globalmech} is well defined and satisfies $(\epsilon,d)$-mDP on $\mathcal{X}_{\mathrm{tar}}$ if and only if the source and edge-local mechanisms satisfy \textbf{\emph{(A1)--(A3)}}.
\end{retheorem}

\begin{proof}[Proof of Theorem~\ref{thm:global-mdp}]
We prove the sufficiency and necessity directions separately.

\medskip
\noindent\textbf{Sufficiency.}
Suppose that the source and edge-local mechanisms satisfy
\textbf{\emph{(A1)--(A3)}}. We first show that they induce a well-defined
global mechanism. Let $\mathbf{x}\in\mathcal{X}_{\mathrm{tar}}$ and
suppose that two local mechanisms assign distributions to $\mathbf{x}$.
Each assignment may come from a source mechanism $\mathcal{M}_v$, where
$v\in\mathcal{V}_0$, or an edge-local mechanism
$\mathcal{M}_{p\to v}$, where $(p,v)\in\mathcal{E}$. If both assignments
come from source mechanisms, they agree by source-level consistency.
Otherwise, overlap consistency \textbf{\emph{(A2)}} ensures that they
agree on $\mathbf{x}$. Therefore, the distribution
$\mathcal{M}(\mathbf{x})$ is independent of the selected local
assignment, and the induced global mechanism is well defined.

We next establish the global $(\epsilon,d)$-mDP guarantee. Fix arbitrary
$\mathbf{x},\mathbf{x}'\in\mathcal{X}_{\mathrm{tar}}$. By target
coverage, choose nodes $v,z\in\mathcal{V}$ whose associated mechanisms
define their distributions, with
$\mathbf{x}\in\mathcal{X}_v$ and
$\mathbf{x}'\in\mathcal{X}_z$. If both records belong to the same local
region, the required relation follows directly from
\textbf{\emph{(A1)}} or, for a source region, from the source-level mDP
assumption.

Otherwise, by source reachability, choose source nodes
$s,t\in\mathcal{V}_0$ such that
$v\in\mathrm{Desc}(s)$ and $z\in\mathrm{Desc}(t)$, where each node is
regarded as its own depth-zero descendant. Because the source mechanisms
jointly satisfy $(\epsilon,d)$-mDP,
\[
\mathcal{M}_s(\tilde{\mathbf{x}})
\stackrel{\epsilon}{\approx}
\mathcal{M}_t(\tilde{\mathbf{x}}'),
\qquad
\forall\tilde{\mathbf{x}}\in\mathcal{X}_s,\quad
\forall\tilde{\mathbf{x}}'\in\mathcal{X}_t.
\]
Applying Lemma~\ref{lem:descendant-closure} to the source nodes $s$ and
$t$ propagates this relation to their respective descendants $v$ and
$z$. Therefore,
\[
\mathcal{M}_v(\mathbf{x})
\stackrel{\epsilon}{\approx}
\mathcal{M}_z(\mathbf{x}').
\]

By the definition and well-definedness of the induced global mechanism,
the node-level distributions coincide with
$\mathcal{M}(\mathbf{x})$ and $\mathcal{M}(\mathbf{x}')$, respectively.
Hence,
\[
\mathcal{M}(\mathbf{x})
\stackrel{\epsilon}{\approx}
\mathcal{M}(\mathbf{x}').
\]
Because $\mathbf{x}$ and $\mathbf{x}'$ were arbitrary, $\mathcal{M}$
satisfies $(\epsilon,d)$-mDP on $\mathcal{X}_{\mathrm{tar}}$.

\medskip
\noindent\textbf{Necessity.}
Conversely, suppose that the source and edge-local mechanisms jointly
induce a well-defined global mechanism $\mathcal{M}$ satisfying
$(\epsilon,d)$-mDP on $\mathcal{X}_{\mathrm{tar}}$, with every local
mechanism coinciding with the restriction of $\mathcal{M}$ to its
corresponding region. We show that
\textbf{\emph{(A1)--(A3)}} must hold.

First, consider any edge-local mechanism
$\mathcal{M}_{v\to u}$ and any
$\mathbf{x},\mathbf{x}'\in\mathcal{X}_u$. Because
$\mathcal{M}_{v\to u}$ is the restriction of $\mathcal{M}$ to
$\mathcal{X}_u$, global mDP gives
\[
\mathcal{M}_{v\to u}(\mathbf{x})
\stackrel{\epsilon}{\approx}
\mathcal{M}_{v\to u}(\mathbf{x}').
\]
Thus, local mDP validity \textbf{\emph{(A1)}} is necessary.

Second, suppose that two local mechanisms are defined on a shared record
$\mathbf{x}$. Since both coincide with the same global mechanism at
$\mathbf{x}$, they must assign the identical distribution
$\mathcal{M}(\mathbf{x})$. Therefore, overlap consistency
\textbf{\emph{(A2)}} is necessary.

Finally, consider any two successor regions
$\mathcal{X}_{u'}$ and $\mathcal{X}_{w'}$ appearing in
\textbf{\emph{(A3)}}. For every
$\mathbf{x}\in\mathcal{X}_{u'}$ and
$\mathbf{x}'\in\mathcal{X}_{w'}$, global mDP implies
\[
\mathcal{M}(\mathbf{x})
\stackrel{\epsilon}{\approx}
\mathcal{M}(\mathbf{x}').
\]
Because the successor mechanisms are restrictions of $\mathcal{M}$,
this is equivalent to
\[
\mathcal{M}_{u\to u'}(\mathbf{x})
\stackrel{\epsilon}{\approx}
\mathcal{M}_{w\to w'}(\mathbf{x}').
\]
Thus, the relation required by successor-level mDP preservation holds
whenever its premise holds, establishing the necessity of
\textbf{\emph{(A3)}}.

Therefore, within the proposed extension framework,
\textbf{\emph{(A1)--(A3)}} are both necessary and sufficient  for the
local mechanisms to induce a well-defined global $(\epsilon,d)$-mDP
mechanism on $\mathcal{X}_{\mathrm{tar}}$.
\end{proof}

\DEL{
\subsection{Proof of Lemma \ref{lem:iff_1d_multi_interval}}
\begin{proof}
We prove both directions.

\paragraph{(ii) $\Rightarrow$ (i) (Sufficiency).}
Since $f$ is absolutely continuous on $[t_1,t_m]$, for any $s,t\in[t_1,t_m]$ we have the fundamental theorem of calculus:
\begin{equation}
f(t)-f(s)=\int_s^t \frac{\mathrm{d}f(u)}{\mathrm{d}u}\,du.
\end{equation}
Taking absolute values and applying the triangle inequality yields
\begin{equation}
|f(t)-f(s)|
=
\left|\int_s^t \frac{\mathrm{d}f(u)}{\mathrm{d}u}\,du\right|
\le
\int_s^t \left|\frac{\mathrm{d}f(u)}{\mathrm{d}u}\right|\,du.
\end{equation}
Using the slope cap $\left|\frac{\mathrm{d}f(u)}{\mathrm{d}u}\right|\le \epsilon_t$ for a.e.\ $u$, we obtain
\begin{equation}
\int_s^t \left|\frac{\mathrm{d}f(u)}{\mathrm{d}u}\right|\,du \le \int_s^t \epsilon_t\,du = \epsilon_t |t-s|.
\end{equation}
Hence $|f(t)-f(s)|\le \epsilon_t |t-s|$ for all $s,t$, i.e., $f$ is $\epsilon_t$-Lipschitz.

\paragraph{(i) $\Rightarrow$ (ii) (Necessity).}
Assume $f$ is $\epsilon_t$-Lipschitz on $[t_1,t_m]$.
By Rademacher's theorem, every Lipschitz function on an interval is differentiable almost everywhere;
thus $\frac{\mathrm{d}f(t)}{\mathrm{d}t}$ exists for a.e.\ $t\in(t_1,t_m)$.
Fix any point $t$ where $f$ is differentiable. Then
\begin{equation}
\left|\frac{\mathrm{d}f(t)}{\mathrm{d}t}\right|
=
\lim_{h\to 0}\left|\frac{f(t+h)-f(t)}{h}\right|
\le
\lim_{h\to 0}\frac{\epsilon_t |h|}{|h|}
=
\epsilon_t,
\end{equation}
where we used the Lipschitz inequality $|f(t+h)-f(t)|\le \epsilon_t |h|$ in the middle step.
Therefore $\left|\frac{\mathrm{d}f(t)}{\mathrm{d}t}\right|\le \epsilon_t$ at every differentiability point, hence for almost every $t\in(t_1,t_m)$.

\end{proof}
}

\subsection{Proof of Proposition \ref{prop:1d-unified}}
\label{subsec:proof:prop:1d-unified}

\begin{lemma}[One-Dimensional mDP-Preserving Interpolation]
\label{lem:1d-iff-single}
Fix an output {\rev $\mathbf{y} \in \mathcal{Y}$} and a coordinate~$t$.
Consider a one-dimensional interval $[a,b] \subset \mathbb{R}$ along coordinate~$t$
(with all other coordinates held fixed),
and let two anchor records $\mathbf{x}_i, \mathbf{x}_{i'}$ satisfy
$x_{i,t} = a$, $x_{i',t} = b$.
Let {\rev $g:[a,b]\to\mathbb{R}$} be an interpolation curve satisfying
the endpoint constraints
\begin{equation}
  {\rev g(a) = \ln \mathcal{M}(\mathbf{y} \mid \mathbf{x}_i), \qquad
  g(b) = \ln \mathcal{M}(\mathbf{y} \mid \mathbf{x}_{i'}).}
\end{equation}
Assume {\rev $g$} is absolutely continuous on $[a,b]$.\footnote{%
Absolute continuity is a mild regularity condition satisfied by all
standard interpolation schemes (e.g., linear, log-convex, spline).
It guarantees the fundamental theorem of calculus
and rules out pathological (e.g., Cantor-function-type) interpolants.}
Then the following are equivalent:
\begin{enumerate}[label=(\roman*)]
  \item {\rev $g$ is $\epsilon_t$-Lipschitz} on $[a,b]$, i.e.,
  \begin{equation}
    {\rev |g(v) - g(u)| \leqslant \epsilon_t \,|v - u|,
    \qquad \forall\, u,v \in [a,b].}
  \end{equation}
  \item {\rev $g$ is differentiable} almost everywhere on $(a,b)$
  and its derivative satisfies
  \begin{equation}
    {\rev \left|\frac{\mathrm{d}g(u)}{\mathrm{d}u}\right| \leqslant \epsilon_t
    \quad \text{for almost every } u \in (a,b).}
  \end{equation}
\end{enumerate}
Consequently, defining {\rev $\ln f(\mathbf{y} \mid \mathbf{x}(u)) := g(u)$}
yields a one-dimensional $(\epsilon_t, d_1)$-Lipschitz log-probability
(and hence one-dimensional mDP along coordinate~$t$) on this interval.
\end{lemma}

\begin{proof}
We prove both directions.

\medskip\noindent
\emph{(ii)\,$\Rightarrow$\,(i) (Sufficiency).}
Since {\rev $g$} is absolutely continuous on $[a,b]$,
the fundamental theorem of calculus gives,
for any {\rev $u,v\in[a,b]$},
\begin{equation}
  {\rev g(v) - g(u) = \int_u^v \frac{\mathrm{d}g(w)}{\mathrm{d}w}\,dw.}
\end{equation}
Taking absolute values and using the derivative bound:
{\rev 
\begin{eqnarray}
|g(v)-g(u)|&=& \biggl|\int_u^v \frac{\mathrm{d}g(w)}{\mathrm{d}w}\,dw\biggr| \\
&\leqslant& \int_{\min\{u,v\}}^{\max\{u,v\}} \left|\frac{\mathrm{d}g(w)}{\mathrm{d}w}\right|\,dw
\\
&\leqslant& \int_{\min\{u,v\}}^{\max\{u,v\}} \epsilon_t\,dw
\\ 
&=& \epsilon_t\,|v-u|.
\end{eqnarray}}
\medskip\noindent
\emph{(i)\,$\Rightarrow$\,(ii) (Necessity).}
If {\rev $g$} is $\epsilon_t$-Lipschitz,
then by Rademacher's theorem {\rev $g$} is differentiable
almost everywhere on $(a,b)$.
At every point of differentiability,
\begin{equation}
  {\rev \left|\frac{\mathrm{d}g(u)}{\mathrm{d}u}\right|
  = \lim_{h\to 0}\biggl|\frac{g(u+h)-g(u)}{h}\biggr|
  \leqslant \lim_{h\to 0}\frac{\epsilon_t|h|}{|h|}
  = \epsilon_t,}
\end{equation}
where we used {\rev $|g(u+h)-g(u)|\leqslant\epsilon_t|h|$}.
Hence {\rev $\left|\frac{\mathrm{d}g(u)}{\mathrm{d}u}\right|\leqslant\epsilon_t$} for a.e.\ {\rev $u\in(a,b)$}.
\end{proof}
\begin{remark}[Why anchors and multiple intervals do not change the IFF]
The equivalence in Lemma~\ref{lem:1d-iff-single} is a
\emph{global} characterization on $[x^{(1)}_t,\, x^{(m)}_t]$
and does not depend on how {\rev $g$} is constructed
(e.g., via different formulas on different sub-intervals
$[x^{(r)}_t, x^{(r+1)}_t]$).
The anchor constraints
{\rev $g(x^{(r)}_t) = \ln \mathcal{M}(\mathbf{y} \mid \mathbf{x}^{(r)})$}
merely restrict the feasible set of interpolants;
they do not affect the Lipschitz--derivative equivalence.
\end{remark}

We now generalize to a multi-anchor setting, which is the
situation encountered in the tree-based extension:
each one-dimensional line through the grid passes through multiple anchor
records, and the interpolation may use different formulas on
different sub-intervals.

\begin{reproposition}[One-Dimensional Interpolation Validity --- Unified Intra- and Across-Interval]
\label{prop:1d-unified}
Fix a coordinate~$t$ and an output {\rev $\mathbf{y} \in \mathcal{Y}$}.
Let $\mathbf{x}^{(1)}, \ldots, \mathbf{x}^{(m)} \in \hat{X}$ be anchor records
that differ only in their $t$-th coordinate, with
\begin{equation}
  x^{(1)}_t < x^{(2)}_t < \cdots < x^{(m)}_t.
\end{equation}
Suppose that consecutive anchor pairs satisfy the
$(\epsilon_t, d_1)$-Lipschitz bound:
\begin{equation}
  \bigl|\ln \mathcal{M}(\mathbf{y} \mid \mathbf{x}^{(r)})
       - \ln \mathcal{M}(\mathbf{y} \mid \mathbf{x}^{(r+1)})\bigr|
  \leqslant \epsilon_t \,\bigl|x^{(r)}_t - x^{(r+1)}_t\bigr|,
  ~r = 1,\ldots,m-1.
\end{equation}
Let {\rev $g:[x^{(1)}_t,\, x^{(m)}_t]\to\mathbb{R}$} be any absolutely
continuous interpolation function satisfying:
\begin{enumerate}[label=(\roman*)]
  \item \textbf{Anchor consistency:}
  {\rev $g\bigl(x^{(r)}_t\bigr)
  = \ln \mathcal{M}(\mathbf{y} \mid \mathbf{x}^{(r)})$}
  for all $r=1,\ldots,m$;
  \item \textbf{Slope cap:}
  {\rev $\left|\frac{\mathrm{d}g(u)}{\mathrm{d}u}\right| \leqslant \epsilon_t$
  for a.e.\ $u\in\bigl(x^{(1)}_t,\, x^{(m)}_t\bigr)$.}
\end{enumerate}
Then for \emph{any} two records $\mathbf{x}, \mathbf{x}' \in \mathcal{X}$
that differ from the anchors only in the $t$-th coordinate,
with {\rev $x_t,\, x'_t
\in [x^{(1)}_t,\, x^{(m)}_t]$}
--- regardless of whether {\rev $x_t$ and $x'_t$} fall in
the same sub-interval
$[x^{(r)}_t,\, x^{(r+1)}_t]$
or in different sub-intervals ---
the interpolated log-perturbation probabilities satisfy
\begin{equation}
  {\rev \bigl|\ln f(\mathbf{y} \mid \mathbf{x})
       - \ln f(\mathbf{y} \mid \mathbf{x}')\bigr|
  \leqslant \epsilon_t\,|x_t - x'_t|,}
\end{equation}
i.e., the $(\epsilon_t, d_1)$-Lipschitz bound holds between
$\mathbf{x}$ and $\mathbf{x}'$.
\end{reproposition}

\begin{proof}
By condition~(ii), {\rev $\left|\frac{\mathrm{d}g(u)}{\mathrm{d}u}\right| \leqslant \epsilon_t$ for
a.e.\ $u \in (x^{(1)}_t,\, x^{(m)}_t)$.}
Applying the (ii)\,$\Rightarrow$\,(i) direction of
Lemma~\ref{lem:1d-iff-single} to the full interval
$[x^{(1)}_t,\, x^{(m)}_t]$, we conclude that {\rev $g$} is globally
$\epsilon_t$-Lipschitz.
The result follows by evaluating the Lipschitz inequality at
{\rev $u = x_t$ and $v = x'_t$}.
\end{proof}

\begin{remark}
Proposition~\ref{prop:1d-unified} simultaneously subsumes both the
\emph{intra-interval} case
(where {\rev $x_t$ and $x'_t$} lie in the same sub-interval,
corresponding to the original Proposition~1)
and the \emph{across-interval} case
(where they lie in different sub-intervals,
corresponding to the original Proposition~2).
No case distinction is needed:
the global Lipschitz property established via
Lemma~\ref{lem:1d-iff-single} applies uniformly to all pairs.
\end{remark}

\begin{remark}[Why anchors and multiple intervals do not change the IFF]
The equivalence in Lemma~\ref{lem:1d-iff-single} is a
\emph{global} characterization on $[x^{(1)}_t,\, x^{(m)}_t]$
and does not depend on how $f$ is constructed
(e.g., via different formulas on different sub-intervals
$[x^{(r)}_t, x^{(r+1)}_t]$).
The anchor constraints
$f(x^{(r)}_t) = \ln \mathcal{M}(\mathbf{y} \mid \mathbf{x}^{(r)})$
merely restrict the feasible set of interpolants;
they do not affect the Lipschitz--derivative equivalence.
\end{remark}

\subsection{Proof of Theorem~\ref{thm:successor-preservation-2d}}
\label{subsec:app:proof:thm:successor-preservation-2d}
\setcounter{retheorem}{2}
\begin{retheorem}[Successor-level preservation under two-dimensional interpolation] Let $u$ and $w$ be two cells at the same extension level, and suppose
their mechanisms satisfy the coordinate-wise cross-cell Lipschitz bound:
for any $\mathbf{x}\in\mathcal{X}_u$, $\mathbf{x}'\in\mathcal{X}_w$, and $\mathbf{y}\in\mathcal{Y}$,
\[
\left|
\ln \mathcal{M}_u(\mathbf{y}|\mathbf{x})-\ln \mathcal{M}_w(\mathbf{y}|\mathbf{x}')
\right|
\le
\sum_{t=1}^{2}\epsilon_t |x_t-x'_t|.
\]
Let $u'\in\mathrm{Succ}(u)$ and $w'\in\mathrm{Succ}(w)$, and suppose that
the mechanisms on $u'$ and $w'$ are generated from $\mathcal{M}_u$ and $\mathcal{M}_w$,
respectively, using the two-dimensional interpolation operator in
Definition~\ref{def:2d-interpolation}. Then, for any
$\mathbf{x}\in\mathcal{X}_{u'}$, $\mathbf{x}'\in\mathcal{X}_{w'}$, and $\mathbf{y}\in\mathcal{Y}$,
the pre-normalization interpolants satisfy
\[
\left|
\ln f_u(\mathbf{y}|\mathbf{x})-\ln f_w(\mathbf{y}|\mathbf{x}')
\right|
\le
\sum_{t=1}^{2}\epsilon_t |x_t-x'_t|.
\]
Consequently, if $\epsilon_1$ and $\epsilon_2$ satisfy the
dimension-wise composition condition in Theorem~\ref{thm:composition},
then
\[
\left|
\ln f_u(\mathbf{y}|\mathbf{x})-\ln f_w(\mathbf{y}|\mathbf{x}')
\right|
\le
\bar{\epsilon} d_p(\mathbf{x},\mathbf{x}'),
\]
where $\bar{\epsilon}$ is the composed pre-normalization budget. In
particular, with the calibrated choice $\bar{\epsilon}=\epsilon/2$, the
normalized successor mechanisms satisfy the target $(\epsilon,d)$-mDP bound.
\end{retheorem}

\begin{lemma}[Cross-cell bound from 2D interpolation when one coordinate coincides]
\label{lem:cross-cell-one-coordinate}
Let $u$ and $w$ be two cells at the same extension level of the
aligned grid tree, and let
\[
u'\in\mathrm{Succ}(u),\qquad w'\in\mathrm{Succ}(w)
\]
be successor cells generated by the two-dimensional interpolation rule in
Definition~\ref{def:2d-interpolation}. Let
\[
\mathbf{x}=(x_1,x_2)\in\mathcal X_{u'},
\qquad
\mathbf{x}'=(x_1',x_2')\in\mathcal X_{w'}.
\]
Assume that either $x_1=x_1'$ or $x_2=x_2'$. Also assume that the
corresponding endpoint pairs used by the interpolation satisfy the
coordinate-wise cross-cell Lipschitz bound. Then, for every
$\mathbf{y}\in\mathcal Y$,
\begin{equation}
\label{eq:cross-cell-one-coordinate}
\bigl|
\ln f_u(\mathbf \mathbf{y}\mid \mathbf x)
-
\ln f_w(\mathbf \mathbf{y}\mid \mathbf{x}')
\bigr|
\le
\epsilon_1 |x_1-x_1'|+\epsilon_2 |x_2-x_2'|.
\end{equation}
\end{lemma}

\begin{proof}
We prove the case $x_1=x_1'$; the case $x_2=x_2'$ is symmetric by
exchanging coordinates $1$ and $2$.

Let $\mathbf{x}^{L},\mathbf{x}^{R}$ be the two endpoints of the horizontal
interpolation segment through $\mathbf{x}$ inside cell $u$, and let
$\mathbf{x}'^{L},\mathbf{x}'^{R}$ be the corresponding endpoints of the
horizontal interpolation segment through $\mathbf{x}'$ inside cell $w$.
Since $u$ and $w$ are at the same extension level and use the same
extension template, the equality $x_1=x_1'$ implies that
$\mathbf{x}$ and $\mathbf{x}'$ have the same relative position along
coordinate $1$. Hence they use the same interpolation coefficient
$\lambda$ along that coordinate.

By the one-dimensional restriction property of
Definition~\ref{def:2d-interpolation},
\[
\ln f_u(\mathbf \mathbf{y}\mid \mathbf x)
=
\lambda \ln f_u(\mathbf \mathbf{y}\mid \mathbf{x}^{L})
+
(1-\lambda)\ln f_u(\mathbf \mathbf{y}\mid \mathbf{x}^{R}),
\]
and
\[
\ln f_w(\mathbf \mathbf{y}\mid \mathbf{x}')
=
\lambda \ln f_w(\mathbf \mathbf{y}\mid \mathbf{x}'^{L})
+
(1-\lambda)\ln f_w(\mathbf \mathbf{y}\mid \mathbf{x}'^{R}).
\]
Subtracting the two equalities and applying the triangle inequality gives
\begin{align}
&
\bigl|
\ln f_u(\mathbf \mathbf{y}\mid \mathbf x)
-
\ln f_w(\mathbf \mathbf{y}\mid \mathbf{x}')
\bigr|
\nonumber\\
&\le
\lambda
\bigl|
\ln f_u(\mathbf \mathbf{y}\mid \mathbf{x}^{L})
-
\ln f_w(\mathbf \mathbf{y}\mid \mathbf{x}'^{L})
\bigr|
\nonumber\\
&\quad+
(1-\lambda)
\bigl|
\ln f_u(\mathbf \mathbf{y}\mid \mathbf{x}^{R})
-
\ln f_w(\mathbf \mathbf{y}\mid \mathbf{x}'^{R})
\bigr|.
\label{eq:case1-step}
\end{align}
By the coordinate-wise cross-cell Lipschitz assumption applied to the two
corresponding endpoint pairs,
\[
\bigl|
\ln f_u(\mathbf \mathbf{y}\mid \mathbf{x}^{L})
-
\ln f_w(\mathbf \mathbf{y}\mid \mathbf{x}'^{L})
\bigr|
\le
\epsilon_2 |x_2-x_2'|,
\]
and
\[
\bigl|
\ln f_u(\mathbf \mathbf{y}\mid \mathbf{x}^{R})
-
\ln f_w(\mathbf \mathbf{y}\mid \mathbf{x}'^{R})
\bigr|
\le
\epsilon_2 |x_2-x_2'|.
\]
Substituting these bounds into \eqref{eq:case1-step} yields
\[
\bigl|
\ln f_u(\mathbf \mathbf{y}\mid \mathbf x)
-
\ln f_w(\mathbf \mathbf{y}\mid \mathbf{x}')
\bigr|
\le
\epsilon_2 |x_2-x_2'|.
\]
Since $x_1=x_1'$, this is exactly
\[
\bigl|
\ln f_u(\mathbf \mathbf{y}\mid \mathbf x)
-
\ln f_w(\mathbf \mathbf{y}\mid \mathbf{x}')
\bigr|
\le
\epsilon_1 |x_1-x_1'|+\epsilon_2 |x_2-x_2'|.
\]
This proves the claim.
\end{proof}

\begin{proof}
Fix arbitrary
\[
\mathbf{x}=(x_1,x_2)\in\mathcal X_{u'},
\qquad
\mathbf{x}'=(x_1',x_2')\in\mathcal X_{w'},
\qquad
\mathbf{y}\in\mathcal Y.
\]
Let $u=C_{i_0j_0}$ and $w=C_{i_1j_1}$. Define the intermediate cell
\[
t := C_{i_0j_1},
\]
and let $t'\in\mathrm{Succ}(t)$ be the successor cell occupying the same
relative position under the common extension template as $u'$ and $w'$.
If the retained grid excludes $t$ or $t'$, we treat them as virtual
aligned cells used only for the proof; the same interpolation rule assigns
pre-normalization values on these virtual points.

Define the intermediate point
\[
\mathbf{m}:=(x_1,x_2').
\]
By alignment of the extension template, $\mathbf{m}\in\mathcal X_{t'}$
or is a virtual interpolation point in the aligned extension.

Now $\mathbf{x}\in\mathcal X_{u'}$ and $\mathbf{m}\in\mathcal X_{t'}$
have the same first coordinate. Hence, by
Lemma~\ref{lem:cross-cell-one-coordinate},
\begin{equation}
\label{eq:general-NN-step1}
\bigl|
\ln f_u(\mathbf{y}\mid \mathbf{x})
-
\ln f_t(\mathbf{y}\mid \mathbf{m})
\bigr|
\le
\epsilon_2 |x_2-x_2'|.
\end{equation}
Similarly, $\mathbf{m}\in\mathcal X_{t'}$ and
$\mathbf{x}'\in\mathcal X_{w'}$ have the same second coordinate. Again by
Lemma~\ref{lem:cross-cell-one-coordinate},
\begin{equation}
\label{eq:general-NN-step2}
\bigl|
\ln f_t(\mathbf{y}\mid \mathbf{m})
-
\ln f_w(\mathbf{y}\mid \mathbf{x}')
\bigr|
\le
\epsilon_1 |x_1-x_1'|.
\end{equation}

Applying the triangle inequality and then substituting
\eqref{eq:general-NN-step1} and \eqref{eq:general-NN-step2}, we obtain
\[
\bigl|
\ln f_u(\mathbf{y}\mid \mathbf{x})
-
\ln f_w(\mathbf{y}\mid \mathbf{x}')
\bigr|
\le
\epsilon_1 |x_1-x_1'|+\epsilon_2 |x_2-x_2'|.
\]
This proves the coordinate-wise cross-cell preservation bound.

By the dimension-wise composition theorem, if $\epsilon_1$ and
$\epsilon_2$ compose to the pre-normalization budget $\bar\epsilon$, then
\[
\bigl|
\ln f_u(\mathbf{y}\mid \mathbf{x})
-
\ln f_w(\mathbf{y}\mid \mathbf{x}')
\bigr|
\le
\bar\epsilon d_p(\mathbf{x},\mathbf{x}').
\]

Finally, consider the normalized mechanisms
\[
\mathcal{M}_u(\mathbf \mathbf{y}\mid \mathbf x)
=
\frac{f_u(\mathbf \mathbf{y}\mid \mathbf x)}
{\sum_{\mathbf y'\in\mathcal Y} f_u(\mathbf y'\mid \mathbf x)}
\]
and similarly for $\mathcal{M}_w$. Since the pre-normalization log-scores are
$\bar\epsilon$-Lipschitz for every output, the corresponding log-normalizers
are also $\bar\epsilon$-Lipschitz. Therefore,
\[
\bigl|
\ln \mathcal{M}_u(\mathbf \mathbf{y}\mid \mathbf x)
-
\ln \mathcal{M}_w(\mathbf \mathbf{y}\mid \mathbf{x}')
\bigr|
\le
2\bar\epsilon d_p(\mathbf{x},\mathbf{x}').
\]
Choosing $\bar\epsilon=\epsilon/2$ yields the desired
$(\epsilon,d)$-mDP bound for the normalized successor mechanisms.
\end{proof}

\subsection{Proof of Theorem \ref{thm:subtree-preservation}}
\label{subsec:app:proof:thm:subtree-preservation}
\begin{retheorem}[Subtree-level preservation of the tree-based extension]
Consider the tree-based extension algorithm. Fix a node $v$ in the
extension tree, and define
\[
\mathcal{X}_{\mathrm{sub}}(v)
:=
\bigcup_{\mathbf{x}\in \mathrm{Desc}(v)} \mathcal{X}_z .
\]
Assume that the corner anchors of $v$ satisfy the coordinate-wise
Lipschitz bounds in log space with pre-normalization budgets
$\epsilon_1,\epsilon_2$, and that every extension inside the
subtree rooted at $v$ is performed by the 2-dimensional interpolation
operator in Definition~\ref{def:2d-interpolation}. Assume further that
$\epsilon_1,\epsilon_2$ satisfy the dimension-wise composition
condition with total pre-normalization budget $\epsilon$ under the
$\ell_p$ metric.

Then, for any two records
$\mathbf{x},\mathbf{x}'\in \mathcal{X}_{\mathrm{sub}}(v)$ and any output
$\mathbf{y}\in\mathcal{Y}$, the pre-normalization interpolants satisfy
\[
\left|
\ln f_v(\mathbf{y}\mid \mathbf{x})-\ln f_v(\mathbf{y}\mid \mathbf{x}')
\right|
\le
\epsilon d_p(\mathbf{x},\mathbf{x}').
\]
Consequently, after the calibrated normalization step, the induced
mechanism on $\mathcal{X}_{\mathrm{sub}}(v)$ satisfies
$(\epsilon,d)$-mDP.
\end{retheorem}

\begin{table*}[t]
\centering
\caption{Comparison of conventional central DP, LDP, and mDP. }
\label{tab:dp_comparison}
{\rev 
\begin{tabular}{llll}
\toprule
Property & Central DP & LDP & mDP \\
\midrule
Input
& Dataset $D$
& Individual record $\mathbf{x}$
& Secret $\mathbf{x}$ or dataset $D$ \\
Randomization
& Trusted curator
& User device
& Central or local \\
Compared inputs
& Neighboring datasets
& Any two records
& Any two secrets \\
Privacy bound
& $e^\epsilon$
& $e^\epsilon$
& $e^{\epsilon d(\mathbf{x},\mathbf{x}')}$ \\
Distance
& Dataset adjacency
& Discrete record metric
& Application-defined metric \\
\bottomrule
\end{tabular}}%
\end{table*}

\begin{lemma}[Coordinate-wise preservation in a subtree]
\label{lem:subtree-coordinate-preservation}
Fix a node $v$ in the extension tree, and define
\[
\mathcal{X}_{\mathrm{sub}}(v)
:=
\bigcup_{\mathbf{x}\in \mathrm{Desc}(v)} \mathcal{X}_z .
\]
Assume that the corner anchors of $v$ satisfy the coordinate-wise
Lipschitz bounds in log space, and that every extension inside the
subtree rooted at $v$ is performed by the 2-dimensional interpolation
operator in Definition~\ref{def:2d-interpolation}. Then, for any
two records $\mathbf{x},\mathbf{x}'\in\mathcal{X}_{\mathrm{sub}}(v)$ that differ only
in coordinate $t\in\{1,2\}$, the pre-normalization interpolants
satisfy
\[
\left|
\ln f_v(\mathbf{y}\mid \mathbf{x})-\ln f_v(\mathbf{y}\mid \mathbf{x}')
\right|
\le
\epsilon_t |x_t-\mathbf{x}'_t|,
\qquad \forall \mathbf{y}\in\mathcal{Y}.
\]
\end{lemma}
\begin{proof}
We prove the case $t=1$; the case $t=2$ is symmetric.

By Definition~\ref{def:2d-interpolation}, the restriction of each
2-dimensional interpolation step to any axis-aligned segment parallel
to coordinate $1$ is a one-dimensional interpolation rule satisfying
the $(\epsilon_1,d_1)$-Lipschitz bound in log space. By
Proposition~\ref{prop:1d-validity}, every pair of
records generated on such a segment satisfies the coordinate-wise
bound.

Moreover, since the parent anchors of $v$ satisfy the coordinate-wise
Lipschitz constraints, the first extension from $v$ produces child
corner values that also satisfy these constraints. Therefore, the
child corners are valid anchors for the next extension step. Repeating
this argument level by level shows that every one-dimensional
interpolation segment inside the subtree rooted at $v$ satisfies the
same coordinate-wise Lipschitz bound.

Now let
\[
x=(x_1,t), \qquad \mathbf{x}'=(x'_1,t)
\]
be two records in $\mathcal{X}_{\mathrm{sub}}(v)$ with $x_1<x'_1$.
Because the extensions are axis-aligned and nested, the horizontal
segment from $x$ to $\mathbf{x}'$ can be partitioned into finitely many
subsegments
\[
x=\mathbf{p}_0,\mathbf{p}_1,\ldots,\mathbf{p}_m=\mathbf{x}',
\]
where each consecutive pair $\mathbf{p}_{i-1},\mathbf{p}_i$ lies on one such
one-dimensional interpolation segment. Hence, for every
$i=1,\ldots,m$,
\[
\left|
\ln f_v(\mathbf{y}\mid \mathbf{p}_i)-\ln f_v(\mathbf{y}\mid \mathbf{p}_{i-1})
\right|
\le
\epsilon_1 |(\mathbf{p}_i)_1-(\mathbf{p}_{i-1})_1|.
\]
By the triangle inequality,
\[
\begin{aligned}
\left|
\ln f_v(\mathbf{y}\mid \mathbf{x})-\ln f_v(\mathbf{y}\mid \mathbf{x}')
\right|
&\le
\sum_{i=1}^{m}
\left|
\ln f_v(\mathbf{y}\mid \mathbf{p}_i)-\ln f_v(\mathbf{y}\mid \mathbf{p}_{i-1})
\right| \\
&\le
\epsilon_1
\sum_{i=1}^{m}
|(\mathbf{p}_i)_1-(\mathbf{p}_{i-1})_1| \\
&=
\epsilon_1 |x_1-x'_1|.
\end{aligned}
\]
Thus the coordinate-wise preservation bound holds for $t=1$.
The case $t=2$ follows identically by considering a vertical
segment.
\end{proof}
\begin{proof}
By Lemma~\ref{lem:subtree-coordinate-preservation}, the
2-dimensional interpolation operator preserves the coordinate-wise
Lipschitz bounds throughout the subtree rooted at $v$. Therefore, for
any two records in $\mathcal{X}_{\mathrm{sub}}(v)$ that differ only in
coordinate $t\in\{1,2\}$, we have
\[
\left|
\ln f_v(\mathbf{y}\mid \mathbf{x})-\ln f_v(\mathbf{y}\mid \mathbf{x}')
\right|
\le
\epsilon_t |x_t-\mathbf{x}'_t|,
\qquad \forall \mathbf{y}\in\mathcal{Y}.
\]

Now take arbitrary records
\[
x=(x_1,x_2), \qquad \mathbf{x}'=(x'_1,\mathbf{x}'_2)
\]
in $\mathcal{X}_{\mathrm{sub}}(v)$, and fix any output
$\mathbf{y}\in\mathcal{Y}$. Define the intermediate point
\[
m := (x_1,\mathbf{x}'_2).
\]
If $m$ is not an explicitly retained record, we treat it as a virtual
interpolation point inside the aligned extension of the subtree rooted
at $v$. The same interpolation rule assigns a pre-normalization value
$f_v(\mathbf{y}\mid m)$ to this point.

Since $x$ and $m$ differ only in the second coordinate, Lemma
\ref{lem:subtree-coordinate-preservation} gives
\[
\left|
\ln f_v(\mathbf{y}\mid \mathbf{x})-\ln f_v(\mathbf{y}\mid m)
\right|
\le
\epsilon_2 |x_2-\mathbf{x}'_2|.
\]
Similarly, since $m$ and $\mathbf{x}'$ differ only in the first coordinate, the
same lemma gives
\[
\left|
\ln f_v(\mathbf{y}\mid m)-\ln f_v(\mathbf{y}\mid \mathbf{x}')
\right|
\le
\epsilon_1 |x_1-x'_1|.
\]
By the triangle inequality,
\[
\begin{aligned}
\left|
\ln f_v(\mathbf{y}\mid \mathbf{x})-\ln f_v(\mathbf{y}\mid \mathbf{x}')
\right|
&\le
\left|
\ln f_v(\mathbf{y}\mid \mathbf{x})-\ln f_v(\mathbf{y}\mid m)
\right|  \\
&\quad+
\left|
\ln f_v(\mathbf{y}\mid m)-\ln f_v(\mathbf{y}\mid \mathbf{x}')
\right| \\
&\le
\epsilon_1 |x_1-x'_1|
+
\epsilon_2 |x_2-\mathbf{x}'_2|.
\end{aligned}
\]
Since $\epsilon_1,\epsilon_2$ satisfy the dimension-wise composition
condition with total pre-normalization budget $\epsilon/2$, we have
\[
\epsilon_1 |x_1-x'_1|
+
\epsilon_2 |x_2-\mathbf{x}'_2|
\le
\frac{\epsilon}{2} d_p(\mathbf{x},\mathbf{x}').
\]
Therefore,
\[
\left|
\ln f_v(\mathbf{y}\mid \mathbf{x})-\ln f_v(\mathbf{y}\mid \mathbf{x}')
\right|
\le
\frac{\epsilon}{2} d_p(\mathbf{x},\mathbf{x}').
\]

Finally, by the calibrated normalization step, a pre-normalization
$(\epsilon/2,d_p)$-Lipschitz bound in log space yields a final
$(\epsilon,d_p)$-mDP mechanism. Hence the induced mechanism on
$\mathcal{X}_{\mathrm{sub}}(v)$ satisfies $(\epsilon,d)$-mDP.
\end{proof}

{\rev 
\section{Comparison Between DP, LDP, and mDP}
\label{sec:DPvsmDP}
Table~\ref{tab:dp_comparison} compares conventional central DP, LDP, and mDP. Central DP operates on a dataset and provides indistinguishability between neighboring datasets, typically through randomization performed by a trusted curator. In contrast, LDP randomizes each user's record locally before release and protects any pair of possible records using the same multiplicative privacy bound, $e^\epsilon$. The randomized reports may still be aggregated by the server; thus, the key distinction between central DP and LDP is where randomization occurs and whether the curator has access to raw records, rather than perturbation versus aggregation. mDP generalizes these guarantees by incorporating an application-defined metric $d(\mathbf{x},\mathbf{x}')$, resulting in the distance-dependent privacy bound $e^{\epsilon d(\mathbf{x},\mathbf{x')}}$. Thus, secrets that are closer under the metric receive stronger indistinguishability guarantees.

Importantly, "local" specifies where perturbation occurs, whereas "metric" specifies how indistinguishability scales with the distance between secrets. When the input is an individual record and $d(\mathbf{x},\mathbf{x}')=\mathbf{1}_{[\mathbf{x}\neq\mathbf{x}']}$, mDP reduces to $\epsilon$-LDP. When the input is a dataset and $d$ is the dataset-level Hamming distance, mDP reduces to standard $\epsilon$-DP for neighboring datasets and to its group-privacy extension for datasets differing in multiple records. Therefore, mDP can be instantiated in either the central or local setting, depending on the protected input and where randomization is performed.}

\section{Discussion: Order independence of transform-linear interpolation.}
\label{sec:app:discussion:order-independence}


\begin{reproposition}[Order independence for transform-linear interpolation]
\label{prop:order-independence}
Consider a predefined interpolation operator
\[
\mathcal{M}(a,d;t)
=
f^{-1}\bigl((1-\alpha(t))f(a)+\alpha(t)f(d)\bigr),
\]
where $f$ is continuous and strictly monotone, and
$\alpha:[0,1]\to[0,1]$ is continuous and strictly monotone with
$\alpha(0)=0$ and $\alpha(1)=1$. Starting from $x_0=a$, define
\[
x_k := \mathcal{M}(x_{k-1},d;t_k), \qquad k=1,\ldots,n .
\]
Then $x_n$ depends on $t_1,\ldots,t_n$ only through
\[
\prod_{k=1}^n (1-\alpha(t_k)).
\]
Consequently, $x_n$ is independent of the order of the parameters
$t_1,\ldots,t_n$.
\end{reproposition}

\begin{proof}
Recall that the interpolation operator is defined as
\begin{equation}
\mathcal{I}(\xi_L,\xi_R;t)
=
\psi^{-1}\!\left(
(1-\alpha(t))\psi(\xi_L)+\alpha(t)\psi(\xi_R)
\right),
\end{equation}
where $\psi$ is continuous and strictly monotone, and
$\alpha:[0,1]\to[0,1]$ is continuous and strictly monotone with
$\alpha(0)=0$ and $\alpha(1)=1$. Define
\begin{equation}
\beta(t):=1-\alpha(t).
\end{equation}
Then
\begin{equation}
\psi\!\left(\mathcal{I}(\xi_L,\xi_R;t)\right)
=
\beta(t)\psi(\xi_L)
+
\bigl(1-\beta(t)\bigr)\psi(\xi_R).
\label{eq:beta_form}
\end{equation}

Starting from $\xi_0=\xi_L$, consider repeated interpolation toward
the same fixed endpoint $\xi_R$:
\begin{equation}
\xi_k := \mathcal{I}(\xi_{k-1},\xi_R;t_k),
\qquad k=1,\ldots,n .
\label{eq:recursive_interp}
\end{equation}
Applying \eqref{eq:beta_form} with $\xi_L=\xi_{k-1}$ gives
\begin{equation}
\psi(\xi_k)
=
\beta(t_k)\psi(\xi_{k-1})
+
\bigl(1-\beta(t_k)\bigr)\psi(\xi_R).
\label{eq:one_step}
\end{equation}

We prove by induction that
\begin{equation}
\psi(\xi_n)
=
\left(\prod_{i=1}^{n}\beta(t_i)\right)\psi(\xi_L)
+
\left(1-\prod_{i=1}^{n}\beta(t_i)\right)\psi(\xi_R).
\label{eq:n_step_formula}
\end{equation}
For $n=1$, the claim follows directly from \eqref{eq:beta_form}.
Assume that it holds for $n-1$. Then, by \eqref{eq:one_step},
\begin{align}
\psi(\xi_n)
&=
\beta(t_n)\psi(\xi_{n-1})
+
\bigl(1-\beta(t_n)\bigr)\psi(\xi_R) \notag\\
&=
\beta(t_n)
\left[
\left(\prod_{i=1}^{n-1}\beta(t_i)\right)\psi(\xi_L)
+
\left(1-\prod_{i=1}^{n-1}\beta(t_i)\right)\psi(\xi_R)
\right] \notag\\
&\quad+
\bigl(1-\beta(t_n)\bigr)\psi(\xi_R) \notag\\
&=
\left(\prod_{i=1}^{n}\beta(t_i)\right)\psi(\xi_L)
+
\left(1-\prod_{i=1}^{n}\beta(t_i)\right)\psi(\xi_R).
\end{align}
Thus \eqref{eq:n_step_formula} holds for all $n$.

Let
\begin{equation}
B_n := \prod_{i=1}^{n}\beta(t_i).
\end{equation}
Since $\alpha$ is continuous and strictly increasing with
$\alpha(0)=0$ and $\alpha(1)=1$, the function
$\beta(t)=1-\alpha(t)$ is continuous and strictly decreasing from
$1$ to $0$. Hence $\beta$ is bijective on $[0,1]$. Therefore,
for $B_n\in[0,1]$, we may define
\begin{equation}
T_n := \beta^{-1}(B_n).
\end{equation}
Substituting $B_n=\beta(T_n)$ into \eqref{eq:n_step_formula}, we get
\begin{equation}
\psi(\xi_n)
=
\beta(T_n)\psi(\xi_L)
+
\bigl(1-\beta(T_n)\bigr)\psi(\xi_R)
=
\psi\!\left(\mathcal{I}(\xi_L,\xi_R;T_n)\right).
\end{equation}
Since $\psi$ is strictly monotone, it is injective. Therefore,
\begin{equation}
\xi_n=\mathcal{I}(\xi_L,\xi_R;T_n).
\label{eq:collapse}
\end{equation}

Finally, $T_n$ depends on $t_1,\ldots,t_n$ only through the product
\[
B_n=\prod_{i=1}^{n}\beta(t_i).
\]
Because multiplication is commutative, $B_n$, and hence $T_n$, is
unchanged under any permutation of $t_1,\ldots,t_n$. Therefore, the
final value $\xi_n$ is independent of the order in which the
parameters $t_1,\ldots,t_n$ are applied. This proves the claimed
order-independence property for repeated interpolation toward a fixed
endpoint.
\end{proof}

Notably, Proposition~\ref{prop:order-independence} applies only to
interpolation rules that can be written in transform-linear form,
namely
\[
\mathcal{I}(\xi_L,\xi_R;t)
=
\psi^{-1}\!\left(
(1-\alpha(t))\psi(\xi_L)+\alpha(t)\psi(\xi_R)
\right).
\]
Among the three interpolation methods considered in this paper, the
log-convex interpolation rule satisfies this requirement directly. In
particular, for positive endpoint values $\xi_L,\xi_R$, log-convex
interpolation can be written with $\psi(\xi)=\ln \xi$:
\[
\mathcal{I}(\xi_L,\xi_R;t)
=
\exp\!\left((1-\alpha(t))\ln \xi_L
+
\alpha(t)\ln \xi_R\right)
=
\xi_L^{1-\alpha(t)}\xi_R^{\alpha(t)}.
\]
Thus, repeated log-convex interpolation toward a fixed endpoint has
the order-independence property stated in
Proposition~\ref{prop:order-independence}.

The McShane--Whitney extension does not generally satisfy the
transform-linear form. Although it produces Lipschitz-valid values by
choosing a point between lower and upper McShane--Whitney envelopes,
the resulting value is determined by envelope operations rather than by
a fixed transform-linear average of two endpoints. Therefore,
Proposition~\ref{prop:order-independence} does not directly apply
to the McShane--Whitney extension, except in special cases where the
chosen value happens to reduce to a transform-linear interpolation rule.

Similarly, the LP-based interpolation method does not generally satisfy
the proposition. Its values are obtained as the solution of a constrained
optimization problem and may depend on the objective, the feasible
region, and the active constraints. Hence, unless the LP is explicitly
restricted so that its solution coincides with a transform-linear rule,
the order-independence property in Proposition~\ref{prop:order-independence}
is not guaranteed.

Therefore, Proposition~\ref{prop:order-independence} should be
viewed as an auxiliary property of predefined transform-linear
interpolation rules, most notably log-convex interpolation. It is not
used in the proof of \textnormal{(A3)}, which relies instead on
coordinate-wise Lipschitz preservation and dimension-wise composition.

\section{Details of Interpolation Algorithms Satisfying Requirement A1}
\label{sec:app:Interpo}
We present three representative extension algorithms that can be designed to satisfy the \emph{local mDP constraint} \textbf{\emph{(A1)}}. Fix a node $v$ with associated record set $\mathcal{X}_v$, and suppose a predecessor mechanism
$\{\mathcal{M}_v(\cdot\mid \hat{\mathbf{x}})\}_{\hat{\mathbf{x}}\in\mathcal{X}_v}$ is already specified on $\mathcal{X}_v$. Each algorithm below constructs successor mechanisms $\{\mathcal{M}_{v\to u}\}_{u\in\mathrm{Succ}(v)}$ over the local extension cover $\bigcup_{u\in\mathrm{Succ}(v)} \mathcal{X}_u$.

\begin{itemize}
    \item[(1)] \textbf{McShane--Whitney extension~\cite{borgs2018extend}.}
    For each output $\mathbf{y}\in\mathcal{Y}$, assume that $\log \mathcal{M}_v(\mathbf{y} \mid \hat{\mathbf{x}})$ satisfies the $\epsilon/2$-Lipschitz bound on $\mathcal{X}_v$. Then, for any record $\mathbf{x}\in \bigcup_{u\in\mathrm{Succ}(v)} \mathcal{X}_u$, define the McShane--Whitney envelopes
    \begin{align}
    l_{\mathbf{y}}(\mathbf{x}) &\triangleq \sup_{\hat{\mathbf{x}}\in \mathcal{X}_v}
    \Bigl(\log \mathcal{M}_v(\mathbf{y} \mid \hat{\mathbf{x}})-\frac{\epsilon}{2} d(\mathbf{x},\hat{\mathbf{x}})\Bigr), \label{eq:mw-lower}\\
    u_{\mathbf{y}}(\mathbf{x}) &\triangleq \inf_{\hat{\mathbf{x}}\in \mathcal{X}_v}
    \Bigl(\log \mathcal{M}_v(\mathbf{y} \mid \hat{\mathbf{x}})+\frac{\epsilon}{2} d(\mathbf{x},\hat{\mathbf{x}})\Bigr). \label{eq:mw-upper}
    \end{align}
    Since $l_{\mathbf{y}}(\mathbf{x})\le u_{\mathbf{y}}(\mathbf{x})$, choose any $f_{\mathbf{y}}(\mathbf{x})\in[l_{\mathbf{y}}(\mathbf{x}),u_{\mathbf{y}}(\mathbf{x})]$ (e.g., the midpoint), and for each $u\in\mathrm{Succ}(v)$ define
    \begin{equation}
    \mathcal{M}_{v\to u}(\mathbf{y} \mid \mathbf{x})
    \triangleq
    \frac{\exp(f_{\mathbf{y}}(\mathbf{x}))}{\sum_{\mathbf{y}'\in\mathcal{Y}} \exp(f_{\mathbf{y}'}(\mathbf{x}))},
    \qquad \forall \mathbf{x}\in \mathcal{X}_u.
    \label{eq:mw-normalize}
    \end{equation}
    By construction, the resulting successor mechanisms satisfy the required local $(\epsilon,d)$-mDP inequalities across successor regions.

    \item[(2)] \textbf{Log-convex interpolation~\cite{Qiu-USec2026}.}
    Assume that $\mathcal{X}_v$ consists of the vertices of a $d$-dimensional grid cell, and that the successor regions of $v$ contain records inside the same cell. Let $\bar\epsilon_t$ denote the pre-normalization privacy budget assigned to coordinate $t$. We assume that the anchor mechanism on $\mathcal{X}_v$ satisfies the coordinate-wise log-Lipschitz bound on every pair of $t$-neighbor anchors $\hat{\mathbf{x}},\hat{\mathbf{x}}'\in\mathcal{X}_v$:
    \begin{equation}
    \bigl|\log \mathcal{M}_v(\mathbf{y}\mid \hat{\mathbf{x}})
    -\log \mathcal{M}_v(\mathbf{y}\mid \hat{\mathbf{x}}')\bigr|
    \le
    \bar\epsilon_t |\hat{x}_t-\hat{x}'_t|,
    \quad \forall \mathbf{y}\in\mathcal{Y}.
    \label{eq:anchor-neighbor-mdp}
    \end{equation}
    For each record $\mathbf{x}\in\bigcup_{u\in\mathrm{Succ}(v)}\mathcal{X}_u$, let
    $\{\lambda_{\hat{\mathbf{x}}}(\mathbf{x})\}_{\hat{\mathbf{x}}\in\mathcal{X}_v}$
    be the multilinear interpolation weights induced by the coordinates of
    $\mathbf{x}$ in the cell. The pre-normalization log-score is defined as
    \begin{equation}
    \log f_{\mathbf{y}}(\mathbf{x})
    \triangleq
    \sum_{\hat{\mathbf{x}}\in\mathcal{X}_v}
    \lambda_{\hat{\mathbf{x}}}(\mathbf{x})
    \log \mathcal{M}_v(\mathbf{y}\mid \hat{\mathbf{x}}).
    \label{eq:log-convex}
    \end{equation}
    The successor mechanism is then obtained by normalization:
    \begin{equation}
    \mathcal{M}_{v\to u}(\mathbf{y}\mid \mathbf{x})
    \triangleq
    \frac{f_{\mathbf{y}}(\mathbf{x})}
    {\sum_{\mathbf{y}'\in\mathcal{Y}} f_{\mathbf{y}'}(\mathbf{x})},
    \qquad
    \forall \mathbf{x}\in\mathcal{X}_u.
    \label{eq:log-convex-normalize}
    \end{equation}
    Under the coordinate-wise interpolation rule, the pre-normalization
    log-scores satisfy the coordinate-wise Lipschitz bounds, which compose
    to a pre-normalization $\bar\epsilon$-Lipschitz bound under the
    dimension-wise composition condition. By choosing the pre-normalization
    budget so that normalization yields the target budget $\epsilon$ 
    (e.g., $\bar\epsilon=\epsilon/2$), the resulting successor mechanisms
    satisfy the required local mDP constraints in \textbf{\emph{(A1)}}.

    \item[(3)] \textbf{LP-based interpolation.}
    The successor mechanisms \newline $\{\mathcal{M}_{v\to u}\}_{u\in\mathrm{Succ}(v)}$ can also be constructed by solving a LP over the local extension cover. Specifically, for each $u\in\mathrm{Succ}(v)$ and each $\mathbf{x}\in\mathcal{X}_u$, impose the simplex constraint
    \begin{equation}
    \textstyle \sum_{\mathbf{y}\in\mathcal{Y}} \mathcal{M}_{v\to u}(\mathbf{y} \mid \mathbf{x})=1,
    \quad \forall u\in\mathrm{Succ}(v),\ \forall \mathbf{x}\in\mathcal{X}_u,
    \label{eq:lp-simplex}
    \end{equation}
    together with predecessor-inheritance constraints on shared records,
    \begin{eqnarray}
    &&\mathcal{M}_{v\to u}(\mathbf{y} \mid \hat{\mathbf{x}})=\mathcal{M}_v(\mathbf{y} \mid \hat{\mathbf{x}}),
    \\
    &&\forall u\in\mathrm{Succ}(v),\ \forall \hat{\mathbf{x}}\in \mathcal{X}_v\cap\mathcal{X}_u,\ \forall \mathbf{y}\in\mathcal{Y},
    \label{eq:lp-anchor}
    \end{eqnarray}
    and the local $(\epsilon,d)$-mDP constraints
    \begin{eqnarray}
    &&\mathcal{M}_{v\to u}(\mathbf{y} \mid \mathbf{x})\le e^{\epsilon d(\mathbf{x},\mathbf{x}')}\mathcal{M}_{v\to w}(\mathbf{y} \mid \mathbf{x}'), \\ 
    &&
    \forall u,w\in\mathrm{Succ}(v),\ \forall \mathbf{x}\in\mathcal{X}_u,\ \forall \mathbf{x}'\in\mathcal{X}_w,\ \forall \mathbf{y}\in\mathcal{Y},
    \label{eq:lp-mdp}
    \end{eqnarray}
    together with the symmetric constraint obtained by exchanging $(x,u)$ and $(\mathbf{x}',w)$. Finally, optimize a utility objective such as
    \begin{equation}
    \min\ \sum_{u\in\mathrm{Succ}(v)}\sum_{\mathbf{x}\in\mathcal{X}_u}\sum_{\mathbf{y}\in\mathcal{Y}}
    \pi(\mathbf{x})\,\mathcal{L}(x,y)\,\mathcal{M}_{v\to u}(\mathbf{y} \mid \mathbf{x}),
    \label{eq:lp-obj}
    \end{equation}
    where $\pi$ is a prior over the local extension cover and $\mathcal{L}(x,y)$ is a task-specific loss. Any feasible optimum satisfies \textbf{\emph{(A1)}} by construction.
\end{itemize}

\section{Perturbation Optimization for Seed Records}
\label{sec:seed-optimization}

Before applying the tree-based extension procedure, we first optimize a
perturbation mechanism on a small set of seed records. This seed mechanism
serves as the anchor mechanism for subsequent recursive extension, while
avoiding full-domain optimization over all fine-grained records.

Let
$\mathcal{X}_{\mathrm{real}}=\{\mathbf{r}_1,\ldots,\mathbf{r}_N\}$,
$\mathcal{X}_{\mathrm{seed}}=\{\mathbf{x}_1,\ldots,\mathbf{x}_J\}$, and
$\mathcal{Y}=\{\mathbf{y}_1,\ldots,\mathbf{y}_K\}$. For each real location
$\mathbf{r}_n$ and output $\mathbf{y}_k$, let
$L_{n,k}=t(\mathbf{r}_n,\mathbf{y}_k)$ denote the task-specific utility
loss. In our experiments, this loss measures the travel-distance estimation
error when $\mathbf{y}_k$ is used in place of the true location
$\mathbf{r}_n$.

To make the seed-level optimization aware of the full real-location domain,
we aggregate the full-domain loss matrix onto the seed records. Let
$W\in \mathbb{R}^{N\times J}$ be a corner-weight matrix, where $W_{n,j}$
is the interpolation weight of seed record $\mathbf{x}_j$ for representing
real location $\mathbf{r}_n$. The approximated seed-level cost matrix is
$C^{\mathrm{approx}}=W^\top L$, i.e.,
\[
c^{\mathrm{approx}}_{j,k}
=
\sum_{n=1}^{N} W_{n,j} L_{n,k},
\qquad
j=1,\ldots,J,\quad k=1,\ldots,K.
\]
Thus, although the optimization variables are defined only on seed records,
the objective incorporates utility loss over the selected real locations.

For each seed record $\mathbf{x}_j$ and output $\mathbf{y}_k$, define
\[
z_{j,k}
=
\Pr[\mathcal{M}_{\mathrm{seed}}(\mathbf{x}_j)=\mathbf{y}_k].
\]
Let $Z=\{z_{j,k}\}_{j=1,k=1}^{J,K}$ denote the seed-level perturbation
matrix. We compute $Z$ by solving the following linear program:
\[
\begin{aligned}
\min_{Z}\quad
& \sum_{j=1}^{J}\sum_{k=1}^{K}
c^{\mathrm{approx}}_{j,k} z_{j,k} \\
\mathrm{s.t.}\quad
& \sum_{k=1}^{K} z_{j,k}=1,
&& \forall j,\\
& \delta \le z_{j,k}\le 1,
&& \forall j,k,\\
& z_{j,k}
\le
\exp\!\left(\bar{\epsilon}_{t} d_p(\mathbf{x}_i,\mathbf{x}_j)\right)
z_{i,k},
&& \forall (i,j)\in \mathcal{E}_{t},\ 
t\in\{1,2\},\ \forall k .
\end{aligned}
\]
Here, $\delta>0$ is a small constant used to avoid zero probabilities; we set
$\delta=10^{-6}$ in our implementation. The sets $\mathcal{E}_1$ and
$\mathcal{E}_2$ contain neighboring seed-record pairs along the first and
second coordinate directions, respectively. We use
$\bar{\epsilon}_1=\epsilon_1/2$ and
$\bar{\epsilon}_2=\epsilon_2/2$ as the directional privacy budgets in the
seed optimization step. If $\mathcal{E}_{t}$ is defined as an undirected
neighbor set, the above constraint is imposed in both directions.

Let $Z^\star$ be an optimal solution. The optimized seed mechanism is then
given by
\[
\mathcal{M}_{\mathrm{seed}}(\mathbf{y}_k\mid \mathbf{x}_j)
=
z^\star_{j,k}.
\]
This mechanism initializes the source cells in the tree-based extension
algorithm, which constructs perturbation distributions for finer records
through interpolation and normalization.

\section{Pseudo Code of Three-based Extension Algorithm}
\label{sec:app:tree-extension-pseudocode}

\begin{table}[h]
\centering
\caption{One-dimensional interpolation rules for different extension methods.}
\label{tab:1d_interpolation_methods}
\small 
\begin{tabular}{p{1.4cm} p{6.2cm}}
\toprule
\textbf{Method} & \textbf{One-dimensional interpolation rule} \\
\midrule
McShane--Whitney
&
For anchors $x^L,x^R$ and output $y$, define $
l_{\mathbf{y}}(\mathbf{x})=\sup_{\hat{\mathbf{x}}\in\{x^L,x^R\}}
\bigl(\log \mathcal{M}(\mathbf{y}\mid \hat{\mathbf{x}})-\epsilon_t |x-\hat{\mathbf{x}}|\bigr)$ \newline and $
u_{\mathbf{y}}(\mathbf{x})=\inf_{\hat{\mathbf{x}}\in\{x^L,x^R\}}
\bigl(\log \mathcal{M}(\mathbf{y}\mid \hat{\mathbf{x}})+\epsilon_t |x-\hat{\mathbf{x}}|\bigr)$. Then, choose the midpoint $\log f(\mathbf{y}\mid \mathbf{x})=\frac{l_{\mathbf{y}}(\mathbf{x}),u_{\mathbf{y}}(\mathbf{x})}{2}$. \\
\midrule

Log-convex
&
For $x^L_t<x^R_t$, let $\lambda_t(\mathbf{x})=\frac{x^R_t-x_t}{x^R_t-x^L_t}$.
Then $
\log f(\mathbf{y}\mid \mathbf{x})
=
\lambda_t(\mathbf{x})\log \mathcal{M}(\mathbf{y}\mid x^L)
+
\bigl(1-\lambda_t(\mathbf{x})\bigr)\log \mathcal{M}(\mathbf{y}\mid x^R)$. 
Equivalently, $
f(\mathbf{y}\mid \mathbf{x})=\mathcal{M}(\mathbf{y}\mid x^L)^{\lambda_t(\mathbf{x})}\mathcal{M}(\mathbf{y}\mid x^R)^{1-\lambda_t(\mathbf{x})}$. 
\\
\midrule

LP-based
&
Introduce variables $\mathcal{M}(\mathbf{y}\mid \mathbf{x})$ for refined one-dimensional points $x$, and solve
an LP subject to anchor inheritance, simplex constraints, and $
\mathcal{M}(\mathbf{y}\mid \mathbf{x})\le e^{\epsilon_t |x-\mathbf{x}'|}\mathcal{M}(\mathbf{y}\mid \mathbf{x}')$, $\forall \mathbf{x},\mathbf{x}',\ \forall \mathbf{y}\in\mathcal{Y}$, plus the symmetric reverse inequality. The interpolated values are given by the feasible or optimal LP solution. \\
\bottomrule
\end{tabular}
\end{table}

\begin{algorithm}[h]
\caption{BFS-based tree extension}
\label{alg:bfs_tree_extension}
\KwIn{Extension tree $\mathcal{G}=(\mathcal{V},\mathcal{E},\{\mathcal{X}_v\}_{v\in V})$; source mechanisms
$\{\mathcal{M}_v\}_{v\in V_0}$ initialized from $\mathcal{M}_{\mathrm{seed}}$}
\KwOut{Global mechanism $\mathcal{M}$ on $\mathcal{X}_{\mathrm{tar}}$}

Initialize a FIFO queue $Q$ with all source nodes in $V_0$\;

\While{$Q$ is not empty}{
    dequeue a node $v$\;

    \If{$\mathrm{Succ}(v)=\varnothing$}{
        continue\;
    }

    Apply the local extension operator
    \[
    \{\mathcal{M}_{v\to u}\}_{u\in\mathrm{Succ}(v)}
    \leftarrow
    \mathrm{Ext}_v(\mathcal{M}_v)
    \]
    by:
    \begin{enumerate}
        \item treating the parent-corner distributions as anchors;
        \item interpolating from parent corners to child-boundary vertices using
        log-convex interpolation (as described in Table~\ref{tab:1d_interpolation_methods});
        \item interpolating from child-boundary vertices to child-interior vertices using coordinate-wise interpolation;
        \item normalizing the pre-normalization values over $\mathcal{Y}$.
    \end{enumerate}

    \ForEach{$u\in\mathrm{Succ}(v)$}{
        assign $\mathcal{M}_u \leftarrow \mathcal{M}_{v\to u}$ on $\mathcal{X}_u$\;
        enqueue $u$\;
    }
}

Assemble the induced global mechanism $\mathcal{M}$ from the node mechanisms on
$\mathcal{X}_{\mathrm{tar}}$\;
\Return $\mathcal{M}$\;
\end{algorithm}

\section{Additional Experimental Results}
\label{sec:appendix:addexp}

This appendix reports additional experimental results on the NYC and London
road-map datasets. Tables~\ref{tab:ULNYC}--\ref{tab:seed_ablation_london} and Figs.~\ref{fig:tradeoff}(a)(b) follow the same evaluation protocol as the Rome results in the main text. 

Specifically, Tables~\ref{tab:ULNYC} and~\ref{tab:ULLondon} report utility
loss on NYC and London, respectively; Tables~\ref{tab:timeNYC} and~\ref{tab:timeLondon} report computation time;
and Tables~\ref{tab:violationNYC} and~\ref{tab:violationLondon} report empirical mDP violation rates. We further report the seed-density ablation results in Tables~\ref{tab:seed_ablation_nyc} and~\ref{tab:seed_ablation_london}, together with the corresponding utility--computation trade-off plots in Figs.~\ref{fig:tradeoff}(a)(b).

Overall, the results follow trends similar to those observed on Rome: the proposed tree-based extension variants improve utility over predefined and hybrid baselines, achieve better scalability than full optimization-based methods, and maintain negligible empirical mDP violations. These results further support the stability of the proposed utility--scalability--privacy trade-off across different urban road-network domains.

\begin{table*}[h]
\footnotesize 
\begin{tabular}{p{1.85cm} |  p{1.00cm} p{1.21cm} p{1.21cm}  p{1.21cm} | p{1.21cm} p{1.21cm} p{1.21cm} | p{1.21cm} p{1.21cm} p{1.21cm}} 
\toprule
\multicolumn{2}{c }{{\rev \# Participating Points}} & \multicolumn{3}{c }{225 (depth=2)} & \multicolumn{3}{c }{1849 (depth=3)}& \multicolumn{3}{c }{28561 (depth=4)}\\
\midrule
\multicolumn{2}{c }{Privacy budget (km$^{-1}$)} & $\epsilon = 0.5$& $\epsilon = 1.0$& $\epsilon = 1.5$ & $\epsilon = 0.5$& $\epsilon = 1.0$& $\epsilon = 1.5$& $\epsilon = 0.5$& $\epsilon = 1.0$& $\epsilon = 1.5$\\
\midrule
\multicolumn{11}{c }{NYC road map}\\
\hline
Pre-defined& EM & 12.35±0.85 & 10.21±0.86 & 9.63±0.83 & 11.98±0.82 & 9.83±0.83 & 9.25±0.80 & 11.80±0.80 & 9.65±0.81 & 9.06±0.78 \\
\cline{2-2}
Noise Distribution & Laplace  & 12.30±0.47 & 10.21±0.53 & 9.63±0.53 & 11.94±0.46 & 9.84±0.52 & 9.25±0.52 & 11.76±0.46 & 9.66±0.51 & 9.07±0.52 \\
\cline{1-2}
Hybrid  & EM+RMP & 12.26±0.81 & 10.14±0.83 & 9.59±0.81 & 11.89±0.78 & 9.77±0.80 & 9.22±0.78 & 11.71±0.76 & 9.59±0.78 & 9.03±0.77 \\
\cline{2-2}
Methods &  COPT  & ------------- & ------------- & ------------- & ------------- & ------------- & ------------- & ------------- & ------------- & ------------- \\
\cline{1-2}
Optimization &  LP  & ------------- & ------------- & ------------- & ------------- & ------------- & ------------- & ------------- & ------------- & ------------- \\
\cline{2-2}
Based & LP-A & 9.78±0.42 & 9.63±0.43 & 10.11±0.46 & 9.54±0.42 & 9.37±0.44 & 9.82±0.46 & 9.42±0.42 & 9.25±0.44 & 9.68±0.47 \\
\cline{2-2} 
Methods & PAnDA & ------------- & ------------- & ------------- & ------------- & ------------- & ------------- & ------------- & ------------- & ------------- \\
\hline
\cline{1-2}
\rowcolor{lightgray!20} Existing Extension &  AIPO  & 7.19±0.15 & 5.27±0.20 & 4.60±0.22 & 7.23±0.15 & 5.30±0.20 & 4.63±0.22 & 7.23±0.15 & 5.31±0.20 & 4.63±0.22 \\
\cline{2-2}
\rowcolor{lightgray!20} Based Methods & MW & 10.07±0.24 & 7.60±0.21 & 6.33±0.24 & 10.13±0.25 & 7.69±0.21 & 6.43±0.24 & 10.14±0.25 & 7.71±0.21 & 6.44±0.24 \\
\midrule
\rowcolor{lightgray!40} Tree-based &  MLaEt-A  & 7.19±0.15 & 5.27±0.20 & 4.60±0.22 & 7.23±0.15 & 5.30±0.20 & 4.63±0.22 & 7.23±0.15 & 5.31±0.20 & 4.63±0.22 \\
\cline{2-2}
\rowcolor{lightgray!40} extension & MLaEt-M & 7.08±0.18 & 5.15±0.24 & 4.47±0.27 & 7.11±0.18 & 5.18±0.24 & 4.50±0.27 & 7.11±0.17 & 5.19±0.23 & 4.50±0.27 \\
\rowcolor{lightgray!40} algorithms & MLaEt-O & 7.02±0.33 & 5.14±0.43 & 4.52±0.47 & 6.97±0.33 & 5.09±0.43 & 4.50±0.47 & 6.95±0.33 & 5.07±0.43 & 4.48±0.47 \\
\bottomrule
\end{tabular}
\centering
\caption{\centering Utility loss of different perturbation methods (NYC; {\rev kilometers}) \\ Mean$\pm$1.96$\times$standard error; for the algorithm without getting the results, we label its results by ``--------''.}
\label{tab:ULNYC}
\end{table*}

\begin{table*}[h]
\footnotesize 
\begin{tabular}{p{1.85cm} |  p{1.00cm} p{1.21cm} p{1.21cm}  p{1.21cm} | p{1.21cm} p{1.21cm} p{1.21cm} | p{1.21cm} p{1.21cm} p{1.21cm}} 
\toprule
\multicolumn{2}{c }{{\rev \# Participating Points}} & \multicolumn{3}{c }{285 (depth=2)} & \multicolumn{3}{c }{2365 (depth=3)}& \multicolumn{3}{c }{36673 (depth=4)}\\
\midrule
\multicolumn{2}{c }{Privacy budget (km$^{-1}$)} & $\epsilon = 0.5$& $\epsilon = 1.0$& $\epsilon = 1.5$ & $\epsilon = 0.5$& $\epsilon = 1.0$& $\epsilon = 1.5$& $\epsilon = 0.5$& $\epsilon = 1.0$& $\epsilon = 1.5$\\
\midrule
\multicolumn{11}{c }{London road map}\\
\hline
Pre-defined& EM & 12.34±0.72 & 10.05±0.66 & 9.50±0.64 & 11.91±0.70 & 9.62±0.65 & 9.07±0.63 & 11.70±0.69 & 9.41±0.65 & 8.86±0.62 \\
\cline{2-2}
Noise Distribution & Laplace  & 12.34±1.02 & 10.09±0.93 & 9.56±0.90 & 11.92±1.00 & 9.68±0.91 & 9.14±0.88 & 11.72±0.99 & 9.47±0.90 & 8.93±0.87 \\
\cline{1-2}
Hybrid  & EM+RMP & 12.24±0.67 & 10.02±0.65 & 9.49±0.64 & 11.82±0.66 & 9.60±0.64 & 9.06±0.63 & 11.62±0.66 & 9.39±0.64 & 8.85±0.6 \\
\cline{2-2}
Methods &  COPT  & ------------- & ------------- & ------------- & ------------- & ------------- & ------------- & ------------- & ------------- & ------------- \\
\cline{1-2}
Optimization &  LP  & ------------- & ------------- & ------------- & ------------- & ------------- & ------------- & ------------- & ------------- & ------------- \\
\cline{2-2}
Based & LP-A & 8.34±0.65 & 7.84±0.72 & 8.18±0.69 & 8.10±0.65 & 7.60±0.71 & 7.95±0.68 & 7.98±0.64 & 7.48±0.70 & 7.82±0.67 \\
\cline{2-2} 
Methods & PAnDA & ------------- & ------------- & ------------- & ------------- & ------------- & ------------- & ------------- & ------------- & ------------- \\
\hline
\cline{1-2}
\rowcolor{lightgray!20} Existing Extension &  AIPO  & 8.18±0.20 & 6.06±0.30 & 5.40±0.34 & 8.21±0.20 & 6.10±0.30 & 5.42±0.33 & 8.21±0.20 & 6.10±0.30 & 5.42±0.33 \\
\cline{2-2}
\rowcolor{lightgray!20} Based Methods & MW & 11.11±0.09 & 8.36±0.06 & 6.93±0.05 & 11.15±0.09 & 8.42±0.06 & 7.02±0.06 & 11.16±0.09 & 8.43±0.06 & 7.03±0.06 \\
\midrule
\rowcolor{lightgray!40} Tree-based &  MLaEt-A  & 8.18±0.20 & 6.06±0.30 & 5.40±0.34 & 8.21±0.20 & 6.10±0.30 & 5.42±0.33 & 8.21±0.20 & 6.10±0.30 & 5.42±0.33 \\
\cline{2-2}
\rowcolor{lightgray!40} extension & MLaEt-M & 8.11±0.21 & 5.97±0.29 & 5.26±0.32 & 8.13±0.21 & 5.99±0.29 & 5.28±0.32 & 8.14±0.21 & 6.00±0.29 & 5.29±0.32 \\
\rowcolor{lightgray!40} algorithms & MLaEt-O & 7.94±0.16 & 5.83±0.23 & 5.20±0.25 & 7.87±0.16 & 5.78±0.23 & 5.17±0.25 & 7.84±0.16 & 5.76±0.23 & 5.15±0.25 \\
\bottomrule
\end{tabular}
\centering
\caption{\centering Utility loss of different perturbation methods (London; {\rev kilometers})\\ Mean$\pm$1.96$\times$standard error; for the algorithm without getting the results, we label its results by ``--------''.}
\label{tab:ULLondon}
\end{table*}

\begin{table*}[h]
\footnotesize 
\begin{tabular}{p{1.85cm} |  p{1.00cm} p{1.21cm} p{1.21cm}  p{1.21cm} | p{1.21cm} p{1.21cm} p{1.21cm} | p{1.21cm} p{1.21cm} p{1.21cm}} 
\toprule
\multicolumn{2}{c }{{\rev \# Participating Points}} & \multicolumn{3}{c }{225 (depth=2)} & \multicolumn{3}{c }{1849 (depth=3)}& \multicolumn{3}{c }{28561 (depth=4)}\\
\midrule
\multicolumn{2}{c }{Privacy budget (km$^{-1}$)} & $\epsilon = 0.5$& $\epsilon = 1.0$& $\epsilon = 1.5$ & $\epsilon = 0.5$& $\epsilon = 1.0$& $\epsilon = 1.5$& $\epsilon = 0.5$& $\epsilon = 1.0$& $\epsilon = 1.5$\\
\midrule
\multicolumn{11}{c }{NYC road map}\\
\hline
Pre-defined& EM & 0.0004±0.0012 & 0.0001±0.0000 & 0.0001±0.0000 & 0.0002±0.0001 & 0.0002±0.0001 & 0.0002±0.0000 & 0.0007±0.0000 & 0.0007±0.0001 & 0.0007±0.0000 \\
\cline{2-2}
Noise Distribution & Laplace  & 0.0021±0.0087 & 0.0002±0.0002 & 0.0001±0.0001 & 0.0002±0.0001 & 0.0002±0.0003 & 0.0002±0.0002 & 0.0004±0.0002 & 0.0003±0.0000 & 0.0003±0.0001 \\
\cline{1-2}
Hybrid  & EM+RMP & 0.0006±0.0020 & 0.0003±0.0005 & 0.0002±0.0002 & 0.0004±0.0002 & 0.0004±0.0002 & 0.0003±0.0001 & 0.0011±0.0003 & 0.0011±0.0002 & 0.0011±0.0001 \\
\cline{2-2}
Methods &  COPT  & ------------- & ------------- & ------------- & ------------- & ------------- & ------------- & ------------- & ------------- & ------------- \\
\cline{1-2}
Optimization &  LP  & ------------- & ------------- & ------------- & ------------- & ------------- & ------------- & ------------- & ------------- & ------------- \\
\cline{2-2}
Based & LP-A & 11.06±3.88 & 6.78±0.78 & 7.96±1.46 & 11.06±3.88 & 6.78±0.78 & 7.97±1.46 & 11.06±3.88 & 6.79±0.78 & 7.97±1.46 \\
\cline{2-2} 
Methods & PAnDA & ------------- & ------------- & ------------- & ------------- & ------------- & ------------- & ------------- & ------------- & ------------- \\
\hline
\cline{1-2}
\rowcolor{lightgray!20} Existing Extension &  AIPO  & 6.23±0.00 & 6.11±0.01 & 6.05±0.02 & 6.57±0.00 & 6.45±0.01 & 6.38±0.02 & 54.13±0.26 & 54.00±0.24 & 53.82±0.23 \\
\cline{2-2}
\rowcolor{lightgray!20} Based Methods & MW & 6.22±0.01 & 6.09±0.04 & 6.02±0.03 & 6.26±0.01 & 6.13±0.04 & 6.07±0.03 & 10.62±0.02 & 10.29±0.24 & 10.17±0.29 \\
\midrule
\rowcolor{lightgray!40} Tree-based &  MLaEt-A  & 5.53±0.22 & 4.06±0.05 & 3.83±0.04 & 5.59±0.22 & 4.12±0.05 & 3.89±0.04 & 6.11±0.22 & 4.63±0.05 & 4.41±0.04 \\
\cline{2-2}
\rowcolor{lightgray!40} extension & MLaEt-M & 14.07±0.96 & 8.36±0.34 & 6.72±0.08 & 14.11±0.96 & 8.40±0.34 & 6.76±0.08 & 14.47±0.96 & 8.75±0.34 & 7.11±0.08 \\
\rowcolor{lightgray!40} algorithms & MLaEt-O & 6.18±0.29 & 4.74±0.07 & 4.49±0.05 & 14.29±0.30 & 12.71±0.39 & 7.64±0.95 & 42.98±1.87 & 38.29±1.96 & 33.90±2.30 \\
\bottomrule
\end{tabular}
\centering
\caption{\centering Computation time of different perturbation methods (NYC{\rev ; seconds}) \\ Mean$\pm$1.96$\times$standard error; for the algorithm without getting the results, we label its results by ``--------''.}
\label{tab:timeNYC}
\end{table*}

\begin{table*}[h]
\footnotesize 
\begin{tabular}{p{1.85cm} |  p{1.00cm} p{1.21cm} p{1.21cm}  p{1.21cm} | p{1.21cm} p{1.21cm} p{1.21cm} | p{1.21cm} p{1.21cm} p{1.21cm}} 
\toprule
\multicolumn{2}{c }{{\rev \# Participating Points}} & \multicolumn{3}{c }{285 (depth=2)} & \multicolumn{3}{c }{2365 (depth=3)}& \multicolumn{3}{c }{36673 (depth=4)}\\
\midrule
\multicolumn{2}{c }{Privacy budget (km$^{-1}$)} & $\epsilon = 0.5$& $\epsilon = 1.0$& $\epsilon = 1.5$ & $\epsilon = 0.5$& $\epsilon = 1.0$& $\epsilon = 1.5$& $\epsilon = 0.5$& $\epsilon = 1.0$& $\epsilon = 1.5$\\
\midrule
\multicolumn{11}{c }{London road map}\\
\hline
Pre-defined& EM & 0.0004±0.0012 & 0.0001±0.0000 & 0.0001±0.0000 & 0.0003±0.0001 & 0.0003±0.0001 & 0.0002±0.0001 & 0.0009±0.0000 & 0.0009±0.0000 & 0.0009±0.0000 \\
\cline{2-2}
Noise Distribution & Laplace  & 0.0001±0.0000 & 0.0001±0.0000 & 0.0001±0.0000 & 0.0002±0.0000 & 0.0002±0.0000 & 0.0001±0.0000 & 0.0004±0.0000 & 0.0003±0.0000 & 0.0003±0.0000 \\
\cline{1-2}
Hybrid  & EM+RMP & 0.0004±0.0012 & 0.0001±0.0000 & 0.0001±0.0000 & 0.0004±0.0001 & 0.0004±0.0001 & 0.0004±0.0001 & 0.0013±0.0001 & 0.0012±0.0000 & 0.0012±0.0000 \\
\cline{2-2}
Methods &  COPT  & ------------- & ------------- & ------------- & ------------- & ------------- & ------------- & ------------- & ------------- & ------------- \\
\cline{1-2}
Optimization &  LP  & ------------- & ------------- & ------------- & ------------- & ------------- & ------------- & ------------- & ------------- & ------------- \\
\cline{2-2}
Based & LP-A & 8.16±1.54 & 7.25±1.17 & 7.41±0.88 & 8.16±1.54 & 7.25±1.17 & 7.41±0.88 & 8.16±1.54 & 7.25±1.17 & 7.41±0.88 \\
\cline{2-2} 
Methods & PAnDA & ------------- & ------------- & ------------- & ------------- & ------------- & ------------- & ------------- & ------------- & ------------- \\
\hline
\cline{1-2}
\rowcolor{lightgray!20} Existing Extension &  AIPO  & 16.47±0.08 & 16.35±0.01 & 16.31±0.05 & 16.97±0.09 & 16.85±0.02 & 16.81±0.05 & 120.93±0.36 & 120.89±0.29 & 120.76±0.02 \\
\cline{2-2}
\rowcolor{lightgray!20} Based Methods & MW & 16.49±0.03 & 16.31±0.01 & 16.21±0.04 & 16.54±0.04 & 16.36±0.01 & 16.25±0.04 & 23.34±0.29 & 22.79±0.00 & 22.61±0.01 \\
\midrule
\rowcolor{lightgray!40} Tree-based &  MLaEt-A  & 8.22±0.11 & 7.52±0.05 & 7.24±0.04 & 8.28±0.11 & 7.57±0.05 & 7.30±0.04 & 8.80±0.11 & 8.09±0.05 & 7.82±0.04 \\
\cline{2-2}
\rowcolor{lightgray!40} extension & MLaEt-M & 14.51±0.40 & 11.89±0.14 & 11.18±0.08 & 14.55±0.40 & 11.93±0.14 & 11.22±0.08 & 14.91±0.41 & 12.29±0.14 & 11.57±0.08 \\
\rowcolor{lightgray!40} algorithms & MLaEt-O & 9.01±0.32 & 8.29±0.27 & 8.00±0.27 & 18.92±4.47 & 17.47±4.47 & 13.47±4.04 & 84.16±61.00 & 79.03±61.26 & 75.98±60.61 \\
\bottomrule
\end{tabular}
\centering
\caption{\centering Computation time of different perturbation methods (London{\rev ; seconds}) \\ Mean$\pm$1.96$\times$standard error; for the algorithm without getting the results, we label its results by ``--------''.}
\label{tab:timeLondon}
\end{table*}

\begin{table*}[h]
\footnotesize 
\begin{tabular}{p{1.85cm} |  p{1.00cm} p{1.21cm} p{1.21cm}  p{1.21cm} | p{1.21cm} p{1.21cm} p{1.21cm} | p{1.21cm} p{1.21cm} p{1.21cm}} 
\toprule
\multicolumn{2}{c }{{\rev \# Participating Points}} & \multicolumn{3}{c }{225 (depth=2)} & \multicolumn{3}{c }{1849 (depth=3)}& \multicolumn{3}{c }{28561 (depth=4)}\\
\midrule
\multicolumn{2}{c }{Privacy budget (km$^{-1}$)} & $\epsilon = 0.5$& $\epsilon = 1.0$& $\epsilon = 1.5$ & $\epsilon = 0.5$& $\epsilon = 1.0$& $\epsilon = 1.5$& $\epsilon = 0.5$& $\epsilon = 1.0$& $\epsilon = 1.5$\\
\midrule
\multicolumn{11}{c }{NYC road map}\\
\hline
Pre-defined& EM & 0.00±0.00 & 0.00±0.00 & 0.00±0.00 & 0.00±0.00 & 0.00±0.00 & 0.00±0.00 & 0.00±0.00 & 0.00±0.00 & 0.00±0.00 \\
\cline{2-2}
Noise Distribution & Laplace  & 0.00±0.00 & 0.00±0.00 & 0.00±0.00 & 0.00±0.00 & 0.00±0.00 & 0.00±0.00 & 0.00±0.00 & 0.00±0.00 & 0.00±0.00 \\
\cline{1-2}
Hybrid  & EM+RMP & 0.00±0.00 & 0.00±0.00 & 0.00±0.00 & 0.00±0.00 & 0.00±0.00 & 0.00±0.00 & 0.00±0.00 & 0.00±0.00 & 0.00±0.00 \\
\cline{2-2}
Methods &  COPT  & ------------- & ------------- & ------------- & ------------- & ------------- & ------------- & ------------- & ------------- & ------------- \\
\cline{1-2}
Optimization &  LP  & ------------- & ------------- & ------------- & ------------- & ------------- & ------------- & ------------- & ------------- & ------------- \\
\cline{2-2}
Based & LP-A & 0.0329±0.0024 & 0.0148±0.0015 & 0.0067±0.0005 & 0.0329±0.0024 & 0.0148±0.0015 & 0.0067±0.0005 & 0.0329±0.0024 & 0.0148±0.0015 & 0.0067±0.0005 \\
\cline{2-2} 
Methods & PAnDA & ------------- & ------------- & ------------- & ------------- & ------------- & ------------- & ------------- & ------------- & ------------- \\
\hline
\cline{1-2}
\rowcolor{lightgray!20} Existing Extension &  AIPO  & 0.00±0.00 & 0.00±0.00 & 0.00±0.00 & 0.00±0.00 & 0.00±0.00 & 0.00±0.00 & 0.00±0.00 & 0.00±0.00 & 0.00±0.00 \\
\cline{2-2}
\rowcolor{lightgray!20} Based Methods & MW & 0.00±0.00 & 0.00±0.00 & 0.00±0.00 & 0.00±0.00 & 0.00±0.00 & 0.00±0.00 & 0.00±0.00 & 0.00±0.00 & 0.00±0.00 \\
\midrule
\rowcolor{lightgray!40} Tree-based &  MLaEt-A  & 0.00±0.00 & 0.00±0.00 & 0.00±0.00 & 0.00±0.00 & 0.00±0.00 & 0.00±0.00 & 0.00±0.00 & 0.00±0.00 & 0.00±0.00 \\
\cline{2-2}
\rowcolor{lightgray!40} extension & MLaEt-M & 0.00±0.00 & 0.00±0.00 & 0.00±0.00 & 0.00±0.00 & 0.00±0.00 & 0.00±0.00 & 0.00±0.00 & 0.00±0.00 & 0.00±0.00 \\
\rowcolor{lightgray!40} algorithms & MLaEt-O & 0.00±0.00 & 0.00±0.00 & 0.00±0.00 & 0.00±0.00 & 0.00±0.00 & 0.00±0.00 & 0.00±0.00 & 0.00±0.00 & 0.00±0.00 \\
\bottomrule
\end{tabular}
\centering
\caption{\centering Violation ratio of different perturbation methods (NYC) \\ Mean$\pm$1.96$\times$standard error; for the algorithm without getting the results, we label its results by ``--------''.}
\label{tab:violationNYC}
\end{table*}

\begin{table*}[h]
\footnotesize 
\begin{tabular}{p{1.85cm} |  p{1.00cm} p{1.21cm} p{1.21cm}  p{1.21cm} | p{1.21cm} p{1.21cm} p{1.21cm} | p{1.21cm} p{1.21cm} p{1.21cm}} 
\toprule
\multicolumn{2}{c }{{\rev \# Participating Points}} & \multicolumn{3}{c }{285 (depth=2)} & \multicolumn{3}{c }{2365 (depth=3)}& \multicolumn{3}{c }{36673 (depth=4)}\\
\midrule
\multicolumn{2}{c }{Privacy budget (km$^{-1}$)} & $\epsilon = 0.5$& $\epsilon = 1.0$& $\epsilon = 1.5$ & $\epsilon = 0.5$& $\epsilon = 1.0$& $\epsilon = 1.5$& $\epsilon = 0.5$& $\epsilon = 1.0$& $\epsilon = 1.5$\\
\midrule
\multicolumn{11}{c }{London road map}\\
\hline
Pre-defined& EM & 0.00±0.00 & 0.00±0.00 & 0.00±0.00 & 0.00±0.00 & 0.00±0.00 & 0.00±0.00 & 0.00±0.00 & 0.00±0.00 & 0.00±0.00 \\
\cline{2-2}
Noise Distribution & Laplace  & 0.00±0.00 & 0.00±0.00 & 0.00±0.00 & 0.00±0.00 & 0.00±0.00 & 0.00±0.00 & 0.00±0.00 & 0.00±0.00 & 0.00±0.00 \\
\cline{1-2}
Hybrid  & EM+RMP & 0.00±0.00 & 0.00±0.00 & 0.00±0.00 & 0.00±0.00 & 0.00±0.00 & 0.00±0.00 & 0.00±0.00 & 0.00±0.00 & 0.00±0.00 \\
\cline{2-2}
Methods &  COPT  & ------------- & ------------- & ------------- & ------------- & ------------- & ------------- & ------------- & ------------- & ------------- \\
\cline{1-2}
Optimization &  LP  & ------------- & ------------- & ------------- & ------------- & ------------- & ------------- & ------------- & ------------- & ------------- \\
\cline{2-2}
Based & LP-A & 0.0314±0.0011 & 0.0108±0.0009 & 0.0051±0.0005 & 0.0314±0.0011 & 0.0108±0.0009 & 0.0051±0.0005 & 0.0314±0.0011 & 0.0108±0.0009 & 0.0051±0.0005 \\
\cline{2-2} 
Methods & PAnDA & ------------- & ------------- & ------------- & ------------- & ------------- & ------------- & ------------- & ------------- & ------------- \\
\hline
\cline{1-2}
\rowcolor{lightgray!20} Existing Extension &  AIPO  & 0.00±0.00 & 0.00±0.00 & 0.00±0.00 & 0.00±0.00 & 0.00±0.00 & 0.00±0.00 & 0.00±0.00 & 0.00±0.00 & 0.00±0.00 \\
\cline{2-2}
\rowcolor{lightgray!20} Based Methods & MW & 0.00±0.00 & 0.00±0.00 & 0.00±0.00 & 0.00±0.00 & 0.00±0.00 & 0.00±0.00 & 0.00±0.00 & 0.00±0.00 & 0.00±0.00 \\
\midrule
\rowcolor{lightgray!40} Tree-based &  MLaEt-A  & 0.00±0.00 & 0.00±0.00 & 0.00±0.00 & 0.00±0.00 & 0.00±0.00 & 0.00±0.00 & 0.00±0.00 & 0.00±0.00 & 0.00±0.00 \\
\cline{2-2}
\rowcolor{lightgray!40} extension & MLaEt-M & 0.00±0.00 & 0.00±0.00 & 0.00±0.00 & 0.00±0.00 & 0.00±0.00 & 0.00±0.00 & 0.00±0.00 & 0.00±0.00 & 0.00±0.00 \\
\rowcolor{lightgray!40} algorithms & MLaEt-O & 0.00±0.00 & 0.00±0.00 & 0.00±0.00 & 0.00±0.00 & 0.00±0.00 & 0.00±0.00 & 0.00±0.00 & 0.00±0.00 & 0.00±0.00 \\
\bottomrule
\end{tabular}
\centering
\caption{\centering Violation ratio of different perturbation methods (London) \\ Mean$\pm$1.96$\times$standard error; for the algorithm without getting the results, we label its results by ``--------''.}
\label{tab:violationLondon}
\end{table*}

\begin{table}[h]
\centering
\caption{Effect of the seed-grid density on the three \textsc{MLaEt} variants and the LP-A baseline  (NYC).}
\label{tab:seed_ablation_nyc}
\resizebox{\linewidth}{!}{
\begin{tabular}{lrrrrrrr}
\toprule
Grid dimensions & $5{\times}5$ & $7{\times}7$ & $9{\times}9$ & $11{\times}11$ & $13{\times}13$ & $15{\times}15$ & $17{\times}17$ \\
\midrule
\multicolumn{8}{l}{\textbf{\textsc{MLaEt}-A}} \\
Loss (km) & 7.2438 & 6.3518 & 5.9036 & 5.7689 & 5.6618 & 5.6782 & 5.5206 \\
Time (seconds) & 1.1699 & 1.6487 & 2.3931 & 3.7798 & 6.0345 & 10.2088 & 18.1457 \\
mDP violation ratio & 0.00 & 0.00 & 0.00 & 0.00 & 0.00 & 0.00 & 0.00 \\
\addlinespace
\hline
\multicolumn{8}{l}{\textbf{\textsc{MLaEt}-M}} \\
Loss (km) & 7.2119 & 6.3675 & 5.9779 & 5.8689 & 5.6003 & 5.5569 & 5.6087 \\
Time (seconds) & 1.1601 & 1.6087 & 2.3211 & 3.6137 & 5.8508 & 9.7754 & 16.7510 \\
mDP violation ratio & 0.00 & 0.00 & 0.00 & 0.00 & 0.00 & 0.00 & 0.00 \\
\addlinespace
\hline
\multicolumn{8}{l}{\textbf{\textsc{MLaEt}-O}} \\
Loss (km) & 6.6894 & 5.8911 & 5.5087 & 5.5263 & 5.3521 & 5.3064 & 5.3774 \\
Time (seconds) & 1.4979 & 2.3674 & 3.7641 & 6.0669 & 9.6859 & 15.4560 & 24.9720 \\
mDP violation ratio & 0.00 & 0.00 & 0.00 & 0.00 & 0.00 & 0.00 & 0.00 \\
\addlinespace
\hline
\multicolumn{8}{l}{LP-A} \\
Loss (km) & 11.4301 & 10.3333 & 10.0943 & 9.7873 & 9.6238 & 9.6901 & --- \\
Time (seconds) & 0.0704 & 0.3916 & 1.9015 & 7.5537 & 32.8578 & 186.8650 & $>1000$ \\
mDP violation ratio & 0.0256 & 0.0270 & 0.0231 & 0.0204 & 0.0170 & 0.0143 & --- \\
\bottomrule
\end{tabular}}
\end{table}

\begin{table}[h]
\centering
\caption{Effect of the seed-grid density on the three \textsc{MLaEt} variants and the LP-A baseline  (London).}
\label{tab:seed_ablation_london}
\resizebox{\linewidth}{!}{
\begin{tabular}{lrrrrrrr}
\toprule
Grid dimensions & $5{\times}5$ & $7{\times}7$ & $9{\times}9$ & $11{\times}11$ & $13{\times}13$ & $15{\times}15$ & $17{\times}17$ \\
\midrule
\multicolumn{8}{l}{\textbf{\textsc{MLaEt}-A}} \\
Loss (km) & 8.0517 & 6.9707 & 6.6874 & 6.4702 & 6.4357 & 6.5057 & 6.4617 \\
Time (seconds) & 2.7620 & 3.6173 & 4.8942 & 6.7520 & 9.3403 & 13.0815 & 18.3441 \\
mDP violation ratio & 0.00 & 0.00 & 0.00 & 0.00 & 0.00 & 0.00 & 0.00 \\
\addlinespace
\hline
\multicolumn{8}{l}{\textbf{\textsc{MLaEt}-M}} \\
Loss (km) & 7.8718 & 7.2003 & 6.8011 & 6.4812 & 6.5249 & 6.4922 & 6.5794 \\
Time (seconds) & 2.6505 & 3.5117 & 4.7380 & 6.5714 & 8.9594 & 12.5490 & 17.5100 \\
mDP violation ratio & 0.00 & 0.00 & 0.00 & 0.00 & 0.00 & 0.00 & 0.00 \\
\addlinespace
\hline
\multicolumn{8}{l}{\textbf{\textsc{MLaEt}-O}} \\
Loss (km) & 7.5021 & 6.4886 & 6.2865 & 6.4601 & 6.1996 & 6.1854 & 6.1105 \\
Time (seconds) & 2.9908 & 4.2673 & 6.1792 & 8.9965 & 12.9090 & 18.3510 & 25.9370 \\
mDP violation ratio & 0.00 & 0.00 & 0.00 & 0.00 & 0.00 & 0.00 & 0.00 \\
\addlinespace
\hline
\multicolumn{8}{l}{LP-A} \\
Loss (km) & 9.6187 & 8.6136 & 8.3316 & 8.1540 & 8.0537 & 8.0159 & 7.5744 \\
Time (seconds) & 0.0765 & 0.3046 & 1.4509 & 4.6706 & 14.0565 & 38.1636 & 81.3006 \\
mDP violation ratio & 0.0262 & 0.0251 & 0.0201 & 0.0182 & 0.0145 & 0.0129 & 0.0110 \\
\bottomrule
\end{tabular}}
\end{table}

\DEL{
\begin{table}[t]
\centering
\caption{Effect of seed-grid density on \textsc{MLaEt}-A (Rome).}
\label{tab:seed_ablation}
\resizebox{\linewidth}{!}{%
\begin{tabular}{lrrrrrrr}
\toprule
Metric
& $5{\times}5$
& $7{\times}7$
& $9{\times}9$
& $11{\times}11$
& $13{\times}13$
& $15{\times}15$
& $17{\times}17$ \\
\midrule
Loss (km)
& 7.2902 & 6.4534 & 5.9906 & 5.7265 & 5.5845 & 5.4542 & 5.6767 \\
Time (s)
& 0.9248 & 1.3292 & 2.0040 & 3.2189 & 5.0931 & 8.8058 & 14.1384 \\
mDP violation rate
& 0.00 & 0.00 & 0.00 & 0.00 & 0.00 & 0.00 & 0.00 \\
\bottomrule
\end{tabular}%
}
\end{table}

\begin{table}[t]
\centering
\caption{Effect of seed-grid density on \textsc{MLaEt}-A (NYC).}
\label{tab:seed_ablation_nyc}
\resizebox{\linewidth}{!}{%
\begin{tabular}{lrrrrrrr}
\toprule
Metric
& $5{\times}5$
& $7{\times}7$
& $9{\times}9$
& $11{\times}11$
& $13{\times}13$
& $15{\times}15$
& $17{\times}17$ \\
\midrule
Loss (km)
& 7.2438 & 6.3518 & 5.9036 & 5.7689 & 5.6618 & 5.6782 & 5.5206 \\
Time (s)
& 1.1699 & 1.6487 & 2.3931 & 3.7798 & 6.0345 & 10.2088 & 18.1457 \\
mDP violation rate
& 0.00 & 0.00 & 0.00 & 0.00 & 0.00 & 0.00 & 0.00 \\
\bottomrule
\end{tabular}%
}
\end{table}

\begin{table}[t]
\centering
\caption{Effect of seed-grid density on \textsc{MLaEt}-A (London).}
\label{tab:seed_ablation_london}
\resizebox{\linewidth}{!}{%
\begin{tabular}{lrrrrrrr}
\toprule
Metric
& $5{\times}5$
& $7{\times}7$
& $9{\times}9$
& $11{\times}11$
& $13{\times}13$
& $15{\times}15$
& $17{\times}17$ \\
\midrule
Loss (km)
& 8.0517 & 6.9707 & 6.6874 & 6.4702 & 6.4357 & 6.5057 & 6.4617 \\
Time (s)
& 2.7620 & 3.6173 & 4.8942 & 6.7520 & 9.3403 & 13.0815 & 18.3441 \\
mDP violation rate
& 0.00 & 0.00 & 0.00 & 0.00 & 0.00 & 0.00 & 0.00 \\
\bottomrule
\end{tabular}%
}
\end{table}}

\begin{figure}[t]
\centering
\hspace{0.00in}
\begin{minipage}{0.48\textwidth}
 \subfigure[New York]{
\includegraphics[width=0.48\textwidth]{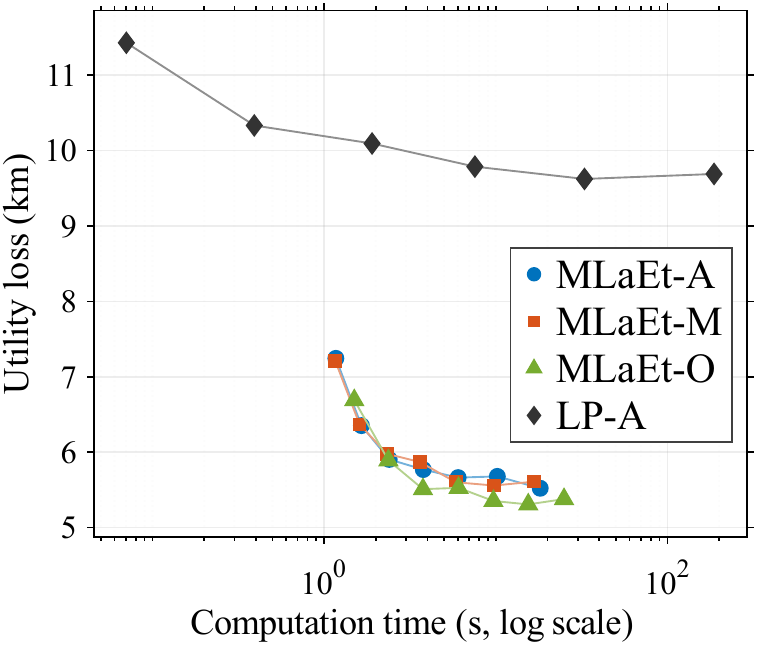}}
 \subfigure[London]{
\includegraphics[width=0.48\textwidth]{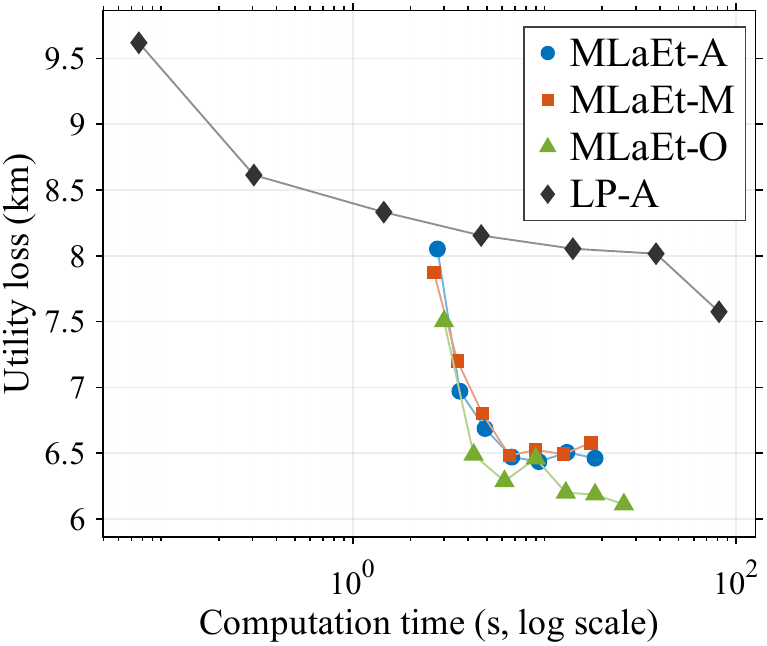}}
\end{minipage}
\caption{Tradeoff between utility and computation time.}
\label{fig:tradeoff}
\end{figure}


\begin{thebibliography}{24}


\ifx \showCODEN    \undefined \def \showCODEN     #1{\unskip}     \fi
\ifx \showISBNx    \undefined \def \showISBNx     #1{\unskip}     \fi
\ifx \showISBNxiii \undefined \def \showISBNxiii  #1{\unskip}     \fi
\ifx \showISSN     \undefined \def \showISSN      #1{\unskip}     \fi
\ifx \showLCCN     \undefined \def \showLCCN      #1{\unskip}     \fi
\ifx \shownote     \undefined \def \shownote      #1{#1}          \fi
\ifx \showarticletitle \undefined \def \showarticletitle #1{#1}   \fi
\ifx \showURL      \undefined \def \showURL       {\relax}        \fi
\providecommand\bibfield[2]{#2}
\providecommand\bibinfo[2]{#2}
\providecommand\natexlab[1]{#1}
\providecommand\showeprint[2][]{arXiv:#2}

\bibitem[ope(2020)]%
        {openstreetmap}
 \bibinfo{year}{2020}\natexlab{}.
\newblock \bibinfo{title}{openstreetmap}.
\newblock \bibinfo{howpublished}{\url{https://www.openstreetmap.org/}}.
\newblock
\shownote{Accessed: 2020-04-07}.
\newblock


\bibitem[Andr{\'e}s et~al\mbox{.}(2013)]%
        {Andres-CCS2013}
\bibfield{author}{\bibinfo{person}{Miguel~E Andr{\'e}s}, \bibinfo{person}{Nicol{\'a}s~E Bordenabe}, \bibinfo{person}{Konstantinos Chatzikokolakis}, {and} \bibinfo{person}{Catuscia Palamidessi}.} \bibinfo{year}{2013}\natexlab{}.
\newblock \showarticletitle{Geo-indistinguishability: Differential privacy for location-based systems}. In \bibinfo{booktitle}{\emph{Proceedings of the 2013 ACM SIGSAC conference on Computer \& communications security}}. \bibinfo{pages}{901--914}.
\newblock


\bibitem[Bordenabe et~al\mbox{.}(2014)]%
        {Bordenabe-CCS2014}
\bibfield{author}{\bibinfo{person}{N.~E. Bordenabe}, \bibinfo{person}{K. Chatzikokolakis}, {and} \bibinfo{person}{C. Palamidessi}.} \bibinfo{year}{2014}\natexlab{}.
\newblock \showarticletitle{Optimal Geo-Indistinguishable Mechanisms for Location Privacy}. In \bibinfo{booktitle}{\emph{Proc. of ACM CCS}}. \bibinfo{pages}{251--262}.
\newblock


\bibitem[Borgs et~al\mbox{.}(2018)]%
        {borgs2018extend}
\bibfield{author}{\bibinfo{person}{Christian Borgs}, \bibinfo{person}{Jennifer Chayes}, \bibinfo{person}{Adam Smith}, {and} \bibinfo{person}{Ilias Zadik}.} \bibinfo{year}{2018}\natexlab{}.
\newblock \bibinfo{title}{Private Algorithms Can Always Be Extended}.
\newblock
\showeprint[arxiv]{1810.12518}~[math.ST]
\urldef\tempurl%
\url{https://arxiv.org/abs/1810.12518}
\showURL{%
\tempurl}


\bibitem[Chatzikokolakis et~al\mbox{.}(2013)]%
        {Chatzikokolakis-PETS2013}
\bibfield{author}{\bibinfo{person}{Konstantinos Chatzikokolakis}, \bibinfo{person}{Miguel~E Andr{\'e}s}, \bibinfo{person}{Nicol{\'a}s~Emilio Bordenabe}, {and} \bibinfo{person}{Catuscia Palamidessi}.} \bibinfo{year}{2013}\natexlab{}.
\newblock \showarticletitle{Broadening the scope of differential privacy using metrics}. In \bibinfo{booktitle}{\emph{international symposium on privacy enhancing technologies symposium}}. Springer, \bibinfo{pages}{82--102}.
\newblock


\bibitem[Chatzikokolakis et~al\mbox{.}(2017)]%
        {chatzikokolakis2017efficient}
\bibfield{author}{\bibinfo{person}{Konstantinos Chatzikokolakis}, \bibinfo{person}{Ehab Elsalamouny}, {and} \bibinfo{person}{Catuscia Palamidessi}.} \bibinfo{year}{2017}\natexlab{}.
\newblock \showarticletitle{Efficient utility improvement for location privacy}.
\newblock \bibinfo{journal}{\emph{Proceedings on Privacy Enhancing Technologies}} (\bibinfo{year}{2017}).
\newblock


\bibitem[Chatzikokolakis et~al\mbox{.}(2015)]%
        {chatzikokolakis2015constructing}
\bibfield{author}{\bibinfo{person}{Konstantinos Chatzikokolakis}, \bibinfo{person}{Catuscia Palamidessi}, {and} \bibinfo{person}{Marco Stronati}.} \bibinfo{year}{2015}\natexlab{}.
\newblock \showarticletitle{Constructing elastic distinguishability metrics for location privacy}.
\newblock \bibinfo{journal}{\emph{Proceedings on Privacy Enhancing Technologies}} (\bibinfo{year}{2015}).
\newblock


\bibitem[Chen et~al\mbox{.}(2021)]%
        {chen2021perceptual}
\bibfield{author}{\bibinfo{person}{Jia-Wei Chen}, \bibinfo{person}{Li-Ju Chen}, \bibinfo{person}{Chia-Mu Yu}, {and} \bibinfo{person}{Chun-Shien Lu}.} \bibinfo{year}{2021}\natexlab{}.
\newblock \showarticletitle{Perceptual indistinguishability-net (pi-net): Facial image obfuscation with manipulable semantics}. In \bibinfo{booktitle}{\emph{Proceedings of the IEEE/CVF Conference on Computer Vision and Pattern Recognition}}. \bibinfo{pages}{6478--6487}.
\newblock


\bibitem[Duchi et~al\mbox{.}(2013)]%
        {Duchi-FOCS2013}
\bibfield{author}{\bibinfo{person}{John~C. Duchi}, \bibinfo{person}{Michael~I. Jordan}, {and} \bibinfo{person}{Martin~J. Wainwright}.} \bibinfo{year}{2013}\natexlab{}.
\newblock \showarticletitle{Local Privacy and Statistical Minimax Rates}. In \bibinfo{booktitle}{\emph{2013 IEEE 54th Annual Symposium on Foundations of Computer Science}}. \bibinfo{pages}{429--438}.
\newblock
\href{https://doi.org/10.1109/FOCS.2013.53}{doi:\nolinkurl{10.1109/FOCS.2013.53}}


\bibitem[Feyisetan and Kasiviswanathan(2021)]%
        {feyisetan2021private}
\bibfield{author}{\bibinfo{person}{Oluwaseyi Feyisetan} {and} \bibinfo{person}{Shiva Kasiviswanathan}.} \bibinfo{year}{2021}\natexlab{}.
\newblock \showarticletitle{Private release of text embedding vectors}. In \bibinfo{booktitle}{\emph{Proceedings of the First Workshop on Trustworthy Natural Language Processing}}. \bibinfo{pages}{15--27}.
\newblock


\bibitem[Hou et~al\mbox{.}(2026)]%
        {Hou-TDSC2026}
\bibfield{author}{\bibinfo{person}{Yuchao Hou}, \bibinfo{person}{Jiazhe Jiao}, \bibinfo{person}{Jie Wang}, \bibinfo{person}{Guangyin Jin}, \bibinfo{person}{Zijian Zhang}, \bibinfo{person}{Xiaoyu Xia}, \bibinfo{person}{Zhiquan Liu}, \bibinfo{person}{Minglu Li}, {and} \bibinfo{person}{Youliang Tian}.} \bibinfo{year}{2026}\natexlab{}.
\newblock \showarticletitle{Federated Analytics Assisted Semantic Alignment for Secure and Privacy-Preserving Image Classification}.
\newblock \bibinfo{journal}{\emph{IEEE Transactions on Dependable and Secure Computing}} (\bibinfo{year}{2026}), \bibinfo{pages}{1--16}.
\newblock
\href{https://doi.org/10.1109/TDSC.2026.3704515}{doi:\nolinkurl{10.1109/TDSC.2026.3704515}}


\bibitem[Imola et~al\mbox{.}(2022a)]%
        {ImolaUAI2022}
\bibfield{author}{\bibinfo{person}{Jacob Imola}, \bibinfo{person}{Shiva Kasiviswanathan}, \bibinfo{person}{Stephen White}, \bibinfo{person}{Abhinav Aggarwal}, {and} \bibinfo{person}{Nathanael Teissier}.} \bibinfo{year}{2022}\natexlab{a}.
\newblock \showarticletitle{Balancing utility and scalability in metric differential privacy}. In \bibinfo{booktitle}{\emph{Proc. of UAI 2022}}.
\newblock


\bibitem[Imola et~al\mbox{.}(2022b)]%
        {imola2022balancing}
\bibfield{author}{\bibinfo{person}{Jacob Imola}, \bibinfo{person}{Shiva Kasiviswanathan}, \bibinfo{person}{Stephen White}, \bibinfo{person}{Abhinav Aggarwal}, {and} \bibinfo{person}{Nathanael Teissier}.} \bibinfo{year}{2022}\natexlab{b}.
\newblock \showarticletitle{Balancing utility and scalability in metric differential privacy}. In \bibinfo{booktitle}{\emph{Uncertainty in Artificial Intelligence}}. PMLR, \bibinfo{pages}{885--894}.
\newblock


\bibitem[Liu and Qiu(2025)]%
        {Liu-CCS2025}
\bibfield{author}{\bibinfo{person}{Ruiyao Liu} {and} \bibinfo{person}{Chenxi Qiu}.} \bibinfo{year}{2025}\natexlab{}.
\newblock \showarticletitle{PAnDA: Rethinking Metric Differential Privacy Optimization at Scale with Anchor-Based Approximation}. In \bibinfo{booktitle}{\emph{Proceedings of The 32nd ACM Conference on Computer and Communications Security (CCS)}}.
\newblock


\bibitem[Qiu(2024)]%
        {qiu-IJCAI2024}
\bibfield{author}{\bibinfo{person}{Chenxi Qiu}.} \bibinfo{year}{2024}\natexlab{}.
\newblock \showarticletitle{Enhancing scalability of metric differential privacy via secret dataset partitioning and benders decomposition}. In \bibinfo{booktitle}{\emph{Proceedings of the Thirty-Third International Joint Conference on Artificial Intelligence}}. \bibinfo{pages}{1944--1952}.
\newblock


\bibitem[Qiu(2026)]%
        {Qiu-USec2026}
\bibfield{author}{\bibinfo{person}{Chenxi Qiu}.} \bibinfo{year}{2026}\natexlab{}.
\newblock \showarticletitle{Interpolation-Based Optimization for Enforcing lp-Norm Metric Differential Privacy in Continuous and Fine-Grained Domains}. In \bibinfo{booktitle}{\emph{Proc. of USENIX 2026}}.
\newblock


\bibitem[Qiu et~al\mbox{.}(2025)]%
        {qiu2025time}
\bibfield{author}{\bibinfo{person}{Chenxi Qiu}, \bibinfo{person}{Ruiyao Liu}, \bibinfo{person}{Primal Pappachan}, \bibinfo{person}{Anna Squicciarini}, {and} \bibinfo{person}{Xinpeng Xie}.} \bibinfo{year}{2025}\natexlab{}.
\newblock \showarticletitle{Time-Efficient Locally Relevant Geo-Location Privacy Protection}.
\newblock \bibinfo{journal}{\emph{Proceedings on Privacy Enhancing Technologies}} (\bibinfo{year}{2025}).
\newblock


\bibitem[{Qiu} et~al\mbox{.}(2022)]%
        {Qiu-TMC2022}
\bibfield{author}{\bibinfo{person}{C. {Qiu}}, \bibinfo{person}{A.~C. {Squicciarini}}, \bibinfo{person}{C. {Pang}}, \bibinfo{person}{N. {Wang}}, {and} \bibinfo{person}{B. {Wu}}.} \bibinfo{year}{2022}\natexlab{}.
\newblock \showarticletitle{Location Privacy Protection in Vehicle-Based Spatial Crowdsourcing via Geo-Indistinguishability}.
\newblock \bibinfo{journal}{\emph{IEEE Transactions on Mobile Computing}} (\bibinfo{year}{2022}), \bibinfo{pages}{1--1}.
\newblock
\href{https://doi.org/10.1109/TMC.2020.3037911}{doi:\nolinkurl{10.1109/TMC.2020.3037911}}


\bibitem[Qiu et~al\mbox{.}(2024)]%
        {Qiu-EDBT2024}
\bibfield{author}{\bibinfo{person}{Chenxi Qiu}, \bibinfo{person}{Sourabh Yadav}, \bibinfo{person}{Yuede Ji}, \bibinfo{person}{Anna Squicciarini}, \bibinfo{person}{Ramanamurthy Dantu}, \bibinfo{person}{Juanjuan Zhao}, {and} \bibinfo{person}{Chengzhong Xu}.} \bibinfo{year}{2024}\natexlab{}.
\newblock \showarticletitle{Fine-Grained Geo-Obfuscation to Protect Workers' Location Privacy in Time-Sensitive Spatial Crowdsourcing}. In \bibinfo{booktitle}{\emph{Proc. of EDBT}}.
\newblock


\bibitem[Qiu~et al.(2025)]%
        {Qiu-PETS2025}
\bibfield{author}{\bibinfo{person}{C. Qiu~et al.}} \bibinfo{year}{2025}\natexlab{}.
\newblock \showarticletitle{Scalable Optimization for Locally Relevant Geo-Location Privacy}. In \bibinfo{booktitle}{\emph{PETS}}.
\newblock


\bibitem[Samet(1984)]%
        {Samet-CSurvey1984}
\bibfield{author}{\bibinfo{person}{Hanan Samet}.} \bibinfo{year}{1984}\natexlab{}.
\newblock \showarticletitle{The Quadtree and Related Hierarchical Data Structures}.
\newblock \bibinfo{journal}{\emph{ACM Comput. Surv.}} \bibinfo{volume}{16}, \bibinfo{number}{2} (\bibinfo{date}{June} \bibinfo{year}{1984}), \bibinfo{pages}{187–260}.
\newblock
\showISSN{0360-0300}
\href{https://doi.org/10.1145/356924.356930}{doi:\nolinkurl{10.1145/356924.356930}}


\bibitem[Wang et~al\mbox{.}(2017)]%
        {Wang-WWW2017}
\bibfield{author}{\bibinfo{person}{L. Wang}, \bibinfo{person}{D. Yang}, \bibinfo{person}{X. Han}, \bibinfo{person}{T. Wang}, \bibinfo{person}{D. Zhang}, {and} \bibinfo{person}{X. Ma}.} \bibinfo{year}{2017}\natexlab{}.
\newblock \showarticletitle{Location Privacy-Preserving Task Allocation for Mobile Crowdsensing with Differential Geo-Obfuscation}. In \bibinfo{booktitle}{\emph{Proc. of ACM WWW}}. \bibinfo{pages}{627--636}.
\newblock


\bibitem[Wang et~al\mbox{.}(2016)]%
        {Wang-ICDM2016}
\bibfield{author}{\bibinfo{person}{L. Wang}, \bibinfo{person}{D. Zhang}, \bibinfo{person}{D. Yang}, \bibinfo{person}{B. Lim}, {and} \bibinfo{person}{X. Ma}.} \bibinfo{year}{2016}\natexlab{}.
\newblock \showarticletitle{Differential Location Privacy for Sparse Mobile Crowdsensing}. \bibinfo{pages}{1257--1262}.
\newblock
\href{https://doi.org/10.1109/ICDM.2016.0169}{doi:\nolinkurl{10.1109/ICDM.2016.0169}}


\bibitem[Yu et~al\mbox{.}(2017)]%
        {Yu-NDSS2017}
\bibfield{author}{\bibinfo{person}{Lei Yu}, \bibinfo{person}{Ling Liu}, {and} \bibinfo{person}{Calton Pu}.} \bibinfo{year}{2017}\natexlab{}.
\newblock \showarticletitle{Dynamic Differential Location Privacy with Personalized Error Bounds.}. In \bibinfo{booktitle}{\emph{NDSS}}, Vol.~\bibinfo{volume}{17}. \bibinfo{pages}{1--15}.
\newblock


\end{thebibliography}
\end{document}